\documentclass[11pt]{article}
\usepackage{jheppub} 
\usepackage{lineno}

\usepackage{amsthm}
\usepackage{amssymb}

\newtheorem{prop}{Proposition}

\makeatletter

\def\nn{\nonumber}

\def\be{\begin{equation}}
\def\ee{\end{equation}}

\global\long\def\mA{\mathcal{A}}%
 
\global\long\def\mB{\mathcal{B}}%

\global\long\def\mD{\mathcal{D}}%

\global\long\def\mH{\mathcal{H}}%
 
\global\long\def\mI{\mathcal{I}}%

\global\long\def\mK{\mathcal{K}}%

\global\long\def\mM{\mathcal{M}}%
 
\global\long\def\mN{\mathcal{N}}%
 
\global\long\def\mO{\mathcal{O}}%

\global\long\def\mZ{\mathcal{Z}}%

\global\long\def\e{\epsilon}%
 
\global\long\def\ra{\rightarrow}%

\global\long\def\avg#1{\left\langle #1\right\rangle }%

\global\long\def\f#1#2{\frac{#1}{#2}}%
 
\global\long\def\del{\partial}%
 
\global\long\def\t{\theta}%
 
\global\long\def\a{\alpha}%
 
\global\long\def\b{\beta}%
 
\global\long\def\g{\gamma}%
 
\global\long\def\G{\Gamma}%

\global\long\def\r{\rho}%
 
\global\long\def\d{\delta}%
 
\global\long\def\Tr{\text{Tr}}%
 
\global\long\def\tr{\text{tr}}%
 
\global\long\def\ket#1{\left\langle #1\right|}%
 
\global\long\def\bra#1{\left|#1\right\rangle }%

\global\long\def\I{\mathbb{I}}%

\global\long\def\R{\mathbb{R}}%
 
\global\long\def\C{\mathbb{C}}%

\global\long\def\p{\varphi}%

\global\long\def\D{\Delta}%
 
\global\long\def\O{\Omega}%
 
\global\long\def\S{\Sigma}%

\global\long\def\l{\ell}%

\global\long\def\app{\approx}%
\global\long\def\sgn{\text{sgn}}%
\global\long\def\arcsinh{\text{arcsinh}}%
\global\long\def\arccosh{\text{arccosh}}%
\global\long\def\vp{\varphi}%

\allowdisplaybreaks

\@ifundefined{showcaptionsetup}{}{%
 \PassOptionsToPackage{caption=false}{subfig}}
\usepackage{subfig}
\makeatother

\title{Modular flow in JT gravity and entanglement wedge reconstruction}

\author{Ping Gao}

\affiliation{Department of Physics and NHETC,\\ Rutgers University, Piscataway, NJ 08854, USA}

\emailAdd{ping.gao@rutgers.edu}

\abstract{It has been shown in recent works that JT gravity with matter with two boundaries has a type II$_\infty$ algebra on each side. As the bulk spacetime between the two boundaries fluctuates in quantum nature, we can only define the entanglement wedge for each side in a pure algebraic sense. As we take the semiclassical limit, we will have a fixed long wormhole spacetime for a generic partially entangled thermal state (PETS), which is prepared by inserting heavy operators on the Euclidean path integral. Under this limit, with appropriate assumptions of the matter theory, geometric notions of the causal wedge and entanglement wedge emerge in this background. In particular, the causal wedge is manifestly nested in the entanglement wedge. Different PETS are orthogonal to each other, and thus the Hilbert space has a direct sum structure over sub-Hilbert spaces labeled by different Euclidean geometries. The full algebra for both sides is decomposed accordingly. From the algebra viewpoint, the causal wedge is dual to an emergent type III$_1$ subalgebra, which is generated by boundary light operators. To reconstruct the entanglement wedge, we consider the modular flow in a generic PETS for each boundary. We show that the modular flow acts locally and is the boost transformation around the global RT surface in the semiclassical limit. It follows that we can extend the causal wedge algebra to a larger type III$_1$ algebra corresponding to the entanglement wedge. Within each sub-Hilbert space, the original type II$_\infty$ reduces to type III$_1$.}

\begin{document}

\maketitle

\section{Introduction}

Jackiw-Teitelboim (JT) gravity \cite{Jackiw:1984je,Teitelboim:1983ux} with negative cosmological constant is a two-dimensional toy model for the near horizon physics of a near-extremal black hole. In recent years, it has be heavily studied, for example in \cite{Almheiri:2014cka,Kitaev:2018wpr,Suh:2020lco,Harlow:2018tqv,Yang:2018gdb,Maldacena:2016upp} (see also \cite{Mertens:2022irh} for a review). It has a dilaton field, which can be understood as the area of transverse $S^2$ in the near-extremal black hole through dimensional reduction. It plays the role of a Lagrange multiplier that restricts the curvature to be $R=-2$, which corresponds to AdS$_2$ spacetime. JT gravity can be quantized and is exactly solvable. We can add matters to JT gravity and study how quantum gravity interacts with matters. These matters could be non-perturbative ones, e.g. end-of-world branes \cite{Gao:2021uro,Goel:2020yxl} and conical singularities \cite{Witten:2020wvy}, or a generic bulk local quantum field theory \cite{Maldacena:2016upp,Mertens:2017mtv}. In this paper, we will focus on the latter, which has minimal coupling to JT gravity in the sense that it does not couple to the dilaton field. This leads to a large simplification since the background spacetime is still AdS$_2$. 

In JT gravity, the graviton degrees of freedom only live on the boundary and behave as the parameterization mode with its Schwarizian derivative as the effective action \cite{Maldacena:2016upp}. We can define the boundary operators for the bulk matter fields by pushing them close to the boundary by the AdS/CFT dictionary. In this sense, we assume these boundary operators form a 0+1 dimensional CFT. The interaction between matters and graviton is through the coupling between the boundary operator and parameterization modes. To fully characterize this theory, we need to understand the algebraic structure of these boundary operators. Each boundary algebra $\mA_{L/R}$ is generated by the boundary matter operators $\mO_{i,L/R}$ and the boundary Hamiltonian $H_{L/R}$. It has been shown in \cite{Penington:2023dql,Kolchmeyer:2023gwa} that for the JT gravity with matter with two boundaries and disk topology, all bounded operators on the left and right boundary form two type II$_\infty$ von Neumann algebras $\mA_{L/R}$, which are commutant to each other and are factors. In plain language, this means that these two algebras commute with each other and the only bounded operator in this theory commutes with both of them is proportional to identity. It follows that the full algebra of bounded operators in the theory is generated by the union of $\mA_L$ and $\mA_R$.

This analysis is at a fully quantum level and phrased in a pure boundary algebra language. For pure JT gravity with two boundaries, the phase space is two-dimensional and the canonical coordinate is the regularized geodesic length $\l$ connecting the left and right boundaries. After canonical quantization, the Hilbert space is $L^2(\R)$, namely all square normalizable Wheeler-DeWitt wavefunctions of $\l$ \cite{Harlow:2018tqv}. As the gravity states are all prepared by Euclidean path integral connecting two boundaries, this $L^2(\R)$ can be equivalently labeled by Hartle-Hawking states $\bra{\b}$, where $\b$ is the length of the Euclidean path. Together with the matters, the full Hilbert space is $\mH=\mH_{matter}\otimes L^2(\R)$. The states in this Hilbert space is generated by acting $\mA_L,\mA_R$ on $\bra{\O}\bra{0}$, where $\bra{\O}$ is the vacuum of matter Hilbert space and $\bra{0}$ is for $\b=0$ state.

In the context of AdS/CFT, an important question is to ask which bulk subregion is dual to a boundary subregion. It is well-known that for empty AdS, a boundary spatial region (and its causal diamond) is dual to the bulk causal region bounded by RT surface \cite{Ryu:2006bv}. This formula was later generalized to HRT surface \cite{Hubeny:2007xt} to include non-static spacetime and to the quantum extremal surface (QES) \cite{Engelhardt:2014gca} when subleading in $1/N$ corrections are included. The bulk subregion dual to the boundary subregion is called the entanglement wedge \cite{Czech:2012bh,Wall:2012uf,Headrick:2014cta}. However, all these proposals are made in a semiclassical sense, which relies on a fixed spacetime background, in which these geometric quantities are well defined. On the other hand, for JT gravity with matter, the spacetime is fluctuating (as $\l$ is not a fixed number), subregions of the bulk are not well-defined and the dilaton profile is also fluctuating. Therefore, there is no bulk notion of any of these surfaces and we can only define the entanglement wedge for the left and right boundaries in a pure algebraic way \cite{Kolchmeyer:2023gwa}.\footnote{Note that there is no boundary subregions but just the left and right boundaries in 0+1 dimension.}

As we take the semiclassical limit in the theory of JT gravity with matter, the spacetime will be dominated by a saddle and all classical geometric background and notions therein become emergent. The parameter that controls the semiclassical limit is $\phi_b$, the renormalized boundary value of the dilaton field. It is the effective $1/G_N$ in JT gravity and the semiclassical limit corresponds to $\phi_b\ra\infty$. The states with multiple heavy operators with weight $\mu\sim O(\phi_b)$ inserted on the Euclidean path are called partially entangled thermal states (PETS), which were first studied in \cite{Goel:2018ubv} (see \eqref{eq:pets}). As $\phi_b\ra\infty$, a generic PETS is dual to a long wormhole spacetime connecting the left and right boundaries (see Figure \ref{fig:Geometry-of-boundary.}). In this emergent classical spacetime, we will have an emergent causal wedge, which is the bulk region as the overlap of the causal past and causal future for the left or right boundary. Moreover, by the RT formula, the entanglement wedge for a boundary is geometrically given by the bulk causal diamond between the global minimal dilaton point (recall that dilaton represents the area of transverse $S^2$ in higher dimensions) and the boundary. Due to the existence of the long wormhole, the causal wedge is nested inside the entanglement wedge \cite{Headrick:2014cta, Wall:2012uf} for a generic PETS, and the gap between them is the causal shadow.

As $\phi_b\ra \infty$, we will assume that the heavy operators become generalized free fields and decouple from light operators, whose conformal weight $\D\sim O(1)$ does not grow with $\phi_b$. This is physically reasonable because we should not expect tunneling or decay of spacetime in the strict $G_N\ra 0$ limit.\footnote{For example, a black hole with a Hawking quanta decay involves a coupling between heavy and light fields, but it is a $O(G_N)$ effect.}  It follows that different generic PETS are orthogonal to each other in the semiclassical limit. Therefore, given a generic PETS $\bra{\Psi}$ and the corresponding Euclidean geometry $g$, we can build a sub-Hilbert space $\mH_g$ spanned by inserting light operators on the Euclidean path defining $\bra{\Psi}$ but keeping the heavy operators therein unchanged. We can analytically continue the Euclidean geometry $g$ to Lorentzian spacetime and identify the $Z_2$ time-reversal symmetric geodesic connecting two boundaries as the global Cauchy slice $\S_g$. The states in $\mH_g$ prepared by Euclidean path integral are equivalent to the states of the dual bulk QFT on $\S_g$. Therefore, the full Hilbert space in the semiclassical limit contains a direct sum structure
\be 
\mH\supset \oplus \mH_g \label{eq:1.1}
\ee 
where each sub-Hilbert space $\mH_g$ is labeled by a different Euclidean geometry $g$. It is not an equal sign because non-generic PETS may have subtleties in their semiclassical limit. 

As pointed out in \cite{Leutheusser:2022bgi}, the subregion for the causal wedge should be dual to a subalgebra on the boundary. In particular, the bulk fields in the causal wedge are an ordinary QFT in a Rindler spacetime, which should be a type III$_1$ von Neumann algebra \cite{Leutheusser:2021frk}. By the standard HKLL reconstruction \cite{Hamilton:2005ju,Hamilton:2006az,Hamilton:2006fh,Hamilton:2007wj}, these bulk fields are dual to
 operators living on the boundary. In our notation, these operators are light operators with conformal weights $\D\sim O(1)$ and survive in the semiclassical limit. In contrast, the operators with conformal weight $\mu\sim O(\phi_b)$ does not survive in the semiclassical limit because they have divergent matrix elements. However, these operators still exist in a PETS and backreact to gravity to form a long wormhole. The subregion-subalgebra duality \cite{Leutheusser:2022bgi} implies that these light operators on each boundary form a type III$_1$ subalgebra $\mA_{L/R}^0$. These two subalgebra $\mA_{L/R}^0$ are universal and independent on the details of $\bra{\Psi}$, which is consistent with the fact that the causal wedge is universally defined for any semiclassical bulk state.

Since there is a causal shadow, $\mA_{L/R}^0$ does not generate the full bulk algebra by acting on a generic PETS $\bra{\Psi}$. An immediate question follows: How do we reconstruct the causal shadow from the boundary algebra? In other words, there must be a larger subalgebra from $\mA_{L/R}$ surviving in the semiclassical limit. To characterize this subalgebra by extending $\mA_{L/R}^0$ is equivalent to an entanglement wedge reconstruction \cite{Dong:2016eik,Headrick:2014cta,Wall:2012uf,Czech:2012bh,Cotler:2017erl,Chen:2019gbt,Jafferis:2015del,Faulkner:2017vdd}. It has been proposed in \cite{Jafferis:2015del,Faulkner:2017vdd}
 that for a given bulk state $\bra{\Psi}$, the modular flow from a boundary subregion will be able to reconstruct its entanglement wedge. Unfortunately, the modular flow is usually quite complicated and acts nonlocally. It is hard to verify this proposal unless in limited cases, let alone to prove it in general. 
 
 In this paper, we will study the modular flow in the theory of JT gravity with matter in the semiclassical limit and will show that the modular flow of one boundary for a generic PETS acts locally as the boost transformation around the global minimal dilaton point, i.e. the RT surface (see Figure \ref{fig:5}). It follows that the modular flow extends the bulk fields in the causal wedge to the whole entanglement wedge for that boundary by bulk locality. The argument for this extension sufficiently covering the whole entanglement wedge can summarized as follows:
 \begin{enumerate}
\item For a generic PETS, only heavy operators backreact to the geometry in the semiclassical limit. The solution has smooth AdS$_2$ metric and non-smooth but continuous dilaton profile due to the backreaction. Since the light operators decouple from both the heavy operators and the dilaton field, they are dual to a bulk QFT living on a smooth AdS$_2$ background and enjoy an $SL(2)$ symmetry.\footnote{Note that $SL(2)$ is an isometry of the metric and also a symmetry of light fields but not the symmetry of a PETS because the dilaton profile is not invariant under $SL(2)$. However, as the dilaton and heavy operators are not part of the algebra in the semiclassical limit, this does not matter for us to utilize the $SL(2)$ for light operators.}
\item Using the replica trick, we show that the modular flow acts locally as an $SL(2)$ transformation for any boundary light operators. This $SL(2)$ transformation is the boost around the global RT surface, namely the globally minimal dilaton point.
\item Using HKLL reconstruction, we can extend the modular flow to any bulk field living in the causal wedge. As the bulk-boundary correlations are $SL(2)$ invariant functions, the modular flow acts locally on the bulk field by the same boost transformation.
\item Since the $SL(2)$ transformation is uniquely and globally defined for AdS$_2$, we can smoothly extended the modular flow parameter in the HKLL reconstruction to any real number to move the bulk field beyond the causal wedge. As the boost transformation can move a bulk point arbitrarily close to the edge of the entanglement wedge.\footnote{To be more precise, consider a bulk field in the causal wedge being flowed to far future and far past. Then the causal diamond defined by these two points is arbitrarily close to the horizon of entanglement wedge.} By bulk causality we can justify that the union of all modular flowed operators generate the whole entanglement wedge.
\end{enumerate}

 From the algebra viewpoint, for a generic PETS, the modular flow extends $\mA_{L/R}^0$ to a larger subalgebra $\mA_{L/R}^{EW_g}$, which is again type III$_1$ by subalgebra-subregion duality \cite{Leutheusser:2022bgi}. The full algebra of bounded operators in the sub-Hilbert space $\mH_g$ is generated by the union of $\mA_{L}^{EW_g}$ and $\mA_{R}^{EW_g}$. Similar to \eqref{eq:1.1}, for the algebra of all bounded operators we have
 \be 
 \mB(\mH)\supset \oplus_g \mB(\mH_g) \simeq \oplus_g \{\mA^{EW_g}_L,\mA^{EW_g}_R\},\quad \mA^{EW_g}_{L/R}\simeq \left(\mA^0_{L/R}\right)_\text{modular extension in $g$}
 \ee
 This is a nontrivial example of solvable modular flow and justifies the entanglement wedge reconstruction.

The organization of this paper is as follows. In Section \ref{sec:2.1}, we quickly review the von Neumann algebra in JT gravity with matter; in Section \ref{sec:2.2} we discuss the semiclassical limit in the Euclidean path integral formalism; subsequently, we show an equivalent geometric formalism by gluing EAdS$_2$ disks with $SL(2)$ charge conservation in the semiclassical limit in Section \ref{sec:2.3}; the structure of the semiclassical Hilbert space \eqref{eq:1.1} is discussed in Section \ref{sec:2.4}; in Section \ref{eq:2.5} we discuss the operator algebra $\mA_{L/R}^0$ in the semiclassical limit. To reconstruct the entanglement wedge, we study the modular flow by replica trick in Section \ref{sec:3.1} and show it is the boost transformation around the global RT surface; following this result, we discuss the extension of operator algebra by the modular flow in Section \ref{sec:3.2}. We summarize the conclusion and discuss a few problems and future directions in Section \ref{sec:4}. In Appendix \ref{app:1}, we provide a new and simple derivation of $sl(2)$ 6j symbol directly from the Euclidean path integral of JT gravity with matter. In Appendix \ref{app:2}, we use this Euclidean path integral to show that the configuration with crossing Wick contractions of heavy operators of the same type is exponentially suppressed relative to a non-crossing Wick contraction configuration in the semiclassical limit. This is a crucial step to simplify the computation of the modular flow in the replica trick. In Appendix \ref{app:3}, we solve a nontrivial example of non-generic PETS to address some subtleties of the structure of Hilbert space discussed in Section \ref{sec:2.4}.

\section{JT gravity with matter in the semiclassical limit} \label{sec:2}

\subsection{A quick review of the von Nuemann algebra in JT gravity with matter} \label{sec:2.1}

JT gravity \cite{Jackiw:1984je,Teitelboim:1983ux} with matter is defined by the following Euclidean action \cite{Harlow:2018tqv,Jafferis:2019wkd}
\begin{align}  \label{JTact} 
I[g,\Phi,\vp]&=-S_0 \chi +I_{JI}[g,\Phi]+I_m[g,\vp] \\
I_{JT}[g,\Phi]&=-\int_\mM \sqrt{g}\Phi (R+2)-2 \int_{\del\mM}\sqrt{h}\Phi(K-1)
\end{align}
with boundary condition
\be  
g_{uu}\ra 1/\e^2,\quad \Phi\ra \phi_b/\e
\ee 
where $u$ is the Euclidean boundary affine time. In this paper, we only consider disk topology (Euler characteristic $\chi=1$) by setting $S_0\ra\infty$. Since the topology is fixed, we will suppress $S_0$ throughout this paper. Since the boundary is a circle of circumference $\b$, $u$ is periodically identified $u\simeq u+\b$. 
Integrating over the dilaton field $\Phi$ leads to AdS$_2$ geometry $R+2=0$, whose metric can be chosen as 
\be  
ds^2=d\r^2+\sinh^2\r d\phi^2,\quad \phi\simeq\phi+2\pi
\ee
The Gibbons-Hawking term in \eqref{JTact} is nontrivial and promotes $\phi$ to dynamical field $\phi(u)$ on the boundary. This is the boundary graviton degrees of freedom of reparameterization, whose dynamics is given by Schwarzian derivative of $\phi(u)$ \cite{Maldacena:2016upp}. One can quantize this Schwarzian theory and solve the quantum gravity of pure JT gravity \cite{Yang:2018gdb,Jafferis:2019wkd,Kitaev:2018wpr}. 

With matter action, since it does not couple to dilaton, this is a quantum field theory of $\vp(\r,\phi)$ living in the AdS$_2$ background. We can define boundary operator by pushing $\vp(\r,\phi)$ to the boundary \cite{Kolchmeyer:2023gwa}
\be  
\mO(\phi)=\lim_{\r\ra\infty}e^{\D \r} \vp(\r,\phi) \label{eq:2.5}
\ee 
where $\D$ is the conformal dimension of $\vp$. We assume $\mO(\phi)$ forms a unitary CFT on a circle of circumference $2\pi$. As $\mO(\phi)$ is an operator located at $\phi$ and $\phi$ is promoted to a dynamical field $\phi(u)$, the correlation function of $\mO$ with graviton correction is given by the Euclidean path integral \cite{Yang:2018gdb}
\be  
\avg{\mO_1(u_1)\cdots\mO_n(u_n)}_{QG}=\int D\phi e^{-Sch(\phi(u))} \avg{\mO_1(\phi(u_1))\cdots\mO_n(\phi(u_n))}_{CFT} \label{pt-f}
\ee 
The CFT correlation function is fixed by OPE coefficients and conformal symmetry. In particular, the two-point function and three-point function on a circle is 
\begin{align}
\avg{\mO_i(\phi_1)\mO_j(\phi_2)}&=\f {\d_{ij}} {(\sin\f{\phi_1-\phi_2}{2})^{2\D_i}} \label{7}\\
\avg{\mO_i(\phi_1)\mO_j(\phi_2)\mO_k(\phi_3)}&=\f {C_{ijk}} {(\sin\f{\phi_1-\phi_2}{2})^{\D_i+\D_j-\D_k}(\sin\f{\phi_1-\phi_3}{2})^{\D_i+\D_k-\D_j}(\sin\f{\phi_2-\phi_3}{2})^{\D_j+\D_k-\D_i}} \label{8}
\end{align}
where $\phi_1-\phi_2,\phi_2-\phi_3,\phi_1-\phi_3\in(0,2\pi)$. Unitarity requires that all OPE coefficients $C_{ijk}$ are real.

We can split the disk along the horizontal diameter to two half-disks. The boundary of the lower half-disk is a semicircle with ends at $\phi=0$ and $\phi=\pi$. Analytically continue the boundary time $u\ra i u$, we will have a Lorentzian theory with two asymptotic boundaries, which we call left and right systems. The Hilbert space of JT with matter theory for this two-sided system is $\mH=\mH_0\otimes L^2(\R)$, where $\mH_0$ is the Hilbert space of the boundary CFT and $L^2(\R)$ is the Hilbert space of pure JT gravity \cite{Harlow:2018tqv}. 

\begin{figure}
\begin{centering}
\subfloat[]{\begin{centering}
\includegraphics[height=1.2cm]{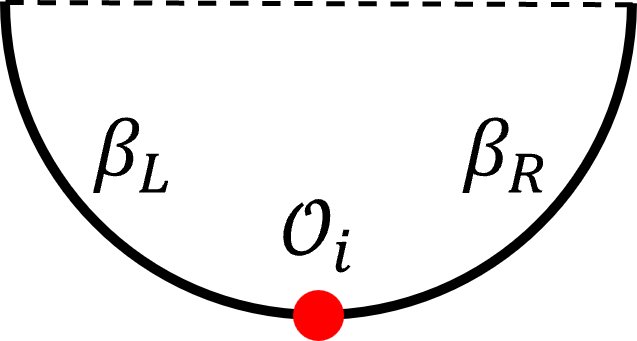}\label{fig:1a}
\par\end{centering}

}\subfloat[]{\begin{centering}
\includegraphics[height=1.2cm]{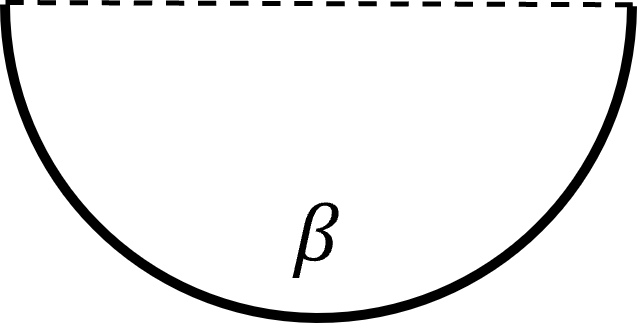}\label{fig:1b}
\par\end{centering}
}\subfloat[]{\begin{centering}
\includegraphics[height=1.7cm]{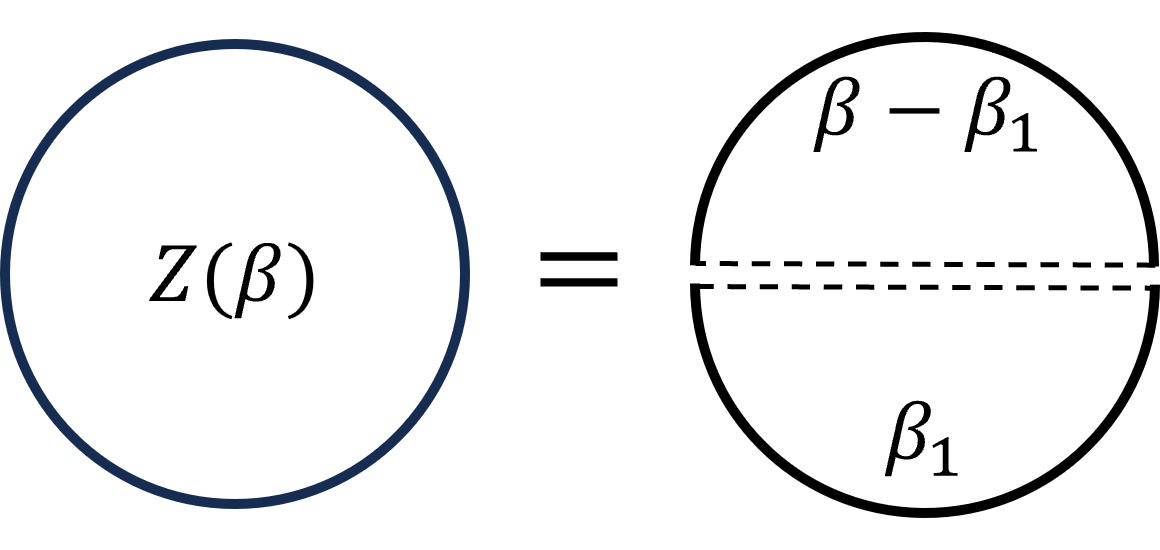}\label{fig:1c0}
\par\end{centering}
}\\
\subfloat[]{\begin{centering}
\includegraphics[height=2cm]{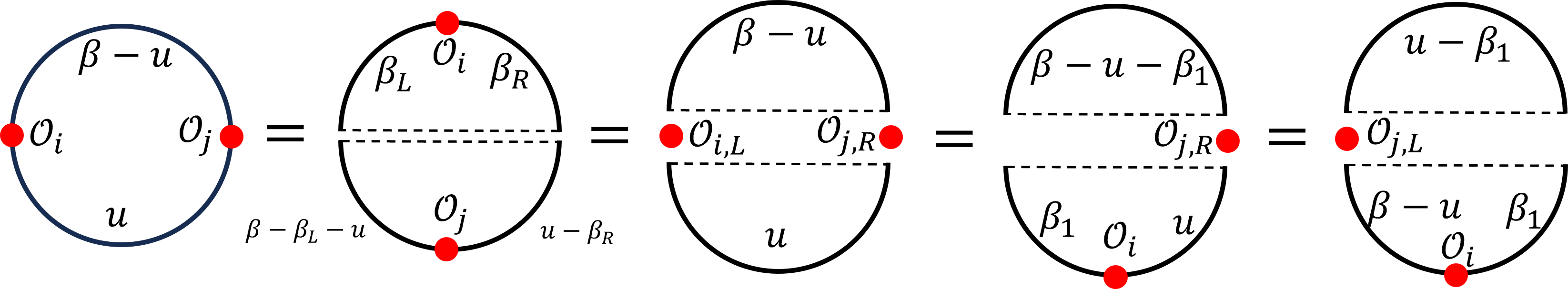}\label{fig:1d0}
\par\end{centering}
}
\\
\subfloat[]{\begin{centering}
\includegraphics[height=2cm]{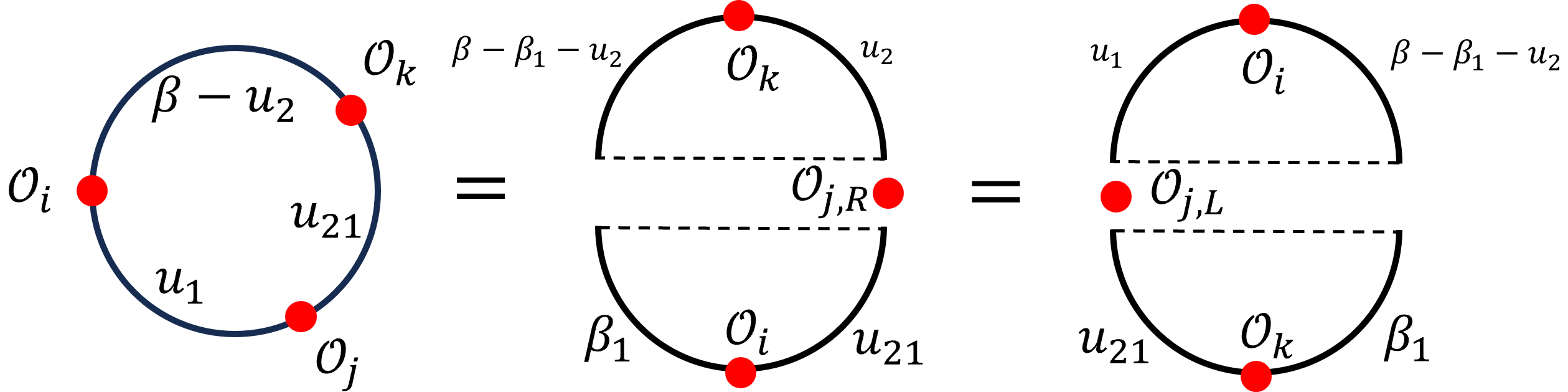}\label{fig:1e0}
\par\end{centering}
}
\par\end{centering} 
\caption{(a) The graphic representation of state $\bra{i;\b_L,\b_R}$. (b) The graphic representation of Hartle-Hawking state without matter $\bra{\b}$. (c) The partition function $Z(\b)$ is equivalent to the inner product of two Hartle-Hawking states. (d) A few identities of two-point functions with inner product and matrix elements. (e) A few identities of three-point functions with matrix elements.}
\end{figure}

For the pure JT gravity part, we can canonically quantize the two-sided theory \cite{Harlow:2018tqv} and find that the system has a canonical variable of renormalized length $\l$ of the geodesic connecting two AdS boundaries at identical left and right Lorentzian time,\footnote{The left and right opposite time shift is the boost transformation in the global gauge symmetry $SL(2)$. Correspondingly, $\l$ is invariant under this opposite time shift.} and the Hamiltonian has a Liouville potential. In this formalism, states are labeled by $\bra{\l}$ and all sqaure-normalizable wavefunctions span the Hilbert space $L^2(\R)$. For the matter part, the Hilbert space is span by all single operator $\mO_i(\phi)$ inserted at $\phi\in[0,\pi]$ on the lower semicircle. Here we do not need to include multiple operator insertion because they can be expanded in terms of single operators through OPE. However, as $\phi$ is dynamical, it is better to consider a new basis of Hartle-Hawking states $\bra{\Psi^\mO_{\b_L,\b_R}(\l)}$ introduced in \cite{Kolchmeyer:2023gwa}, where one $\mO$ is inserted but its location $\phi$ is integrated over $[0,\pi]$ against some wavefunction corresponding to the state with geodesic length $\l$. $\b_L$ and $\b_R$ labels the boundary Euclidean time of $\mO$ to the two ends of the semicircle (see Figure \ref{fig:1a}). Correspondingly, we define the state prepared in this way as
\be 
\bra{i;\b_L,\b_R}\equiv \int_{-\infty}^{\infty}d\l \bra{\Psi^{\mO_{i}}_{\b_L,\b_R}(\l)}\otimes \bra{\l} \label{hh1}
\ee
For the Hartle-Hawking state without matter (see Figure \ref{fig:1b}), we just take $\mO_i=\I$ and this leads to
\be  
\bra{\b}\equiv \int_{-\infty}^{+\infty} d\l \Psi_\b(\l)\bra{0}\otimes\bra{\l}=\int_{-\infty}^{+\infty} d\l\int_0^\infty dk \r(k)e^{-\b k^2/(4\phi_b)} 2K_{2ik}(4e^{-\l/2})\bra{0}\otimes \bra{\l} \label{hh2}
\ee 
where $\Psi_\b(\l)$ is the Hartle-Hawking wavefunction for a half-disk that is bounded by geodesic $\l$ and an EAdS boundary of length $\b=\b_L+\b_R$, $ \bra{0}$ is the vacuum of matter fields and $\r(k)$ is the spectral density with $k\geq 0$ 
\be  
\r(k)=\f {2k} {\pi^2}\sinh (2\pi k)=\f 1 {\pi \G(2ik) \G (-2ik)}
\ee 

Using the Hartle-Hawking states in \eqref{hh1} and \eqref{hh2}, we can express the Euclidean quantum gravity correlation functions in terms of inner products and matrix elements within these states. For example, the partition function of the Schwarzian theory on a circle with circumference $\b$ in our notation is
\be 
Z(\b)= \f {(4\phi_b/\b)^{3/2}} {\sqrt{\pi}} e^{4\phi_b \pi^2/\b} \label{Z}
\ee
where we need to divide an infinite volume of $SL(2)$ in the measure of \eqref{pt-f} if no matter field is inserted because of the global $SL(2)$ gauge symmetry of pure JT gravity. On the other hand, it is the inner product of two Hartle-Hawking states\footnote{The inner product of the Hilbert space $\mH$ is well-defined because $\bra{\l}$ is delta function normalized and the CFT is unitary.} (see Figure \ref{fig:1c0})
\be 
Z(\b)=\avg{\b-\b_1|\b_1}=\int_{-\infty}^\infty d\l \Psi_{\b-\b_1}(\l)\Psi_{\b_1}(\l)=\int_0^\infty dk e^{-\b k^2/(4\phi_b)}\r(k) \label{13}
\ee 
where we used orthogonality of $\bra{\l}$ states $\avg{\l|\l'}=\d(\l-\l')$ and the normalization identity
\be  
\int d\l 2K_{2ik}(4\e^{-\l/2}) 2 K_{2ik'}(4e^{-\l/2})= \f{\d(k-k')}{\r(k)} \label{eq:2.14}
\ee 
Let us define the operators acting on the left (at $\phi=0$) as $\mO_{i,L}$, and the operators acting on the right (at $\phi=\pi$) as $\mO_{i,R}$. The two-point function of matter operators can be expressed in various ways as
\begin{align}
\avg{\mO_j(u)\mO_i(0)}_{QG}&=\avg{i;\b_L,\b_R|j;\b-\b_L-u,u-\b_R}=\avg{\b-u|\mO_{i,L}\mO_{j,R}|u} \label{14-equal}\\
&=\avg{\b-u-\b_1|\mO_{j,R}|i;\b_1,u}=\avg{u-\b_1|\mO_{j,L}|i;\b-u,\b_1}\label{14-equal2}\\
&=\d_{ij}\int_{-\infty}^{\infty}d\l \Psi_{\b-u}(\l) (4e^{-\l})^\D \Psi_{u}(\l) \label{15}\\
&=\d_{ij}\int_0^\infty dk_L dk_R \r(k_L) \r(k_R) \G^{\D_i}_{k_L,k_R}e^{-u k_R^2/(4\phi_b)}e^{-(\b-u)k_L^2/(4\phi_b)} \label{2ptO}
\end{align}
where in the last step we integrate over $\l$ and we define
\be  
\G^{\D}_{k_1,k_2}=\int_{-\infty}^\infty d\l 2K_{2ik_1}(4e^{-\l/2})(4e^{-\l})^\D 2K_{2ik_2}(4e^{-\l/2})=\f{1}{\G(2\D)}\prod_{\pm\pm}\G(\D\pm ik_1\pm ik_2) \label{Gd19}
\ee 
where the product is over all four choices of signs. Taking $\D\ra 0$ in \eqref{Gd19}, it reduces to the normalization \eqref{eq:2.14} for $k_{1,2}\geq 0$.\footnote{To determine the normalization of the delta function we can integrate $k_1$ in \eqref{Gd19} from 0 to $\infty$, which by Barnes integral of four gamma functions leads to \eqref{eq:2.14}. However, this normalization is 1/2 of \cite{Kolchmeyer:2023gwa,Yang:2018gdb}, and our $\r(k)$ is twice of theirs. The discrepancy comes from how to treat the gamma function integral in $\D\ra 0$ limit. However, our normalization is consistent with the 6j symbol in Appendix \ref{app:1}, which reduces to \eqref{Gd19} by setting one of $\D$ to zero. This normalization discrepancy does not affect any main conclusion of this paper.} The first line \eqref{14-equal} is due to the rotation symmetry of the Euclidean diagram in Figure \ref{fig:1d0}. Note that \eqref{15} is a nontrivial result by solving the Schwarzian Euclidean path integral \eqref{pt-f} and express the result in terms of integral over geodesic length $\l$ \cite{Yang:2018gdb}. Comparing with \eqref{7}, we see the rule is simply replacing $\sin(\phi_{ij}/2)$ with $e^{\l_{ij}/2}/2$ when gluing against Hartle-Hawking wavefunctions. This rule is not surprising because semiclassically  a two-point function in AdS background should be an exponential decaying function of the geodesic distance $\l$ with rate of conformal dimension $\D$ at least for heavy fields and long enough distance. Indeed, this rule also holds for three-point functions \eqref{8} (see Figure \ref{fig:1e0})
\begin{align}
&\avg{\mO_k(u_2)\mO_j(u_1)\mO_i(0)}_{QG}\nn\\
=&\avg{k;\b_1,u_{21}|\mO_{j,R}|i;\b-u_2-\b_1,u_1}=\avg{i;u_1,\b-u_2-\b_1|\mO_{j,L}|k;u_{21},\b_1} \label{19}\\
=& C_{ijk}\int d\l_1 d\l_2 d\l_3 \Psi_{\b-u_2}(\l_1)\Psi_{u_1}(\l_2)\Psi_{u_{21}}(\l_3)I(\l_1,\l_2,\l_3) \nn\\
&\times (2e^{-\l_1/2})^{\D_i+\D_k-\D_j}(2e^{-\l_2/2})^{\D_i+\D_j-\D_k}(2e^{-\l_3/2})^{\D_j+\D_k-\D_i}\nn\\
\equiv & \int_0^\infty \left[\prod_{s=1,2,3}dk_s \r(k_s)\right] \sqrt{\G^{\D_i}_{k_1k_2}\G^{\D_j}_{k_2k_3}\G^{\D_k}_{k_3k_1}}e^{-\f{(\b-u_2)k_1^2}{4\phi_b}}e^{-\f{u_1k_2^2}{4\phi_b}}e^{-\f{u_{21}k_3^2}{4\phi_b}}V_{ijk}(k_1,k_2,k_3) \label{3ptO}
\end{align}
where $u_{21}\equiv u_2-u_1,u_1,u_2\in[0,\b]$ and $I(\l_1,\l_2,\l_3)$ is the wavefunction for GHZ state \cite{Yang:2018gdb} that corresponds to the bulk triangular region bounded by three geodesics
\begin{align}
I(\l_1,\l_2,\l_3)&=\int_0^\infty dk \r(k) 2K_{2ik}(4e^{-\l_1/2})2K_{2ik}(4e^{-\l_2/2})2K_{2ik}(4e^{-\l_3/2}) \nn\\
&=\exp \left[ -2\left(e^{(\l_1-\l_2-\l_3)/2}+e^{(\l_2-\l_3-\l_1)/2}+e^{(\l_3-\l_2-\l_1)/2}\right) \right] \label{eq:2.22}
\end{align}
In \eqref{3ptO}, we define the function $V_{ijk}(k_1,k_2,k_3)$, which by \eqref{Gd19} can be written as
\be 
\sqrt{\G^{\D_i}_{k_1k_2}\G^{\D_j}_{k_2k_3}\G^{\D_k}_{k_3k_1}}V_{ijk}(k_1,k_2,k_3)=C_{ijk}\int dk_0 \r(k_0)\G^{\f{\D_i+\D_k-\D_j}{2}}_{k_1,k_0}\G^{\f{\D_i+\D_j-\D_k}{2}}_{k_2,k_0}\G^{\f{\D_j+\D_k-\D_i}{2}}_{k_3,k_0}
\ee 
which in the notation of $V_{abc}(k_{ca},k_{ab},k_{bc})$ is symmetric under permutation of $a,b,c$ assuming $k_{ab}=k_{ba}$. 

From \eqref{2ptO} and \eqref{3ptO}, we can consider an alternative energy basis $\bra{i;k_L,k_R}$ and $\bra{0;k}$ as Laplace transformation of basis $\bra{i;\b_L,\b_R}$ and $\bra{\b}$. We call them energy basis because we can define left Hamiltonian $H_L$ and right Hamiltonian $H_R$ as the conjugate operator to $\b_L$ and $\b_R$ respectively and they act on these basis diagonally
\be  
H_{L/R}\bra{i;k_L,k_R}=\f{k_{L/R}^2}{4\phi_b}\bra{i;k_L,k_R},\quad H_L\bra{0;k}=H_R\bra{0,k}=\f{k^2}{4\phi_b}\bra{0;k}
\ee 
It is shown in \cite{Kolchmeyer:2023gwa} that these energy basis are orthogonal and span the full Hilbert space. Define the following projectors onto vacuum sector and non-vacuum sector 
\be 
\int_0 ^\infty dk\r(k) \bra{0;k}\ket{0;k},\quad \int_0^\infty dk_L dk_R \r(k_L)\r(k_R)\bra{i;k_L k_R}\ket{i;k_L k_R}
\ee 
where the spectral density is a convention of convenience and these continuous energy basis are respectively $1/\r(k)$ and $1/(\r(k_L)\r(k_R))$ normalized. Let us assume one-point function of a single operator vanishes in  \eqref{pt-f}. Inserting these basis into \eqref{13}, \eqref{14-equal}, \eqref{14-equal2} and \eqref{19}, we can work out the following inner products and matrix elements
\begin{gather}
\avg{0;k|\b}=e^{-\b k^2/(4\phi_b)},\quad \avg{j;k_L, k_R|i;\b_L,\b_R}=\d_{ij}\sqrt{\G^{\D_i}_{k_L,k_R}}e^{-\b_L k_L^2/(4\phi_b)}e^{-\b_R k_R^2/(4\phi_b)} \\
\avg{0;k_1|\mO_{i,L/R}|0,k_2}=0,\quad \avg{j;k_L,k_R|\mO_{i,L/R}|0;k}=\d_{ij}\f{\d(k_{R/L}-k)}{\r(k)}\sqrt{\G^{\D_i}_{k_L,k_R}} \label{26}\\
\avg{k;k_L',k_R'|\mO_{j,L/R}|i;k_L,k_R}=\f{\d(k_{R/L}-k_{R/L}')}{\r(k_{R/L})}V_{ijk}(k_{R/L},k_{L/R},k_{L/R}')\sqrt{\G^{\D_j}_{k_{L/R},k_{L/R}'}} \label{27}
\end{gather}

The matrix elements \eqref{26} and \eqref{27} define the operators $\mO_{i,L}$ and $\mO_{i,R}$ completely. Following \cite{Kolchmeyer:2023gwa}, one can define the left algebra $\mA_L=\{\mO_{i,L}, H_L\}''$ and right algebra $\mA_R=\{\mO_{i,R}, H_R\}''$, where the double prime means double commutant, which close the algebra to a von Neumann algebra. Using the matrix elements \eqref{26} and \eqref{27}, we can show that the left and right algebras are commutant to each other $\mA_L=\mA_R'$. Moreover, they are both factors, which means $\mA_L\cap\mA_R=\C \cdot\I$. It follows that the algebra of all bounded operators $\mB(\mH)$ acting on the Hilbert space $\mH$ is generated by the union of $\mA_L$ and $\mA_R$, namely $\mB(\mH)=(\mA_L\cup\mA_R)''$.

Von Neumann algebras are classified into three types. In this paper, we will not review the classification but refer readers to \cite{Sorce:2023fdx}. Roughly speaking, type I allows a factorization of Hilbert space $\mH=\mH_L\otimes\mH_R$, but type II and III von Neumann algebras do not. On the other hand, type I and II von Neumann algebras allow the existence of a trace but type III does not. For type II, if all operators have finite trace then we can rescale the definition of the trace such that all operators have trace within $[0,1]$, and call it type II$_1$; if there exists an operator with infinite trace, then we call it type II$_\infty$. The von Neumann algebra $\mA_L$ and $\mA_R$ are both type II$_\infty$ \cite{Penington:2023dql,Kolchmeyer:2023gwa} because 
\begin{enumerate}
    \item The Hilbert space does not not allow factorization as we can see from the energy basis.
    \item There exists a trace
\be 
\tr a \equiv \lim_{\b\ra 0}\avg{\b|a|\b}, \quad a\in\mA_{L/R} \label{eq:tr}
\ee 
\item The identity operator has infinite trace $\tr \I=Z(0)=\infty$ by \eqref{Z}.
\end{enumerate}

\subsection{The semiclassical limit} \label{sec:2.2}

The II$_\infty$ von Neumann algebra in JT gravity with matter holds for any finite $\phi_b$. In this case, we can algebraically define the entanglement wedge algebra of left and right as $\mA_L$ and $\mA_R$ respectively. However, as the gravity variable $\l$ is a quantum operator, which has fluctuation in a generic state, the bulk field $\p$ does not live in a fixed background. Consequently, we do not have a semiclassical bulk geometry spacetime causal structure. In particular, there is no notation of a bulk causal wedge and the algebraic entanglement wedge does not correspond to a local bulk quantum field theory in a fixed spacetime region bounded by QES formula \cite{Engelhardt:2014gca}.

As we take $\phi_b\ra \infty$, we arrive at the semiclassical limit. Unlike higher dimensional, in JT gravity, there are two levels of semiclassical limit. We first take $S_0\ra\infty$ to restrict ourselves to disk topology and eliminate non-perturbative baby universes with probabilities of $O(e^{-S_0})$. In this case, $\phi_b$ becomes the effective $1/G_N$ analogous to higher dimensions, which controls the perturbative fluctuations of gravity. When we further take $\phi_b\sim 1/G_N\ra \infty$, the Euclidean path integral of JT gravity is dominated by the saddle. Solving the saddle leads to a fixed spacetime background and all geometric notations emerge. It follows that the theory reduces to matters propagating in a fixed curved spacetime background. As we will show shortly, the geometric notation of the causal wedge becomes {\it emergent} under this limit. At the same time, the entanglement wedge of the left or right boundary is endowed with a geometric meaning and is bounded by the RT surface, which in JT gravity corresponds to the global minimal dilaton point. From the geometric picture, we will see that the emergent causal wedge is nested inside the entanglement wedge as expected \cite{Headrick:2014cta, Wall:2012uf}.

The states defined in Section \ref{sec:2.1} are good to show the properties of the von Neumann algebra, but are not convenient to discuss the semiclassical limit. Instead, we will consider a different set of states by inserting primaries at different locations on the lower Euclidean semicircle (see Figure \ref{fig:pets}). As we consider all possible insertions of primaries, these states span the same Hilbert space. These states are called partially entangled thermal state (PETS), which was first studied in \cite{Goel:2018ubv}. The simplest PETS state has one heavy operator insertion
\begin{equation}
\bra{\Psi}\equiv\bra{i;\b_L/2,\b_R/2}=e^{-\b_{L}H_{L}/2-\b_{R}H_{R}/2}\mO_i\bra{v}\bra 0\label{eq:1}
\end{equation}
where $\bra v$ is the ground state of the CFT and $\bra{0}$ is the $\b=0$ Hartle-Hawking state. If we have a UV completion of JT with matter on a two-sided system, e.g. two identical SYK model, the state $\bra{v}\bra{0}$ could be understood as the EPR state of the two-sided system with maximal entanglement entropy. Acting on this state, $\mO_i$ could be chosen as either $\mO_{i,L}$ or $\mO_{i,R}$. The conformal weight of the heavy operator $\mO_i$ is $\mu\sim O(\phi_b)$. A generic PETS with multiple heavy operators insertion is pictorially represented in Figure \ref{fig:pets}. 

Even though the Hilbert space is non-factorized in type II von Neumann algebra, we can still define reduced density matrix $\r_{L,R}\in\mA_{L,R}$ for either left or right side. Using the pictorial representation in Figure \ref{fig:1a} for PETS, we can easily find out the reduced density matrix for one side by conjugating the state and connecting the Euclidean path on the other side. For example, for $\bra{\Psi}$ defined in \eqref{eq:1}, the right reduced density matrix is
\be 
\r_R=e^{-\b_RH_R/2}\mO_{i,R}e^{-\b_L H_R}\mO_{i,R}e^{-\b_RH_R/2} \label{eq:rhoR}
\ee
where we assume $\mO_{i,R}$ is hermitian. We pictorially represent $\r_R$ in Figure \ref{pic-rhoR}. Using the Euclidean path integral, it is straightforward to check that \eqref{eq:rhoR} obeys
\be  
\avg{\Psi|a_R|\Psi}=\Tr \r_R a_R,\quad \forall a_R\in \mA_R
\ee 
with the trace defined in \eqref{eq:tr}.

\begin{figure}
\begin{centering}
\subfloat[]{\begin{centering}
\includegraphics[height=2cm]{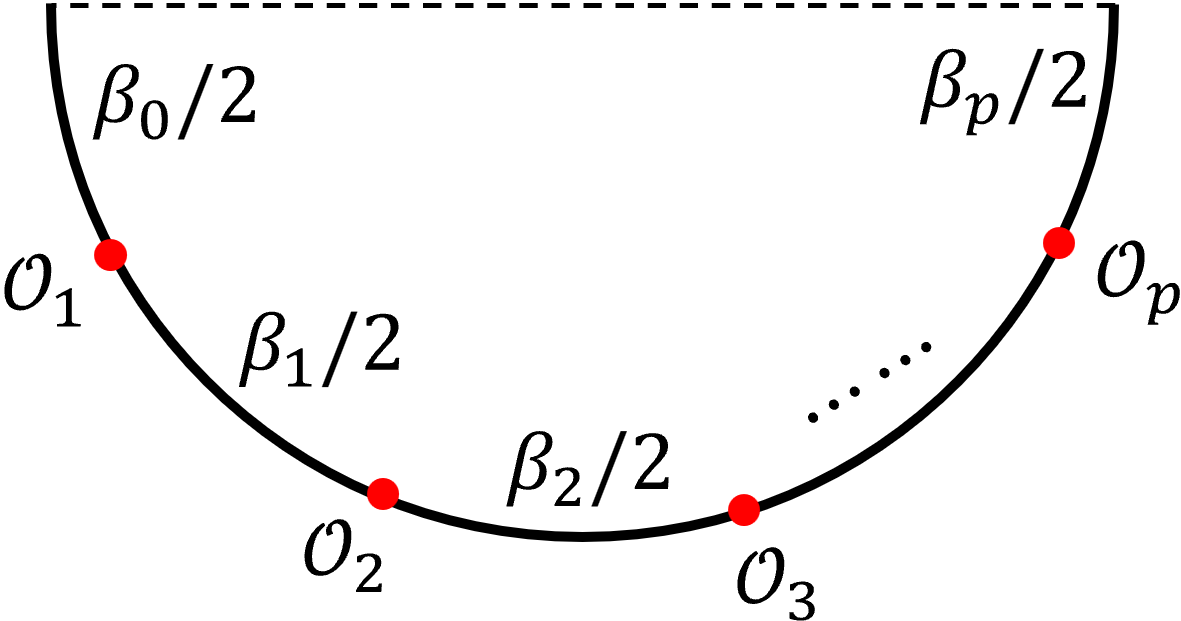}\label{fig:pets}
\par\end{centering}

}\subfloat[]{\begin{centering}
\includegraphics[height=2.4cm]{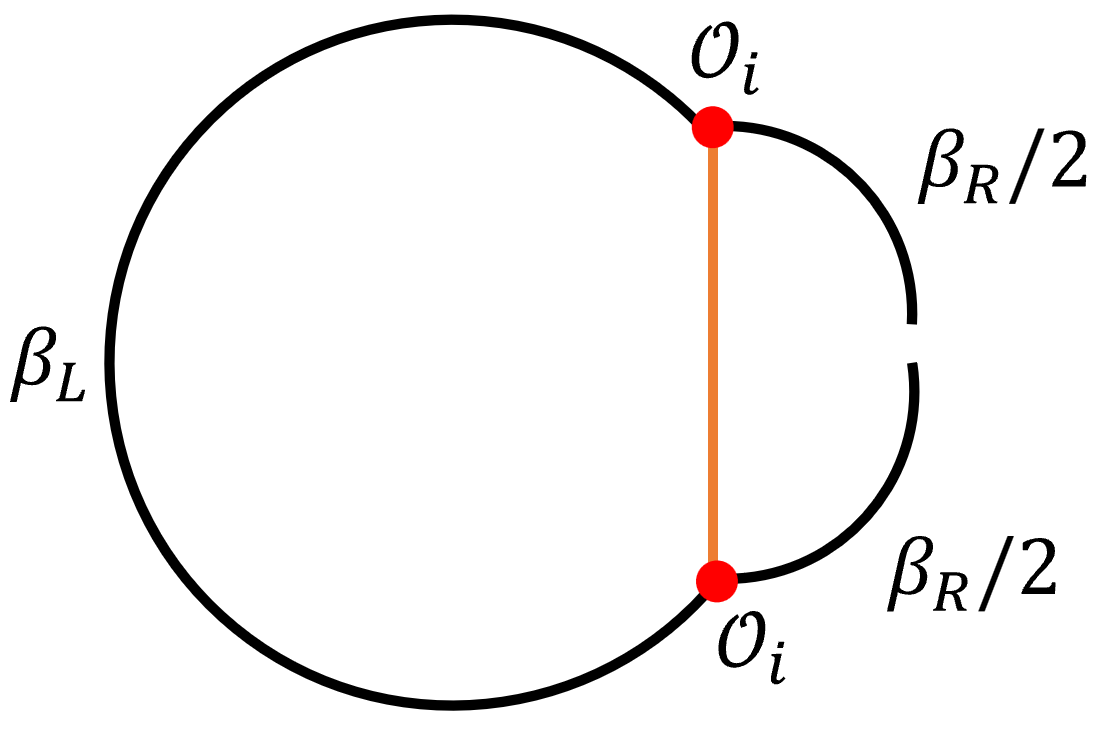}\label{pic-rhoR}
\par\end{centering}
}\subfloat[$\mu<\mu_{*}$\label{fig:1c}]{\begin{centering}
\includegraphics[width=3cm]{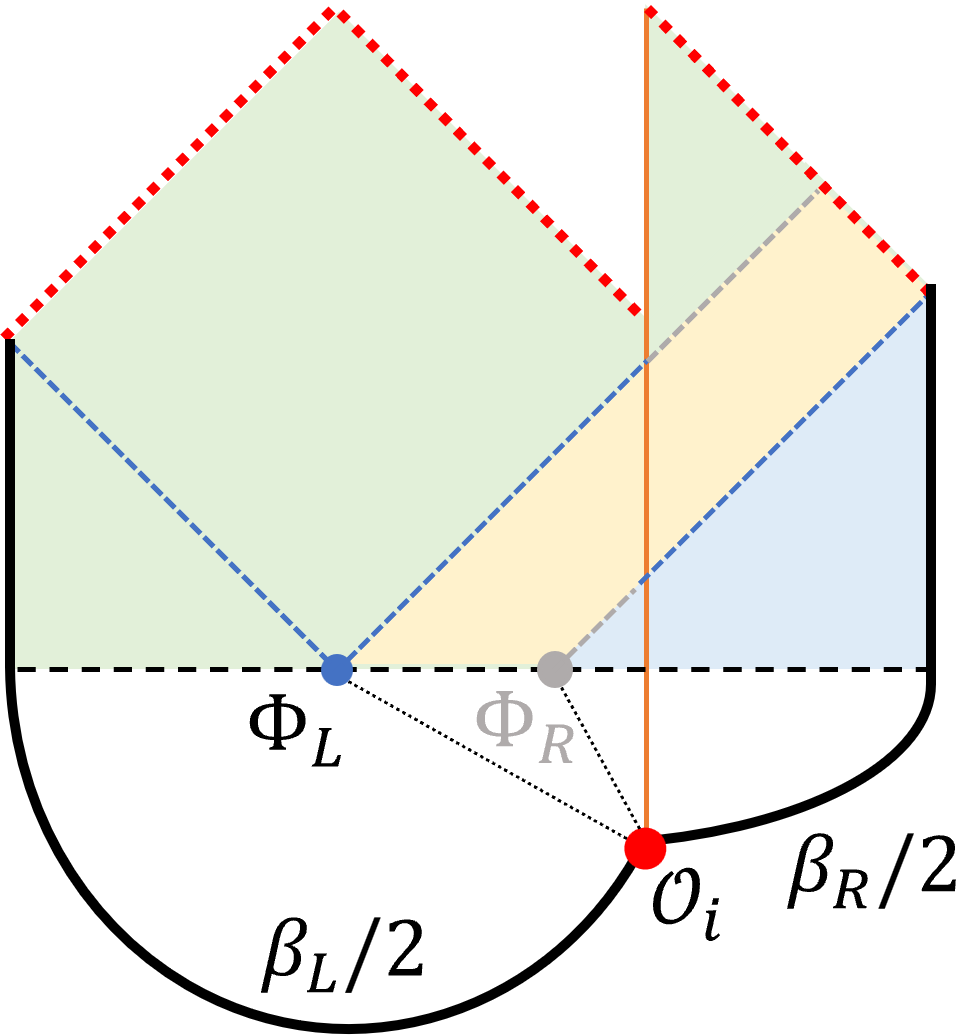}
\par\end{centering}
}\subfloat[$\mu>\mu_{*}$\label{fig:1d}]{\begin{centering}
\includegraphics[width=3cm]{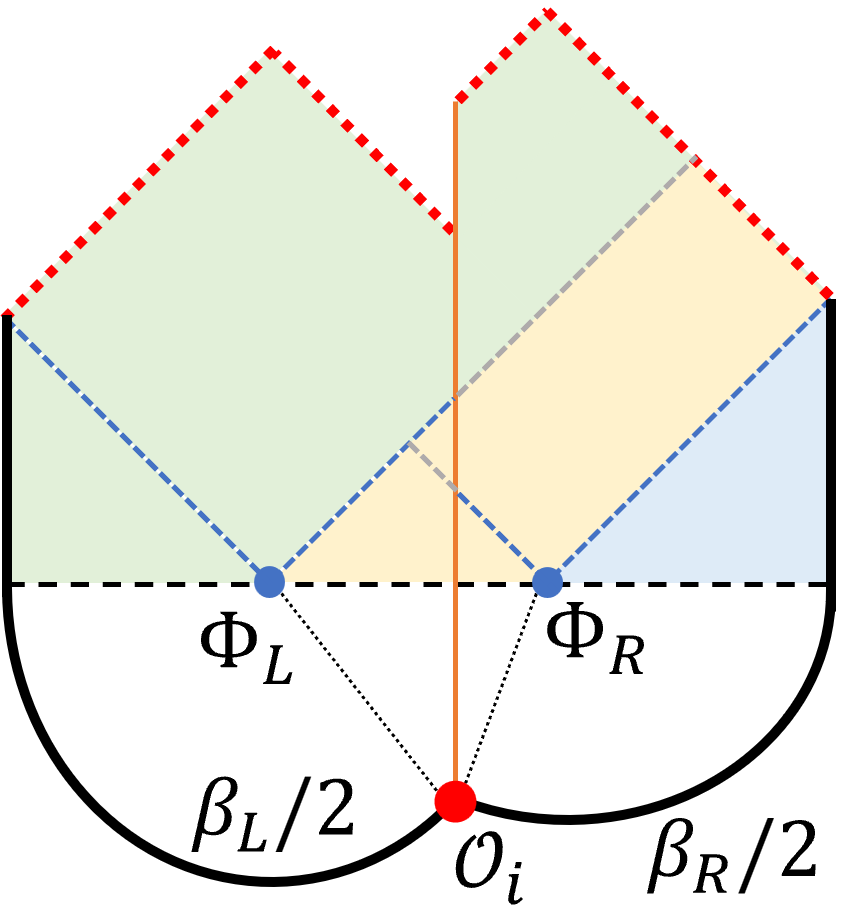}
\par\end{centering}
}
\par\end{centering} 
\caption{(a) A generic PETS with $p$ primaries inserted on the Euclidean path. (b) The right reduced density matrix $\r_R$ for state $\bra{\Psi}$. The orange line connecting two $\mO_i$ represents the geodesic that back reacts on the geometry. (c) and (d) are the Lorentzian spacetime for $\bra{\Psi}$ for $\mu<\mu_*$ and $\mu>\mu_*$ respectively. In both pictures, we assume $\b_{L}>\b_{R}$ and the global minimum of the dilaton profile is always at $\Phi_{L}$. The red dashed line
is when $\Phi=-\Phi_{L,R}$, which corresponds to a cutoff at inner
horizons. The orange line is the geodesic of the heavy matter, along
which the dilaton profiles on both sides are glued continuously (but
not smoothly). We have chosen the inertial frame of this geodesic
such that it is a vertical straight line. For the right system, the yellow region is the entanglement
wedge and the blue region is the causal wedge. For $\mu<\mu_{*}$
there is only one local minimum of dilaton at $\Phi_{L}$ (the invisible
``local minimum'' $\Phi_{R}$ is labeled gray) and there is no python's
lunch. For $\mu>\mu_{*}$ there are two local minima of dilaton at
$\Phi_{L,R}$ and there is a python's lunch between them with a bulge
at the orange geodesic of the heavy matter.}
\end{figure}

It has been studied in \cite{Goel:2018ubv} that for this state there is a fixed spacetime geometry in semiclassical limit $\phi_b\ra\infty$. Due to the backreaction of heavy operator $\mO_i$, the Euclidean
geometry is the gluing of two EAdS disks along the geodesic connecting two
$\mO_i$. For each disk, the origin is the minimum of dilaton $\Phi_{L,R}$
on that disk. If $\b_{L}>\b_{R}$, we have $\Phi_{L}<\Phi_{R}$
and vice versa. Continuing to Lorentzian signature at the Cauchy slice connecting these two minima $\Phi_{L,R}$, we will have a spacetime with a long wormhole connecting two AdS boundaries. There are a few cases depending on the parameters. Without loss of generality, let us assume $\b_{L}>\b_{R}$. In this case, the right causal wedge is nested inside
the right entanglement wedge. On the other hand, the left causal wedge
is always equal to the left entanglement wedge. In other words, the global minimal surface is at $\Phi_{L}$. If $\mu<\mu_{*}$ for some critical $\mu_*$, there is only one visible
local minimal $\Phi=\Phi_L$ from both boundaries because the other one is hidden
behind the geodesic connecting two $\mO_i$'s. If $\mu>\mu_{*}$, there
are two visible local minima $\Phi_{L,R}$, one from each boundary and the
global minimum is $\min(\Phi_{L},\Phi_{R})=\Phi_L$. By tuning $\mu$ over
the critical value $\mu_{*}$ from below, a python's lunch \cite{Brown:2019rox} emerges, which means that the dilaton value (or the transverse area in higher dimension) from $\Phi_R$ to $\Phi_L$ is not monotonically decreasing. We summarize these conclusions in Figure
\ref{fig:1c} and \ref{fig:1d}.

For a generic PETS with multiple heavy operators inserted, the semiclassical limit is subtle because it depends on how the OPE coefficient $C_{ijk}$ scales with $\phi_b$. This is an essential difference between JT with matter theory and higher dimensional holographic theories, such as $\mN=4$ SYM. In the higher dimensional holographic theories, the CFT data is given together with the parameter $N$ that controls the semiclassical limit. In particular, the gravity degrees of freedom emerge from the large $N$ limit of the boundary theory. On the other hand, the JT with matter theory is well-defined for arbitrary unitary CFT as long as it couples with JT in the way described in Section \ref{sec:2.1}. In this construction, the gravity degrees of freedom is independently defined from the CFT matter part. One may naturally expect that not all unitary CFT allows a bulk effective field theory description in the semiclassical limit. To have a sensible semiclassical limit, in this paper, we will make the following assumptions about the CFT data:
\begin{enumerate}
    \item All primary operators in the CFT can be separated into two categories: heavy operators $\mO_i$ with conformal dimension $\mu\sim O(\phi_b)$ and light operators $\mO^l_i$ with conformal dimension $\D\sim O(1)$ in $\phi_b\ra \infty$ limit. To distinguish these two types of operators, we take the notation of $\mu$ to represent the conformal dimension of a heavy operator and the notation of $\D$ to represent the conformal dimension of a light operator.
    \item For heavy operators, we assume there is a set $\mO_i$ analogous to the ``single trace" operators in higher dimensional holographic theories, which become generalized free fields in the strict $\phi_b\ra\infty$ limit. This is physically reasonable because we should not expect tunneling or decay of spacetime in the strict $G_N\ra 0$ limit.\footnote{In principle, this assumption of decoupling between heavy and light operators in the semiclassical limit could be made weaker, but we need to specify how they couple to each other with appropriate scaling of $\phi_b$. We thank Edward Witten for discussing this point.}
    \item For light operators, we assume they decouple from the ``single trace" heavy operators, but light operators could have nontrivial light-light-light OPE coefficients $C_{l_1l_2l_3}\sim O(1)$ in  $\phi_b\ra \infty$ limit.
    \item To make above two assumptions consistent with CFT, we need to impose that all other heavy operators are analogous to ``multi-trace" operators $[\mO_i\mO_j\cdots]$, which only have nontrivial OPE coefficients $C_{\mO_1 \mO_2 [\mO_1\mO_2]}\neq 0$ for at least one of $\mO_{1,2}$ is heavy such that the crossing symmetries are obeyed.\footnote{To be more precise, for the four-point function of ``single trace" heavy operators, we requires $C_{h_1h_2[h_1h_2]}\neq 0$ to obey the crossing symmetry of generalized free fields; for the four-point function of heavy-heavy-light-light, which factorizes as the product of two-point function of heavy and light operators, we need $C_{hl[hl]}\neq 0$ to obey the crossing symmetry for this factorization.} On the other hand, any OPE coefficients involving only ``single trace" heavy operators should vanish $C_{h_1h_2l}=C_{l_1l_2h}=C_{h_1h_2h_3}=0$.
\end{enumerate}
The assumption of ``single trace" heavy operators to form a generalized free field theory is crucial and leads to a fixed spacetime background for each PETS in the semiclassical limit. On the other hand, the decoupling between heavy and light operators leads to an emergent closed algebra for light operators in the semiclassical limit. One could impose a stronger but simpler assumption by starting with two decoupled CFTs (the heavy and light) even for finite $\phi_b$. However, the more technical and weaker assumptions above for the CFT data are analogous to a large $N$ holographic theory in higher dimensions. In particular, the light fields form a close algebra only in $\phi_b\ra\infty$ limit is analogous to the statement of single trace operators forming an emergent algebra only in $N\ra \infty$ limit \cite{Leutheusser:2021frk}.

Let us consider a generic PETS with $p$ insertions of ``single trace" heavy
operators in Figure \ref{fig:pets}. Let us assume the spectrum of heavy ``single trace" operators is dense and a generic PETS should have all $p$ different primaries with conformal weight $\mu_{i}$. For simplicity, we assume they are Hermitian. Note that the insertion of ``multi-trace" heavy operators corresponds to the limit of two close ``single trace" heavy operators (or a ``single trace" heavy operator and a light operator), and thus we do not consider them separately. From now on, we will only consider the insertion of ``single trace" heavy operators and briefly call them ``heavy operators" unless specified. This state can explicitly written as
\begin{equation}
\bra{\Psi_{p}}=e^{-\b_{0}H/2}\mO_{1}e^{-\b_{1}H/2}\cdots \mO_{p-1}e^{-\b_{p-1}H/2}\mO_{p}e^{-\b_{p}H/2}\bra v \bra 0 \label{eq:pets}
\end{equation}
where all operators (and the Hamiltonian) are from $\mA_R$ and we suppressed the subscript $R$. It follows that the right reduced density matrix is
\begin{equation}
\r_{p}=e^{-\b_{p}H/2}\mO_{p}e^{-\b_{p-1}H/2}\cdots \mO_{1}e^{-\b_{0}H}\mO_{1}e^{-\b_{1}H/2}\cdots \mO_{p}e^{-\b_{p}H/2} \label{eq:rhopets}
\end{equation}
Since these heavy operators are generalized free fields, the CFT contribution to the the trace of $\rho_p$ is a product of $p$ two-point functions. As shown in Figure \ref{fig:Partition-function-of}, the full quantum
gravity expression for the partition function $Z_{p}$ for this density
matrix is given by 
\begin{equation}
Z_{p}=\Tr \r_p=\avg{\Psi_p|\Psi_p}=\int\left(\prod_{i=1}^{p}d\l_{i}\right)\Psi_{\b_{0}}(\l_{1})(4e^{-\l_{1}})^{\mu_{1}}\Psi_{\b_{1}}(\l_{1},\l_{2})\cdots (4e^{-\l_{p}})^{\mu_{p}}\Psi_{\b_{p}}(\l_{p})\label{eq:4}
\end{equation}
where the wave functions are given by \eqref{hh2} and
\begin{align}
\Psi_{\b}(\l_{1},\l_{2}) & =e^{-\b k^{2}/(4\phi_{b})}\r(k)2K_{2ik}(4e^{-\l_{1}/2})2K_{2ik}(4e^{-\l_{2}/2})
\end{align}
As explained by \cite{Goel:2018ubv,Mertens:2017mtv}, we can use an
integral representation of Bessel function $K$ 
\begin{equation}
K_{\nu}(z)=\f 1{2i}\int_{-i\infty}^{+i\infty}e^{-z\cos\a}e^{i\nu\a}d\a\label{eq:7}
\end{equation}
to get a more geometry-related expression for the partition function.
Taking this representation into (\ref{eq:4}), we see that for each
$e^{-\mu_{i}\l_{i}}$ there are two additional variables $\g_{i}$
and $\a_{i}$ from the two Bessel functions involving $\l_{i}$ respectively.
Integrating over $\l_{i}$, we have
\begin{equation}
Z_{p}=\int\left(\prod_{i=0}^{p}dk_{i}\right)\prod_{i=1}^{p}d\a_{i}d\g_{i}e^{-I_{p}(k_{i},\b_{i};\mu_{i};\a_{i},\g_{i})}
\end{equation}
where the action $I_p$ is given by
\begin{equation}
I_{p}=\sum_{i=0}^{p}\left(\f{k_{i}^{2}}{4\phi_{b}}\b_{i}-\log\r(k_{i})\right)+2\sum_{i=1}^{p+1}k_{i-1}(\a_{i-1}+\g_{i})+\mu_{i}\log(\cos\a_{i}+\cos\g_{i})+I_{0}(\mu_{i})
\end{equation}
where $\a_{0}=\g_{p+1}=0$ and $I_{0}(\mu_{i})$ is a constant.

\begin{figure}
\begin{centering}
\includegraphics[width=5cm]{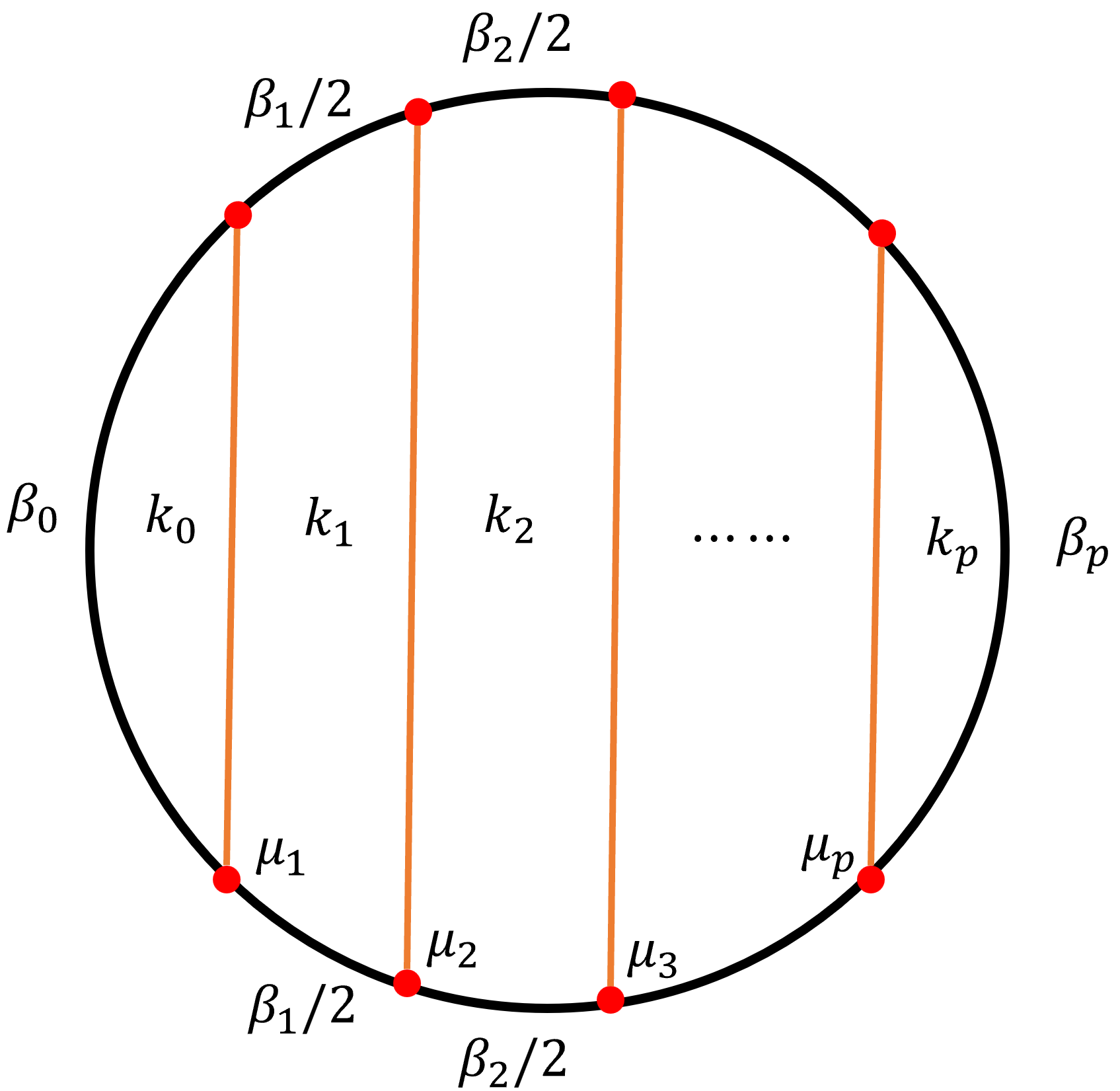}
\par\end{centering}
\caption{Partition function of the density matrix $\r_{p}$.} \label{fig:Partition-function-of}
\end{figure}

To consider the semiclassical limit, we rescale $\mu_{i}\ra 2\mu_{i}\phi_{b},$$k_{i}\ra 2k_{i}\phi_{b}$
and take $\phi_{b}\ra\infty$ limit. Under this limit we have
\begin{equation}
I_{p}/(2\phi_{b})\app\sum_{i=0}^{p}\left(\f{k_{i}^{2}}2\b_{i}-2\pi k_{i}\right)+2\sum_{i=1}^{p+1}k_{i-1}(\a_{i-1}+\g_{i})+\mu_{i}\log(\cos\a_{i}+\cos\g_{i}) \label{eq:2.40}
\end{equation}
where we have dropped the irrelevant constant $I_{0}$. Since $\phi_{b}$
is large, the integral over $k_{i}$, $\a_{i}$ and $\g_{i}$ is approximated
by their saddles. The variation of $\a_{i}$ and $\g_{i}$ gives
\begin{equation}
\f{k_{i-1}}{\sin\g_{i}}=\f{k_{i}}{\sin\a_{i}}=\f{\mu_{i}}{\cos\a_{i}+\cos\g_{i}}\label{eq:11-1}
\end{equation}
The variation of $k_{i}$ gives
\begin{equation}
k_{i}\b_{i}=2(\pi-\a_{i}-\g_{i+1})\label{eq:12-1}
\end{equation}
As shown in \cite{Goel:2018ubv}, the on-shell parameters $\a_{i}$
and $\g_{i}$ have geometric meaning, which indicates an equivalent
geometric computation in the semiclassical limit, as we will explain in the next subsection. This geometric computation
will be very helpful for us to understand the modular flow in the
semiclassical limit. 

\subsection{The geometric formalism in semiclassical limit} \label{sec:2.3}

The geometric formalism in the semiclassical limit treats each boundary as a circle embedded in
EAdS$_{2}$ space with a $SL(2)$ charge $Q$, and the bulk as the
disk inside the circle. Each contracted pair of heavy operators $\mO$
are connected by a geodesic with charge $Q_{\mO}$. For two joint boundaries
connected by $\mO$, their charges need to be conserved, which fixes
the Euclidean geometry. In this sense, $2\a_{i}$ and $2\g_{i}$ are
opening angles of two joint disks to the geodesic of $\mO$ in between.
For example, the geometry for the partition function $Z_{p}$ is shown
in Figure \ref{fig:Geometry-of-boundary.}, where the black curve
is the overall boundary that are connected by many arcs of different
circles.

\begin{figure}
\begin{centering}
\includegraphics[width=9cm]{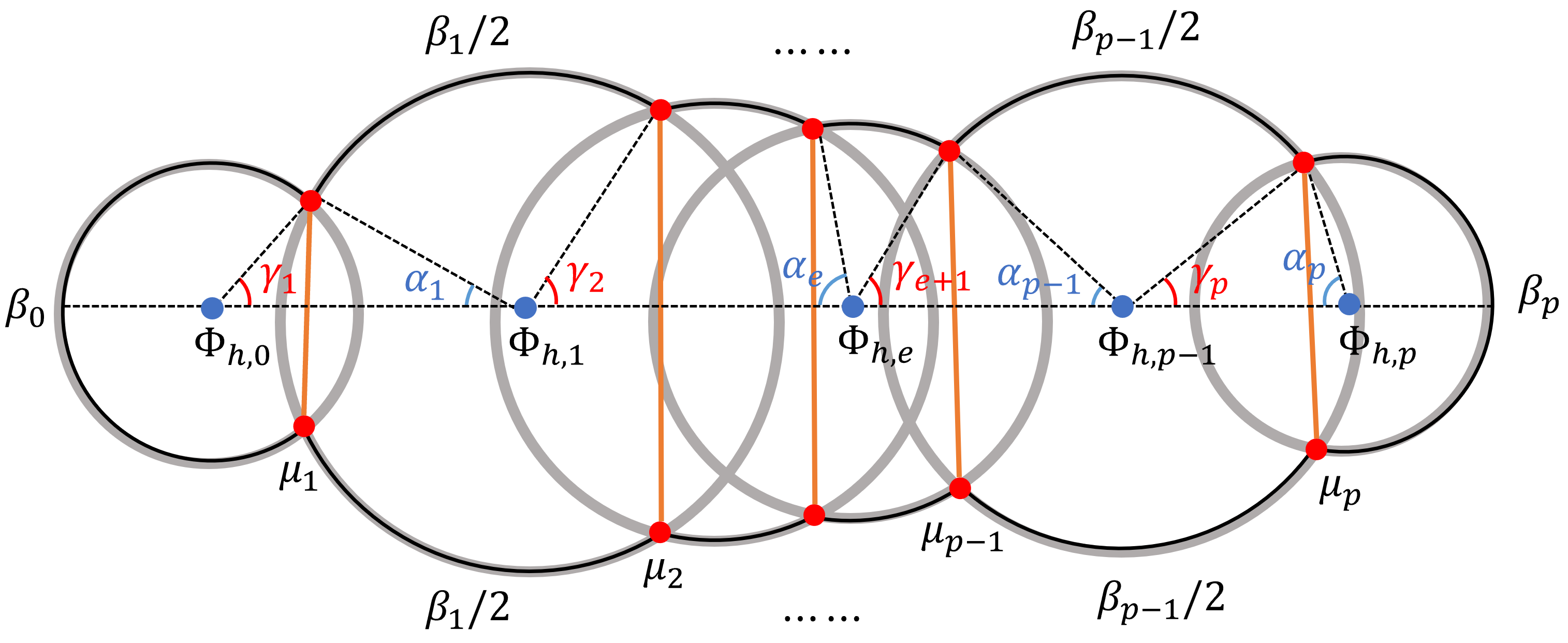}
\par\end{centering} 
\caption{Geometry of the boundary for the partition function $Z_p$.}\label{fig:Geometry-of-boundary.}
\end{figure}

The EAdS$_{2}$ can be parameterized as a hypersurface in three dimension
\begin{equation}
-Y_{-1}^{2}+Y_{0}^{2}+Y_{1}^{2}=-1,\quad ds^{2}=-dY_{-1}^{2}+dY_{0}^{2}+dY_{1}^{2}
\end{equation}
This metric has $SO(2,1)\simeq SL(2)$ isometry generated by 
\begin{equation}
Y_{\mu}\ra M_{i\nu}^{\mu}Y_{\mu},\qquad i=1,2,3
\end{equation}
where 
\begin{equation}
M_{1}(x)=\begin{pmatrix}\cosh x & \sinh x\\
\sinh x & \cosh x\\
 &  & 1
\end{pmatrix},\quad M_{2}(x)=\begin{pmatrix}\cosh x &  & \sinh x\\
 & 1\\
\sinh x &  & \cosh x
\end{pmatrix},\quad M_{3}(x)=\begin{pmatrix}1\\
 & \cos x & \sin x\\
 & -\sin x & \cos x
\end{pmatrix}
\end{equation}
We can choose the coordinate 
\begin{equation}
Y_{-1}=\cosh\r,\quad Y_{0}=\sinh\r\sin\tau,\quad Y_{1}=\sinh\r\cos\tau,\quad\tau\in[0,2\pi],\r\in[0,\infty)\label{eq:8}
\end{equation}
and the metric becomes 
\begin{equation}
ds^{2}=d\r^{2}+\sinh^{2}\r d\tau^{2} \label{eq:2.47}
\end{equation}
The EAdS boundary is a circle with large $\r_{\e}$ and $\tau\in[0,2\pi]$
with boundary condition
\begin{equation}
2\pi\e\sinh\r_{\e}=b,\quad\e\ra0
\end{equation}
This is equivalent to setting the boundary metric as $du_E^{2}/\e^{2}$
for periodicity $u_E\sim u_E+b$. This circle can be equivalently written
as 
\begin{equation}
Q^{\mu}Y_{\mu}=\f 1{\e},\quad Q^{\mu}=\{\f 1{\e\cosh\r_{\e}},0,0\}\label{eq:11}
\end{equation}
where $Q^{\mu}$ is the $SL(2)$ charge of this circle. The dilaton
profile inside each EAdS boundary circle of charge $Q$ is given by
\begin{equation}
\Phi(Y)=\phi_{b}Q\cdot Y\label{eq:20}
\end{equation}
For the circle centered at the origin, the charge is given by (\ref{eq:11})
and we have the dilaton profile as $\Phi=\f{2\pi\phi_{b}}b\cosh\r$
with horizon area $\Phi_{h}=2\pi\phi_{b}/b$. 

Each gray circle in Figure \ref{fig:Geometry-of-boundary.} is an
EAdS$_{2}$ boundary but with different charges 
\begin{equation}
Q_{i}^{\mu}=\{\f 1{\e\cosh\r_{i,\e}},0,0\},\quad2\pi\e\sinh\r_{i,\e}=b_{i}\label{eq:12}
\end{equation}
when we put the center of each circle at the origin of EAdS$_{2}$.
For Figure \ref{fig:Geometry-of-boundary.}, we can choose the circle
with minimal dilaton horizon area at the origin of EAdS$_{2}$. Suppose
this is the $e$-th circle with charge $Q_{e}$ given by (\ref{eq:12}).
Other circles are moved away from the origin by an appropriate $SL(2)$
transformation
\begin{align}
Y^{i} & (\tau)=M_{2}(z_{i})\cdot Y(\r_{i,\e},\tau)\label{eq:22-1}
\end{align}
where $Y^{i}$ is the coordinate of the $i$-th circle and here we
only have $M_{2}$ transformation because all circles in Figure \ref{fig:Geometry-of-boundary.}
are aligned along a horizontal line. The parameters $z_{i}$ are positive
for $i>e$, negative for $i<e$ and equal to zero for $i=e$. Geometrically,
they are the geodesic distance $D_{i}$ between the center of the
$i$-th circle to the origin
\begin{equation}
\cosh D_{i}\equiv-Y^{(i)}\cdot Y^{(e)}|_{\r=0}=\cosh z_{i}\iff D_{i}=z_{i}
\end{equation}

The circles (\ref{eq:22-1}) can be equivalently characterized by
the charges 
\begin{equation}
Q_{i}=\{\f 1{\e\cosh\r_{i,\e}},0,0\}\cdot M_{2}(-z_{i})
\end{equation}
The dilaton profile and horizon area (the minimal value of the dilaton) for each circle is 
\begin{equation}
\Phi_{i}(Y)=\phi_{b}Q_{i}\cdot Y\implies\Phi_{h,i}=\f{2\pi \phi_b\cosh z_{i}}{b_{i}}
\end{equation}
Since the total arc opening angle of each circle is $2(\pi-\a_{i}-\g_{i+1})$,
we have
\begin{equation}
2(\pi-\a_{i}-\g_{i+1})=2\pi\b_{i}/b_{i}
\end{equation}
Comparing with (\ref{eq:12-1}), we should identify
\begin{equation}
k_{i}=2\pi/b_{i}\label{eq:27-1}
\end{equation}
and the horizon areas now becomes $\Phi_{h,i}=\phi_b k_{i}\cosh z_{i}$. 

To show that this geometric formalism is valid, we need to justify
the equations (\ref{eq:11-1}). They are from the charge conservation
and joint constraints for circles. The charge conservation leads to
\begin{equation}
(Q_{i}-Q_{i-1})\cdot(Q_{i}-Q_{i-1})=\mu_{i}^{2}\label{eq:28}
\end{equation}
The joint conditions for two neighboring circles are
\begin{align}
Y^{i-1}(\g_{i}) & =Y^{i}(\pi-\a_{i})\label{eq:29}
\end{align}
Expanding (\ref{eq:28}), we have
\begin{equation}
\mu_{i}^{2}=2k_{i}k_{i-1}\cosh(z_{i}-z_{i-1})-k_{i}^{2}-k_{i-1}^{2}\label{eq:30}
\end{equation}
Expanding (\ref{eq:29}), we have
\begin{equation}
\f{k_{i}}{\sin\a_{i}}=\f{k_{i-1}}{\sin\g_{i}},\quad\cos\a_{i}=\f{k_{i-1}\cosh(z_{i}-z_{i-1})-k_{i}}{k_{i-1}\sinh(z_{i}-z_{i-1})}\label{eq:31}
\end{equation}
One can easily show that (\ref{eq:30}) and (\ref{eq:31}) are equivalent
to (\ref{eq:11-1}).

\subsection{The structure of the semiclassical Hilbert space} \label{sec:2.4}

As we see from the computation above, a generic PETS with all different heavy operator insertions corresponds to a unique fixed geometry of a long wormhole in Figure \ref{fig:Geometry-of-boundary.} that is determined by the saddle equations.

There exists more special PETS with a few identical heavy operators. For example, we may have the following state with two identical insertions
\be 
\bra{\Phi_2}=e^{-\b_L H/2}\mO e^{-\b_0 H} \mO e^{-\b_R H/2}\bra{v}\bra{0}
\ee 
The norm of this state includes different contributions from three Wick contractions of four $\mO$ in the reduced density matrix $\r_R$ (see Figure \ref{fig:phi2}). Since we are working in $\phi_b\ra \infty$ limit, each contraction in Figure \ref{fig:phi2} evaluates as $e^{-\phi_b I}$ for some on-shell action $I$. Note that the last diagram in Figure \ref{fig:phi2}  with crossing Wick contraction is more involved than the first two because it contains $sl(2)$ 6j symbol \cite{Mertens:2017mtv,Suh:2020lco} as we show in Appendix \ref{app:1} using a straightforward Euclidean path integral. 

\begin{figure}
\begin{centering}
\includegraphics[width=10cm]{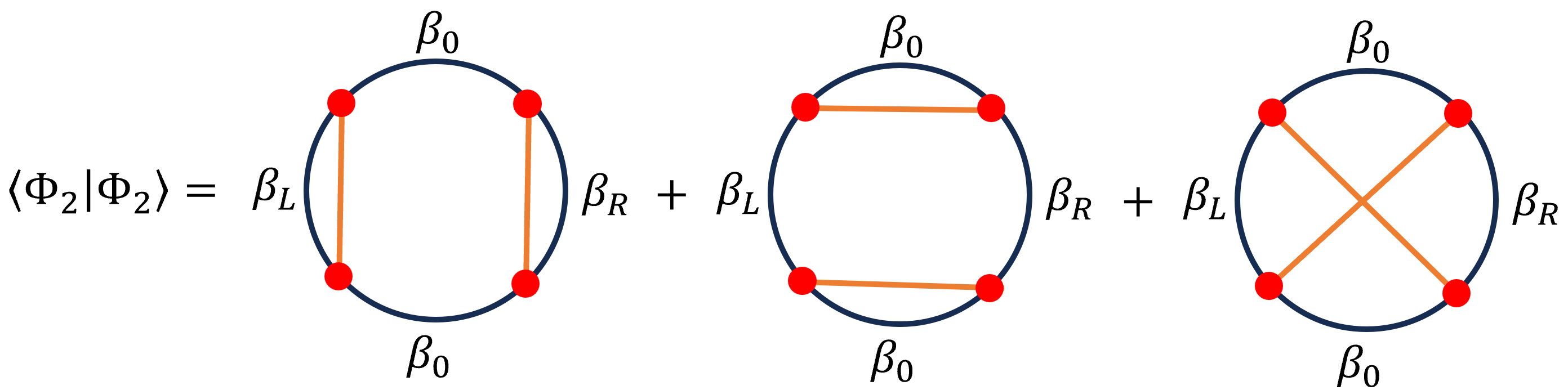}
\par\end{centering} 
\caption{The norm of state $\bra{\Phi_2}$}\label{fig:phi2}
\end{figure}

For generic $\b_{L,R,0}$, this exponential factor will pick out the dominant saddle and lead to a dominant geometry corresponding to this state. It has been shown in \cite{Mertens:2017mtv} that the contraction with crossing (the last diagram in Figure \ref{fig:phi2}) is subleading to the non-crossing diagrams for a four-point function because the former can be regarded as OTOC and the latter is TOC. It is generally true that OTOC is no greater than TOC. In Appendix \ref{app:2}, we show that for a generic Euclidean path integral, the crossing Wick contraction of any two pairs of heavy operators of the same type is always exponentially suppressed relative to a non-crossing Wick contraction in the semiclassical limit. Note that the first two diagrams in Figure \ref{fig:phi2} have different geometries in the semiclassical limit. The first diagram has a long wormhole geometry and one can show that the global minimal dilaton is at the center. The second diagram has a short wormhole similar to a thermofield double state. Depending on the parameters $\b_{L,R,0}$, one of these two diagrams dominates over the other and defines a geometric background dual to $\bra{\Phi_2}$. The detailed analysis is given in Appendix \ref{app:3}.

In the special case with  $\b_L=\b_R=\b_0$, we see that the first two diagrams in Figure \ref{fig:phi2} are identical but with a relative rotation of 90 degrees. Therefore, they evaluate to exactly the same on-shell action though the Lorentzian geometries are completely different. Such a fine-tuned state does not correspond to a unique geometric background, which leads to the ambiguity of defining a bulk field $\p$. Therefore, we may say it does not have a sensible semiclassical limit. We will exclude the fine-tuned states in the rest of this paper.

Though in Appendix \ref{app:2} we excluded the crossing contraction of the same type of heavy operators in the leading saddle, there could be states, whose leading saddle has crossing contractions between different types of heavy operators. For example, consider the state 
\be 
\bra{\Phi_3}=e^{-\b_L H/2}\mO_1 e^{-\b_1 H}  \mO_2 e^{-\b_2 H} \mO_1 e^{-\b_R H/2}\bra{v}\bra{0}
\ee
where $\b_1,\b_2\ll \b_{L,R}$. As the two $\mO_1$ are close to each other on the Euclidean boundary, we may expect the leading saddle of $\avg{\Phi_3|\Phi_3}$ to includes their contractions within the bra and ket states. On the other hand, $\mO_2$ must contract between the bra and ket states, which leads to two crossing with the contractions of $\mO_1$. To analyze the semiclassical limit of this state, we need to study the semiclassical limit of 6j symbols with two different conformal weights.\footnote{In \cite{Jafferis:2022wez} the semiclassical limit of 6j symbol with two identical conformal weight is studied.} This is outside of the scope of this paper and we will only consider the states without crossing contractions between different types of operators.

We may conclude that a non-fine-tuned PETS $\bra{\Psi}$ with only heavy operator insertions and no crossing contractions corresponds to a fixed Euclidean geometry, which is determined by a unique leading saddle point in the semiclassical limit. For each such state, we can consider more states by inserting finite numbers of light operators $\mO_i^l$ on the Euclidean path. As these light operators contribute to the path integral only at $e^{-O(1)}$ order, they do not change the saddle and we can regard them as states with light matters in the same geometric background. Let us denote them as $\bra{\Psi;\mO^l_i\cdots}$ and normalize these states to one by dividing their $e^{-O(\phi_b)}$ norms. These states have nonzero overlap with each other, which is given by the correlation function of the light operators in the Euclidean geometry background $g$ corresponding to $\bra{\Psi}$
\be 
\avg{\Psi;\mO^{l}_i\cdots|\Psi;\mO^{l}_j\cdots}=\avg{\mO^l_i\cdots\mO^l_j\cdots}_g
\ee  
We may consider these states to span a Hilbert space $\mH_g$ labeled by the saddle Euclidean geometry $g$. For two normalized states corresponding to different geometries, their inner product obeys Schwarz inequality
\be  
\avg{\Psi_{g_1}|\Psi_{g_2}}=e^{-\phi_b (I_{12}+(I_1+I_2)/2)}\leq 1, \quad \bra{\Psi_{g_i}}\in\mH_{g_i} \label{eq:innerprod}
\ee 
where $I_{12}$, $I_1$ and $I_2$ are the corresponding on-shell actions of the Euclidean path integral. Since $g_1$ and $g_2$ are different geometries, the action difference must be $\sim O(1)$. In the semiclassical limit, the exponential suppression \eqref{eq:innerprod} vanishes and leads to the orthogonality of two states with distinct geometries. Therefore, the full Hilbert space in the semiclassical limit contains the following direct sum structure
\be 
\mH\supset\oplus_g \mH_g \label{eq:Hib}
\ee 

For a state $\bra{\Psi}\in\mH_g$, the CFT data of the light operators $\mO_{i}^{l}$ on the Euclidean boundary defines a bulk field theory of $\p_i$ in AdS$_2$ regarding $\mO_{i}^{l}$ as the boundary limit of $\p_i$ in the sense of \eqref{eq:2.5}. The mass of bulk fields $\p_i$ is given by the standard relation in AdS/CFT (e.g. conformal dimension $\D_i=\f 1 2 +\sqrt{\f 1 4+m_i^2}$ for scalars) and their couplings are proportional to the OPE coefficient $C_{lll}$. For such state $\bra{\Psi}$, we can consider the Lorentzian continuation along the geodesic connecting the left and right endpoints of $\bra{\Psi}$. This geodesic is the time-reversal symmetric global Cauchy slice $\S_g$ of the Lorentzian spacetime. Therefore, $\mH_g$ can also be identified as the Hilbert space of bulk fields $\p_i$ canonically quantized on $\S_g$. 

\begin{figure}
\begin{centering}
\subfloat[]{\begin{centering}
\includegraphics[width=5cm]{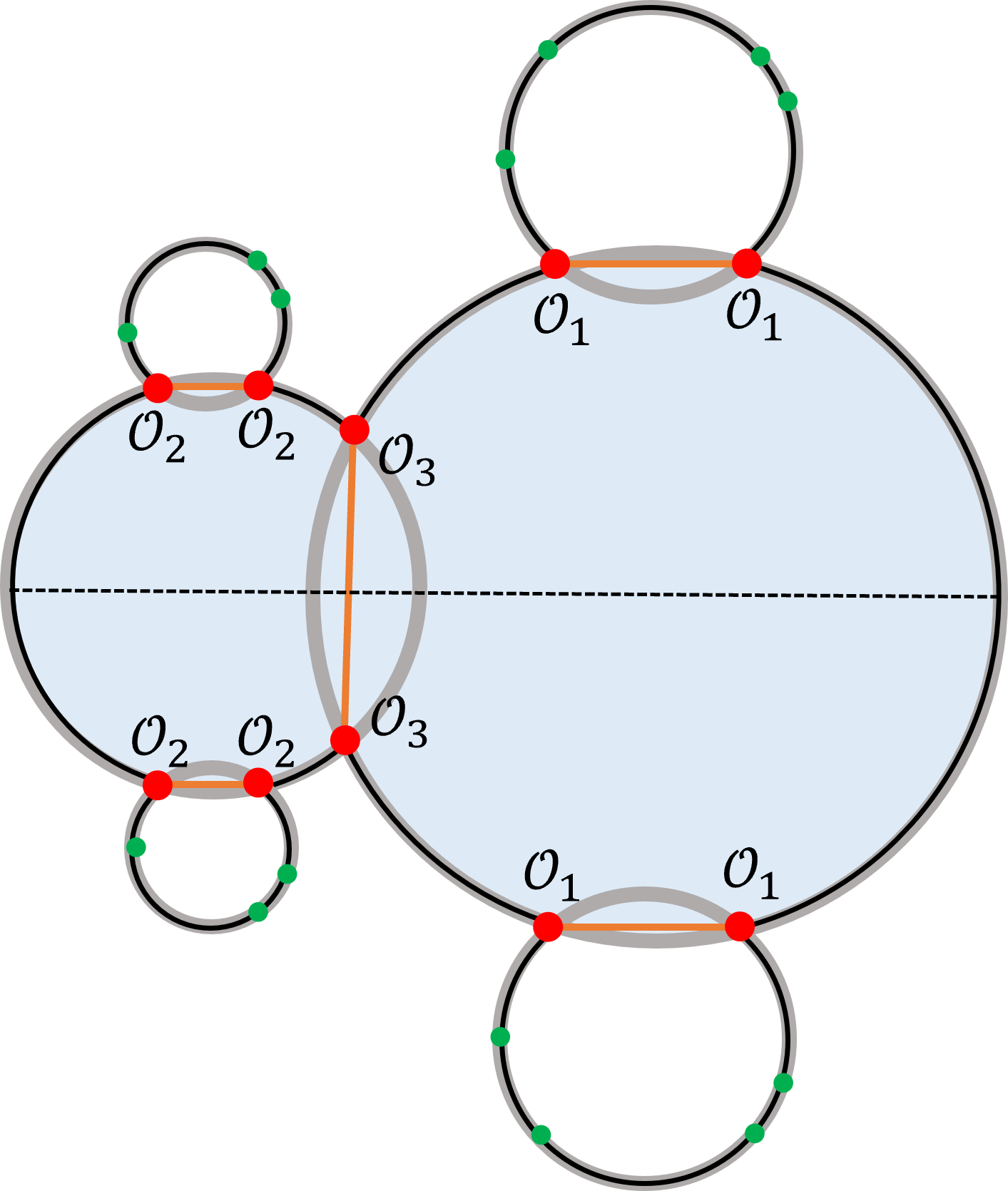}\label{fig:dif1}
\par\end{centering}

}\subfloat[]{\begin{centering}
\includegraphics[width=5cm, trim = 0cm -2.45cm 0cm 0cm]{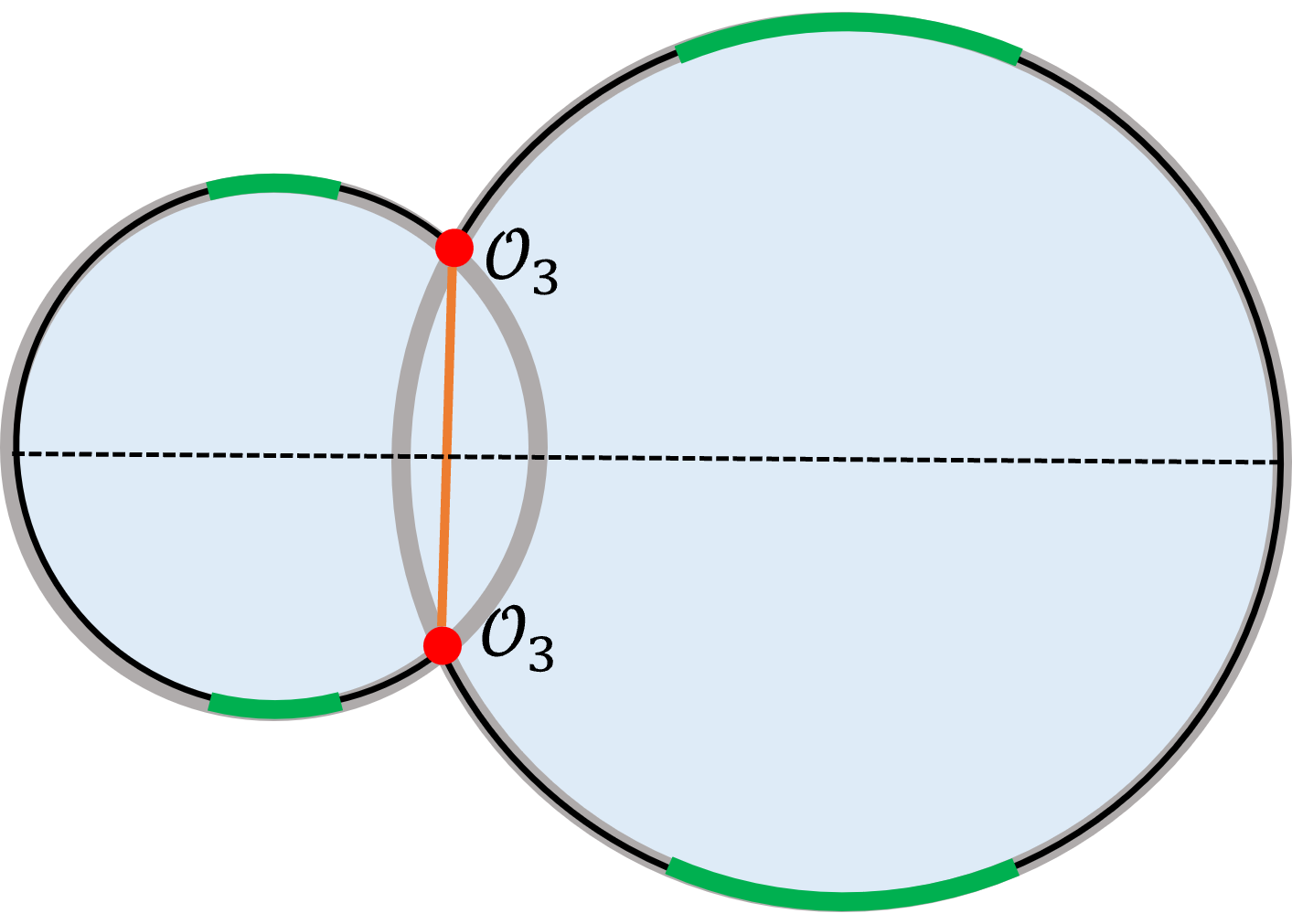}\label{fig:dif2}
\par\end{centering}
}
\par\end{centering} 
\caption{(a) and (b) are two states with different Euclidean geometries, where red dots are heavy operators and green dots/arcs are light operators. The state in (a) has multiple identical heavy operators, which lead to disks above and below the central row disks, which are filled with blue color. However, the central row disks of (a) and (b) are the same. Therefore, when we analytically continue the Euclidean geometry along the horizontal dashed line, they lead to the same Lorentzian geometry. The light operators on the off-central row disks (green dots) in (a) is equivalent to some integral of light operators on the central row disk (green arcs) in (b) if the integral measure is chosen properly. \label{fig:dif}}
\end{figure}

One should note that our Hilbert space is labeled by the Euclidean geometry rather than the Lorentzian geometry. This is because Lorentzian geometry only depends on the central row of EAdS$_2$ disks in the Euclidean saddle geometry in the norm of the state but does not depend on the disks above or below the central row (the blue-shaded disks in Figure \ref{fig:dif}). If the PETS only contains different heavy operators, the saddle geometry only has the central row and the Lorentzian geometry corresponds to the unique Euclidean geometry. However, when a PETS has two identical operators $\mO$, the saddle of the partition function could be the configuration with these two $\mO$'s self-contracted rather than contracting between the bra and ket states. In this case, disks above and below the central row do not affect the Lorentzian continuation, which is along the horizontal dashed line in Figure \ref{fig:dif}. In particular, in the Lorentzian spacetime, bulk fields $\p_i$ cannot probe the existence of the disks below the central row because we can equivalently replace the Euclidean boundary condition of $\p_i$ on these disks by some conditions above the geodesic connecting the two self-contracting $\mO$'s. Figure \ref{fig:dif} gives a pictorial explanation for this point by replacing the light operators on the off-central row disks (green dots) in Figure \ref{fig:dif1} with some integral of light operators on the central row disk (green arcs) in Figure \ref{fig:dif2}. If we choose the integral measure properly, they should lead to the same bulk field configuration of $\p_i$ on the horizontal Cauchy slice.

The argument before shows that a non-fined-tuned PETS without crossing contractions has a unique  Euclidean saddle geometry. However, the inverse statement is more subtle. Let us consider two PETS containing the same number of heavy operators at the same location on the Euclidean path. The only difference is that the heavy operators in the two PETS have $O(1)$ difference in their conformal dimensions. Note that to derive the saddle equation in \eqref{eq:2.40} we need to rescale $\mu\ra 2\mu\phi_b$. Therefore, the $O(1)$ difference in the conformal dimension does not change the saddle and leads to the same geometry. However, since they are different heavy operators, the generalized free field assumption implies zero inner product between these two PETS. Moreover, as we assume the light operators are decoupled from heavy operators, the inner product between these two PETS with any additional light operator insertions always vanishes. Therefore, the Hilbert space $\mH_g$ should be further decomposed as direct sum of $\mH_g^i$ where $i$ indexes the heavy operator variations in the spectrum of width $O(1)$. We will not explore this further in this paper but it is clear that the fine structure of this decomposition of Hilbert space depends on the spectrum of heavy operators.

The spectrum of heavy operators affects another subtle issue of the structure of $\mH_g$. Suppose a PETS $\bra{\Psi}\in\mH_g$ contains two identical heavy operators $\mO$ with saddle geometry with these two $\mO$'s self-contracted. In the language of Figure \ref{fig:dif1}, the norm of this state has disks above and below the central row. We can consider another PETS $\bra{\Psi'}$ by replacing $\mO$ by $\mO'$ with $O(1)$ conformal dimension difference. By the same argument, these two states correspond to the same Euclidean geometry. However, since these operators are self-contracted, the inner product $\avg{\Psi|\Psi'}$ is $O(1)$ because we are free from the selection rule of generalized free field theory. Inserting light operators in the Euclidean path of $\bra\Psi$ and $\bra{\Psi'}$ generate more states. However, there is an isomorphism mapping the states generated from $\bra\Psi$ to $\bra{\Psi'}$ because
\be 
\avg{\Psi;\mO^{l}_i\cdots|\Psi;\mO^{l}_j\cdots}=\avg{\Psi';\mO^{l}_i\cdots|\Psi';\mO^{l}_j\cdots}=\avg{\Psi|\Psi'}\avg{\mO^{l}_i\cdots\mO^{l}_j\cdots}_g
\ee
where $\avg{\mO^l_i\cdots\mO^l_j\cdots}_g$ is the CFT correlation function of light operators on the boundary of the Euclidean geometry $g$. Here the light operators $\mO^l_{i,j}$ are inserted at the same locations in $\bra\Psi$ and $\bra{\Psi'}$. Therefore, the Hilbert space $\mH_g$ has a tensor product structure
\be 
\mH_g\simeq \mH^{\text{heavy}}_g\otimes\tilde{\mH}_g
\ee 
where $\mH^{\text{heavy}}_g$ is spanned by the states with different heavy operators with $O(1)$ conformal weight difference but corresponding to the same Euclidean geometry $g$. Its structure depends on the $O(1)$ width spectrum of heavy operators. The $\tilde \mH_g$ factor is isomorphic to the Hilbert space generated by any state picking from $\mH^{\text{heavy}}_g$. In the bulk viewpoint, it is the Hilbert space of all bulk fields in the geometric background $g$.

In conclusion, the semiclassical Hilbert space $\mH$ of JT gravity with matter theory contains a direct sum over sub-Hilbert spaces $\mH_g$ labeled by the saddle Euclidean geometry $g$ as given in \eqref{eq:Hib}. Depending on the spectrum of heavy operators, for each sub-Hilbert space $\mH_g$, it may be decomposed to more sectors, which share the same Euclidean geometry.

\subsection{Operator algebra in the semiclassical limit} \label{sec:2.5}

As we have seen before, all states in the Hilbert space survive in the semiclassical state and each of them is related to an Euclidean geometry with light matter fields. On the other hand, not all operators in the left and right algebra $\mA_{L,R}$ survive in the semiclassical limit because some of them will have singular matrix elements.

It is straightforward to check that all heavy operators become ill-defined. For example, the matrix element in the second equation of \eqref{26} becomes singular when we take $\mu_i\sim O(\phi_b)\ra\infty$ because $\G^{\mu}_{k,k'}\sim \mu^{2\mu-3/2}$. Moreover, the Hamiltonian $H_{L,R}$ are singular when they act on a PETS. Note that the saddle value of $k \sim O(\phi_b)$ in \eqref{eq:2.40}, which leads to boundary energy $E_{L,R}=k_{L,R}^2/(4\phi_b)\sim O(\phi_b)\ra\infty$ for any PETS in the semiclassical limit. 

As for light operators, it is more subtle to discuss their matrix elements in the basis of \eqref{26} and \eqref{27}, because of the decoupling between ``single trace" heavy and light operators in the semiclassical limit. It is more convenient to confirm the survival of light operators in $\mA_{L,R}$ in each sub-Hilbert space $\mH_g$ because multi-point correlation functions of light operators are finite in any state $\bra{\Psi}\in\mH_g$. Moreover, even though $H_{L,R}$ becomes singular in the semiclassical limit, the light operators $\mO^l_{i,L/R}(u)=e^{iH_{L/R}u}\mO^l_{i,L/R}e^{-iH_{L/R}u}$ with $u\in\R$ in Heisenberg picture are still bounded operators because the unitary $e^{-iH_{L/R}u}$ for real $u$ is  bounded.\footnote{Note that $e^{-iH_{L/R}u}$ for finite $u$ itself may not have a semiclassical limit because of the highly oscillation with frequency $\sim O(\phi_b)$.} Its matrix elements are given by analytic continuation of some correlation functions of light operators in the background $g$.

By our assumption that ``single trace" heavy operators are decoupled from light operators in the semiclassical limit, light operators $\mO^l_{i,L/R}(u)$ for all $u\in \R$ in $\mA_{L/R}$ should generate a closed algebra, which we call $\mA_{L/R}^0$.\footnote{Due to the short distance singularity, to define this algebra property, we need to smear each $\mO^l_{L/R,i}(t)$ appropriately along the Lorentzian time $u$.} Note that these two algebras are {\it emergent} under semiclassical limit, because for finite $\phi_b$ light operators do not close as an algebra. They are universal emergent subalgebras of $\mA_{L/R}$ and act on any sub-Hilbert space $\mH_g$ in the semiclassical limit. This is similar to the single trace operators in a holographic theory, in which they form a universal emergent algebra under $N\ra \infty$ limit independent of the state \cite{Leutheusser:2021frk}. 

In the Lorentzian signature, the algebra $\mA_{L/R}^0$ are the light operators living on the left and right boundary at all times respectively. For the left or right boundary in the geometry of Figure \ref{fig:Geometry-of-boundary.}, the leftmost or rightmost disk with metric \eqref{eq:2.47} after analytic continuation of $\tau\ra it$ becomes
\begin{equation}
ds^{2}=d\r_{L/R}^{2}-\sinh^{2}\r_{L/R} dt_{L/R}^{2}
\end{equation}
where $t_{L/R}$ is the Rindler time for the left or right boundary, which is related to boundary physical time by $u_{{L/R}}(t_{L/R})=(2\pi)^{-1}b_{L/R} t_{L/R}$ with $b_L=b_0$ and $b_R=b_p$, the periodicity of Euclidean boundary physical times for $\bra{\Psi_p}$. By the well-known HKLL reconstruction \cite{Hamilton:2005ju,Hamilton:2006az,Hamilton:2006fh,Hamilton:2007wj}, we can extend a boundary light operator $\mO^l_{i,L/R}(u)$ to a dual bulk field in this causal wedge
\begin{equation}
\p_{i,L/R}(\r_{L/R},t_{L/R})=\int dt K_{\D_i}(\r_{L/R},t_{L/R};t)\mO^l_{i,L/R}(u_{L/R}(t)) \label{eq:hkll}
\end{equation}
where $\D_i$ is the conformal dimension of $\mO^l_{i,L/R}$, and the kernel only depends on the geometry of the causal wedge. This reconstruction serves as a duality between $\mA^0_{L/R}$ and the algebra of bulk fields in the causal wedge $\mA^{CW}_{L/R}$. By the general argument for bulk QFT for a subregion, this emergent algebra $\mA^{CW}_{L/R}\simeq \mA^0_{L/R}$ is type III$_1$.

An immediate limitation is that HKLL does not work beyond the causal wedge. As we see from Figure \ref{fig:Geometry-of-boundary.}, the Lorentzian continuation of a generic PETS $\bra{\Psi}$ corresponds to a long wormhole connecting the left and right boundary. The sub-Hilbert space $\mH_g$ for this geometry consists of generic bulk fields living on the global Cauchy slice, which is clearly beyond the causal wedge. In other words, acting $\mA^0_{L/R}$ on a state in $\mH_g$ does not generate the whole $\mH_g$, and there must be some other operators in $\mA_{L/R}$ survival under the semiclassical limit. 

What is missed? In the next section, we will show that given a reduced density matrix $\r_R$ for a generic PETS, the modular flow of $\mA^0_{R}$ by $\r_R$ extends beyond the causal wedge and reconstruct the full right entanglement wedge, which is the causal diamond between the global minimal dilaton point (i.e. global RT or QES surface in higher dimensions) and the right boundary. The same entanglement wedge construction holds for the left side and altogether generates the bulk field algebra on the global Cauchy slice. Hence, by acting on any state in $\mH_g$, this modular flow extended algebra from both sides generates the whole $\mH_g$.

\section{Modular flow in the semiclassical limit} \label{sec:3}

Let us quickly review a few basic concepts in von Neumann algebra. Given a state $\bra{\Psi}$ in a Hilbert space $\mH$, assume there are two von Neumann algebras $\mA_L$ and $\mA_R$, which are commutant to each other. If $\bra{\Psi}$ cannot be annihilated by any nonzero operator in $\mA_L$ and $\mA_R$, and all operators in $\mA_L$ or $\mA_R$ acting on $\bra{\Psi}$ can generate a dense set in $\mH$, we call $\bra{\Psi}$ a cyclic and separating state. For such a state, Tomita-Takesaki theory states that there exists an antilinear operator $S_\Psi$ such that $S_\Psi a\bra{\Psi}=a^\dagger\bra{\Psi}$ for $a\in\mA_R$, and this operator has a unique polar decomposition $S_\Psi=J_\Psi \D_\Psi^{1/2}$, where $J_\Psi$ is an anti-unitary operator obeying $J_\Psi^2=J_\Psi J_\Psi^\dagger=1$, and $\D_\Psi$ is a hermitian and positive-definite operator. $J_\Psi$ and $\D_\Psi$ are called modular conjugation and modular operator respectively. Similarly, we can define $S_\Psi'$ relative to the commutant algebra $\mA_L$, and find the polar decomposition obeys $J_\Psi'=J_\Psi$ and $\D_\Psi'=\D_\Psi^{-1}$.  There are many nontrivial features of these two operators in the context of algebraic quantum field theory and quantum information theory, and we refer readers to Witten's note \cite{Witten:2018zxz} for a pedagogical review. 

For type III von Neumann algebra, the properties of the modular operator is crucial for the classification though we may not be able to write down the explicit form of $\D_\Psi$. For type I/II von Neumann algebra, thanks to the existence of the density matrix and trace, the modular operator has an explicit form $\D_\Psi=\r_L^{-1} \r_R$, where $\r_{L/R}$ is the reduced density matrix of $\bra{\Psi}$ for $\mA_{L/R}$ respectively. Tomita-Takesaki theory states that the modular flow of $\mA_{R}$ is an automorphism $\D_\Psi^{i s}\mA_R\D_\Psi^{-i s}=\mA_R$ for $s\in\R$. Indeed, this defines a 1-parameter group of modular automorphism $\D_\Psi^{is}a\D_\Psi^{-is}$ for each operator $a\in\mA_R$. This modular automorphism is also conventionally called modular flow. 

In this section, given a generic PETS $\bra{\Psi}$ we will study the modular flow of a right light operator in the semiclassical limit. As the algebra at finite $\phi_b$ is type II$_\infty$, the modular flow of $\mO^l_{j,R}(u)$ can be written in terms of the reduced density matrix
\be 
O^l_{j,R}(s;u)=\D_\Psi^{is}\mO^l_{j,R}(u)\D_\Psi^{-is}=\r_R^{is}\mO^l_{j,R}(u)\r_R^{-is}
\ee
where we used the fact that $\r_L$ commutes with any right operator. It has been argued in \cite{Faulkner:2017vdd} that the modular zero mode, namely integrating $O^l_{j,R}(s;u)$ over all $s\in\R$, approximately
commutes with both left and right bulk algebras after appropriate smearing.\footnote{Rigorously speaking, the modular zero mode is not an operator but a quadratic form because it has matrix elements that make sense, but it does not make sense as an operator, since acting on the vacuum it gives a ``state" that is not square-integrable. We thank Edward Witten for pointing this out to us.} Therefore, it is conjectured to exist an entanglement wedge reconstruction using the modular flow \cite{Jafferis:2015del, Faulkner:2017vdd}. However, as modular flow usually acts in a very nonlocal way except in very limited known cases, it is very hard to show the entanglement wedge reconstruction in general. In the following, we will show that the modular flow for a generic PETS in JT gravity acts geometrically in the semiclassical limit, and suffices to reconstruct the entanglement wedge manifestly.

\subsection{Modular flow by replica trick} \label{sec:3.1}

Let us consider a generic PETS $\bra{\Psi_p}$ in \eqref{eq:pets} with all different heavy operators. The corresponding reduced density matrix is \eqref{eq:rhopets}. As we explained in Section \ref{sec:2.4}, all states in the sub-Hilbert space $\mH_g$ corresponding to $\bra{\Psi_p}$ are spanned by the states with insertion of light operators on the Euclidean path of $\bra{\Psi_p}$. Since the von Neumann algebra is defined to be close under weak operator topology, which means closed in the sense of matrix element, we need to study the matrix element of a modular flowed operator
\be 
\avg{\Psi_p;\mO^l_i\cdots|O^l_{j,R}(s;u)|\Psi_p;\mO^l_k\cdots}
\ee
Due to the OPE expansion of light operators, it is sufficient to just consider the states with at most one light operator insertion. It boils down to the following matrix elements
\be 
\avg{\Psi_p|O^l_{j,R}(s;u)|\Psi_p},\quad \avg{\Psi_p;\mO^l_i|O^l_{j,R}(s;u)|\Psi_p},\quad \avg{\Psi_p;\mO^l_i|O^l_{j,R}(s;u)|\Psi_p;\mO^l_k} \label{eq:3.3}
\ee 
where we are sloppy about the location of the light operators $\mO^l_i$ and $\mO^l_k$.

The modular flow is hard to compute directly in the original real $s$ signature,
and we will use a replica trick \cite{Faulkner:2018faa} for modular flow in JT gravity similar to \cite{Gao:2021tzr} and compute
the following Euclidean version matrix elements
\begin{align}
g_j(n,m,u_E)&=\avg{\Psi_p|\r_R^{n-m-1}\mO^l_{j,R}(u_E)\r_R^m|\Psi_p} \label{eq:me1}\\
g_{ij}(n,m,u_E)&=\avg{\Psi_p;\mO^l_i|\r_R^{n-m-1}\mO^l_{j,R}(u_E)\r_R^m|\Psi_p} \label{eq:me2}\\
g_{ijk}(n,m,u_E)&=\avg{\Psi_p;\mO^l_i|\r_R^{n-m-1}\mO^l_{j,R}(u_E)\r_R^m|\Psi_p;\mO^l_k} \label{eq:me3}
\end{align}
for all integers $n\geq m+1\geq0$. Here we also analytically continued the Lorentzian time to Euclidean time $u\ra-iu_E$ and it is compatible with the integer powers of $\r_R$ because $\r_R$ is in the form of $e^{-\b_p H_R/2}\cdots e^{-\b_p H_R/2}$. In the end, we will analytically
continue $n\ra1$, $m\ra-is$, $u_E\ra i u$. The Euclidean replica matrix elements are much easier to compute because they are just $n$ copies of the Euclidean path integral in $\bra{\Psi_p}$. As $\mO^l$ are light operators, in the sense that their conformal
weight $\D/\phi_{b}\ra 0$ in the semiclassical limit, we should
ignore its back reaction to the geometry. Therefore, we just need
to solve the geometry of $n$ copies of $\r_{R}$ for all matrix elements \eqref{eq:me1}-\eqref{eq:me3} (see Figure \ref{fig:Replica-geometry.})

\begin{figure}
\begin{centering}
\includegraphics[width=6cm]{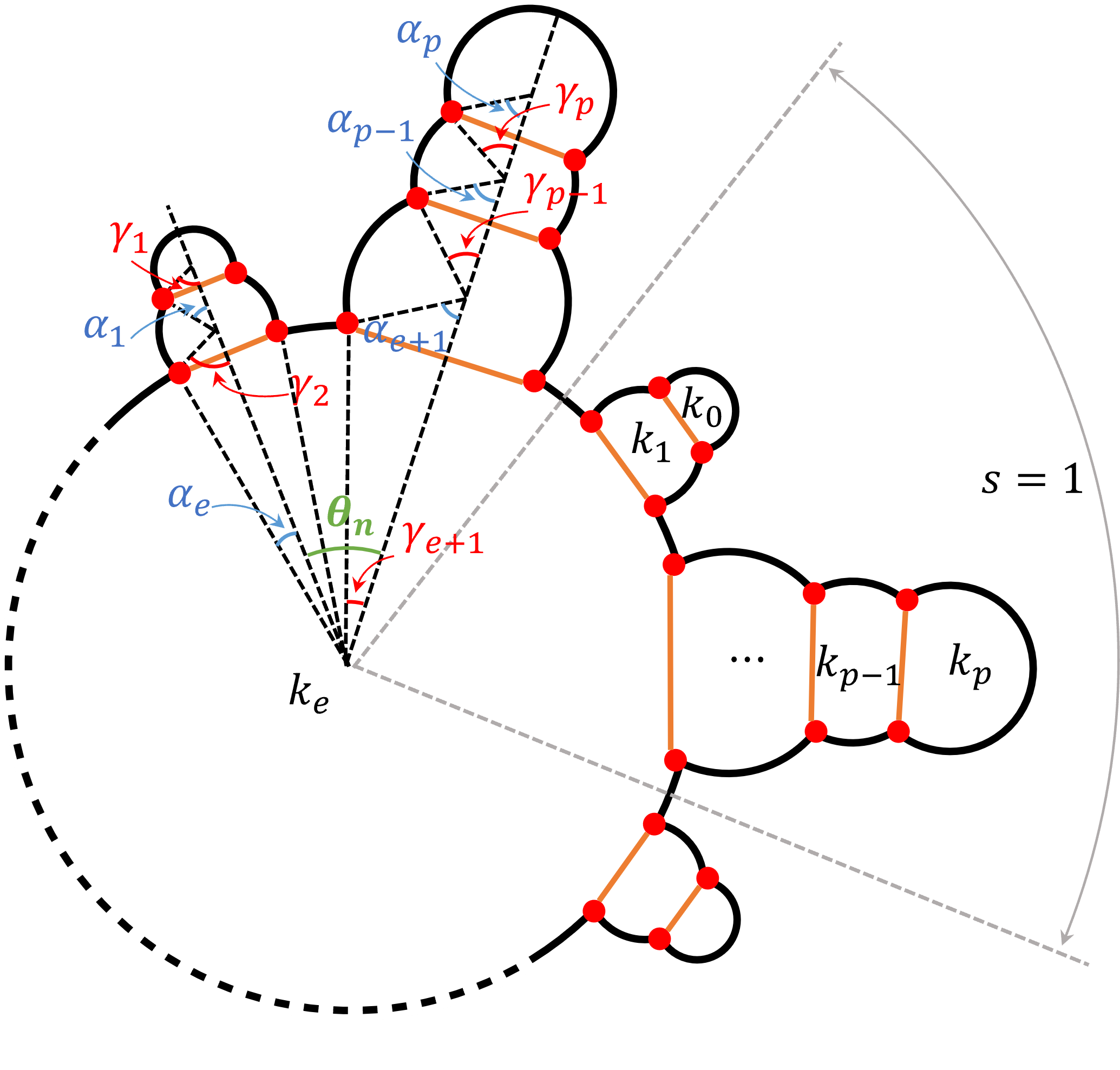}
\par\end{centering}
\caption{Replica geometry for $\protect\Tr \r_R^{n}$, where $ \t_{n}=\pi/n$. The gray region is labeled as the first replica $s=1$ (note that it does not represent the range for one copy of $\r_R$) and subsequent replicas are labeled anti-clockwise.}\label{fig:Replica-geometry.}
\end{figure}

For the $n$ replica partition function $\Tr\r_{R}^{n}$, there are $2n$
$\mO_{i}$ heavy operators in total for each $i$. As we assume these heavy operators become generalized free fields in the semiclassical limit, their contributions are sums of all Wick contractions. As shown in Appendix \ref{app:2}, the case with crossing Wick contractions of $\mO_{i}$ is subdominant to the case of no-crossing contractions. Furthermore, we assume that the leading order saddle respects $n$-replica symmetry. This means that each
$\mO_{i}$ can either contract with the other $\mO_{i}$ in the same copy
of $\r_{R}$ or another $\mO_{i}$ in a neighboring replica of $\r_{R}$.\footnote{There are two neighboring replicas, but the contraction with one of them preserves $n$-replica symmetry.} 
With these restrictions, we only have $p+1$ possibilities in the
$n$ replica partition function $Z_{p}^{(n)}=\sum_{e=0}^{p}Z_{e,p}^{(n)}$,
where
\begin{align}
&Z_{e,p}^{(n)} =\int dk_{e}\prod_{s=1}^{n}\left[\prod_{j\neq e}dk_{j}^{(s)}\prod_{i=1}^{p}d\l_{i}^{(s)}4^{\mu_i}Z_{s,e}(\{k,\l\})\right]\Psi_{n\b_{e}}(k_{e};\l_{e}^{(1)},\cdots,\l_{e}^{(n)},\l_{e+1}^{(1)},\cdots,\l_{e+1}^{(n)})\\
&Z_{s,e}(\{k,\l\})  =\Psi_{\b_{0}}(k_{0}^{(s)};\l_{1}^{(s)})\left(\prod_{i\neq e,p}e^{-\mu_{i}\l_{i}^{(s)}}\Psi_{\b_{i}}(k_{i}^{(s)};\l_{i}^{(s)},\l_{i+1}^{(s)})\right)e^{-\mu_{e}\l_{e}^{(s)}}e^{-\mu_{p}\l_{p}^{(s)}}\Psi_{\b_{p}}(k_{p}^{(s)};\l_{p}^{(s)})\\
&\Psi_{\b}(k;\{\l_{j}\})  =e^{-\b k^{2}/(4\phi_{b})}\r(k)\prod_{j}2K_{2ik}(4e^{-\l_{j}/2})
\end{align}
where $s$ labels different replicas. Using the integral representation
(\ref{eq:7}), and the rescaling $k_i^{(s)}\ra 2\phi_b k_i^{(s)},\mu_i\ra 2\phi_b \mu_i$ in the semiclassical limit, we have
\begin{align}
Z_{e,p}^{(n)} & =\int dk_{e}\prod_{s=1}^{n}\left(\prod_{i\neq e}dk_{i}^{(s)}\prod_{i=1}^{p}d\a_{i}^{(s)}d\g_{i}^{(s)}\right)\exp\left(-\sum_{s}I_{e,s}^{(n)}\right) \label{eq:3.10}\\
I_{e,s}^{(n)}/(2\phi_{b}) & \app\sum_{i=0}^{p}\left(\f{(k_{i}^{(s)})^{2}}2\b_{i}-\f{2\pi k_{i}^{(s)}}{1+(n-1)\d_{i,e}}\right)+2\sum_{i=1}^{p+1}k_{i-1}^{(s)}(\a_{i-1}^{(s)}+\g_{i}^{(s)})+\mu_{i}\log(\cos\a_{i}^{(s)}+\cos\g_{i}^{(s)})\label{eq:39}
\end{align}
where we have dropped an irrelevant constant $I_{n,p}(\mu_i)$.
Since different replicas do not interact with each other (the total action in \eqref{eq:3.10} is just a sum over replicas), it is natural to assume the dominant saddle
of the above action should also respect replica symmetry so that all variables
with different $s$ subscripts should be equal. For simplicity, in
the following we will omit the subscript $s$. Variation of $k_{i}$
gives
\begin{equation}
k_{i}\b_{i}=2(\pi/(1+(n-1)\d_{i,e})-\a_{i}-\g_{i+1})\label{eq:40}
\end{equation}
which is the same as (\ref{eq:12-1}) except $i=e$. However, this
exactly matches with the replica geometry in Figure \ref{fig:Replica-geometry.}
if we use the dictionary (\ref{eq:27-1}). Indeed, from Figure \ref{fig:Replica-geometry.},
the circumference $b_{e}$ of the central circle is given by
\begin{equation}
b_{e}=2\pi\b_{e}/(\t_{n}-\a_{e}-\g_{e+1}),\qquad(\t_{n}\equiv\pi/n)
\end{equation}
which is consistent with (\ref{eq:40}). Variation of $\a_{i}$ and
$\g_{i}$ leads to the same equations as (\ref{eq:31}), which has
been shown to be equivalent to the charge conservation and joint condition
for $\mO_{i}$ in the geometric formalism in Section \ref{sec:2.3}.

Therefore, the replica geometry for each choice of $e$ is constructed
as follows. For the first replica ($s=1$, the gray region in Figure \ref{fig:Replica-geometry.}),
the circles are given by
\begin{equation}
Y_{(1)}^{i}(\tau)=\begin{cases}
M_{2}(z_{i}(n))\cdot Y(\r_{i,\e},\tau) & i\geq e\\
M_{3}(\t_{n})\cdot M_{2}(z_{i}(n))\cdot Y(\r_{i,\e},\tau) & i<e
\end{cases}
\end{equation}
where $z_{e}(n)=0$. Then for the $s$-th replica, the circles are
given by
\begin{equation}
Y_{(s)}^{i}(\tau)=M_{3}(2(s-1)\t_{n})\cdot Y_{(1)}^{i}(\tau)
\end{equation}

\begin{figure}
\begin{centering}
\subfloat[]{\begin{centering}
\includegraphics[width=7cm]{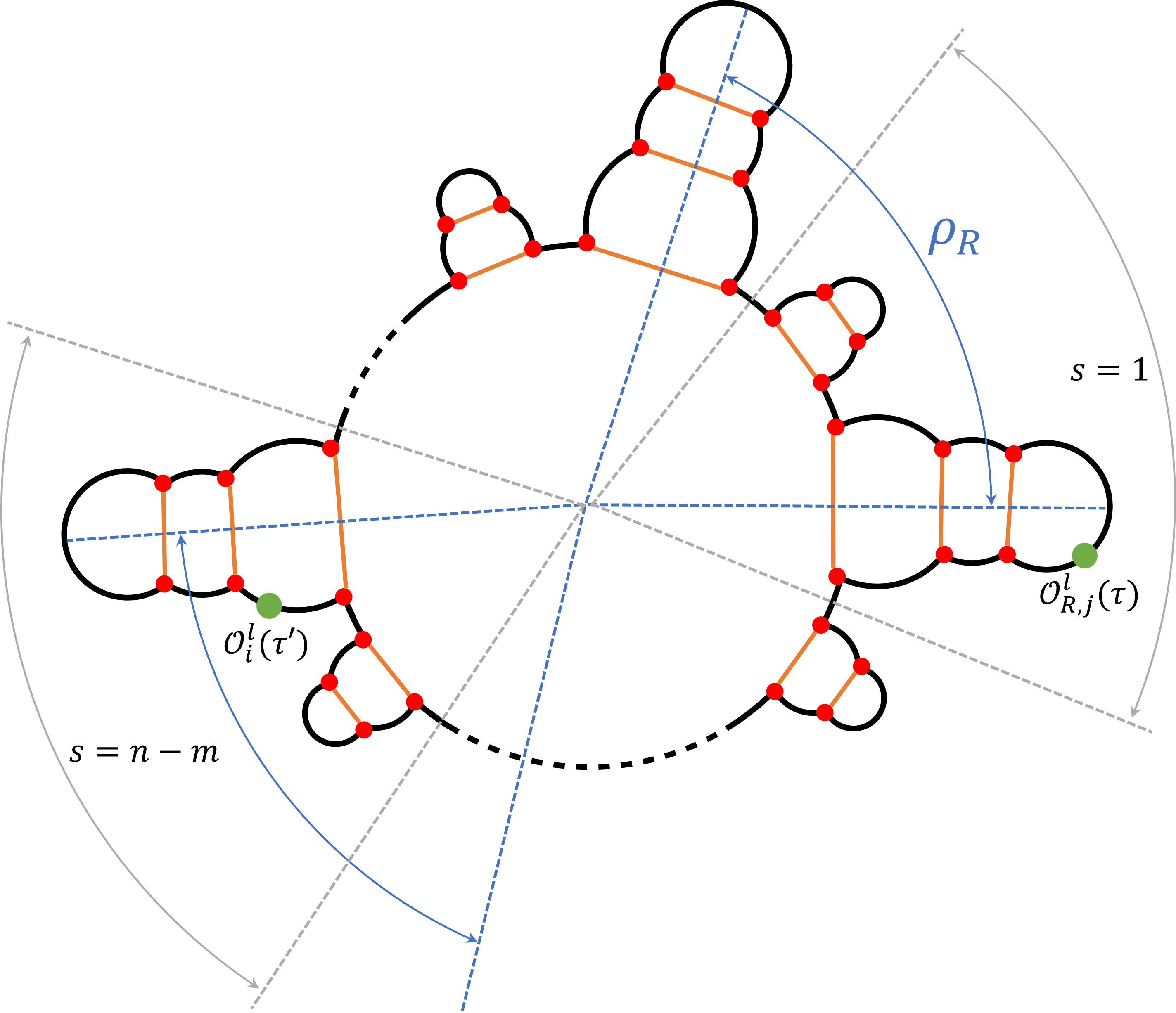}\label{fig:rep2p}
\par\end{centering}

}\subfloat[]{\begin{centering}
\includegraphics[width=7cm]{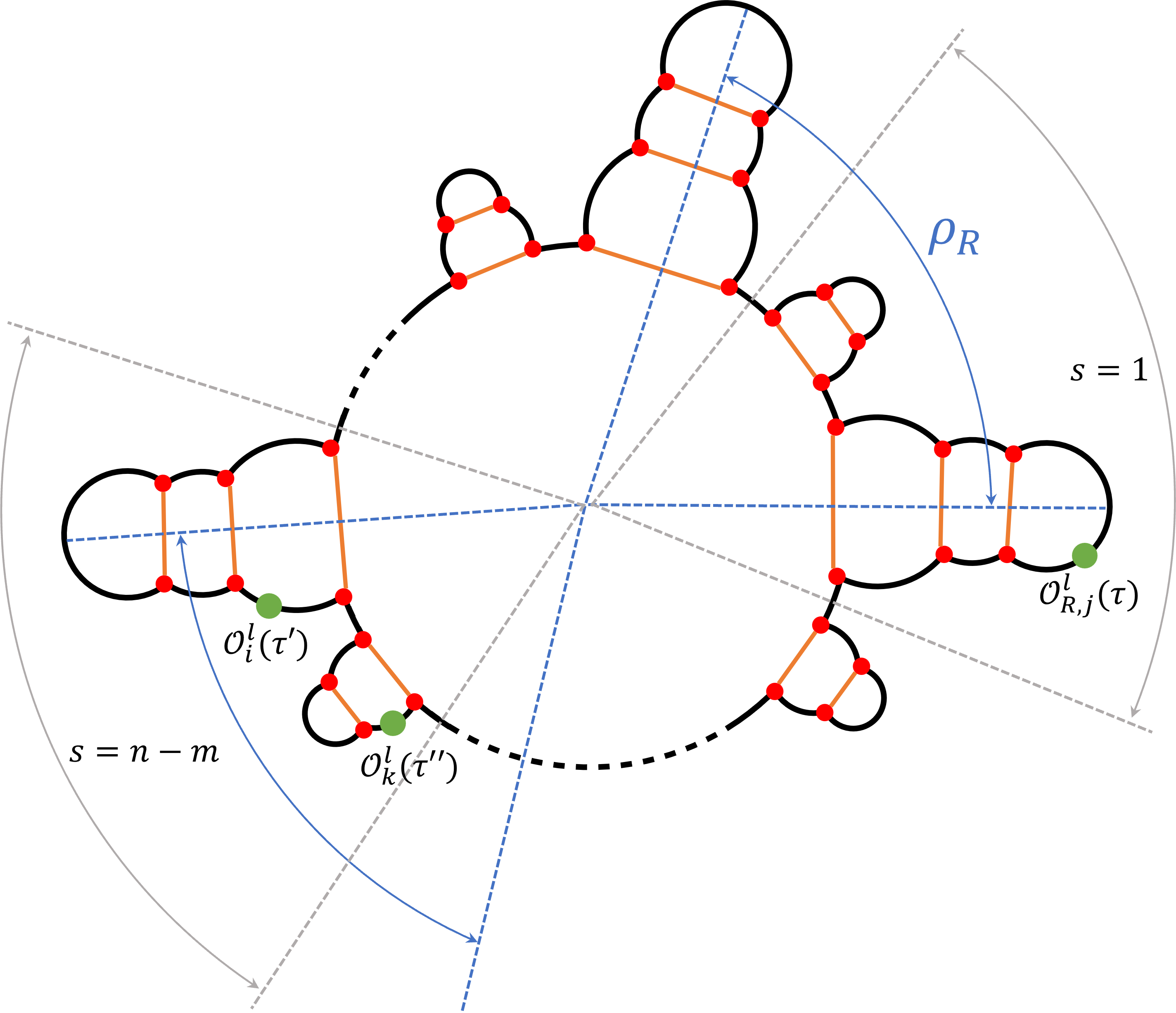}\label{fig:rep3p}
\par\end{centering}

}
\par\end{centering} 
\caption{(a) The configuration for the replica two-point function $g_{ij}(n,m,u_E)$. (b) The configuration for the replica three-point function $g_{ijk}(n,m,u_E)$. The green dots in both plots are the light operators and the blue range is for one copy of $\r_R$. \label{fig:rep}}
\end{figure}

Given this geometry, the matrix elements \eqref{eq:me1}-\eqref{eq:me3} become one-, two- and three-point functions of light operators in this Euclidean geometric background respectively. Recall that the light operators are decoupled from heavy operators and dilaton field in the semiclassical limit, and thus the correlations of light operators in a generic PETS enjoy an $SL(2)$ symmetry. We can assume this CFT of light operators has a trivial one-point function, which implies vanishing \eqref{eq:me1}. For \eqref{eq:me2}, it is a two-point function given by 
\begin{equation}
g_{ij}(n,m,u_E)=\d_{ij}(\cosh D)^{-\D}
\end{equation}
where $D$ is the geodesic distance between $\mO_{j,R}^l(u_E)$ and
$\mO^l_{i}$ on the replica diagram in Figure \ref{fig:rep2p}. For a simpler notation, let us use $\mO_{j,R}^l(\tau)$ momentarily with $\tau=2\pi u_E/b_p$ in terms of the EAdS$_2$ boundary coordinate. Let us assume $\mO^l_i$ is on the $\a$-th disk in $\bra{\Psi_p}$ with boundary coordinate $\tau'$. By the operator $\r_R^{n-m-1}$ between $\ket{\Psi_p;\mO^l_i}$ and $\mO^l_{j,R}(\tau)$, we may put $\mO^l_{j,R}(\tau)$ on the first replica and $\mO^l_i(\tau')$ on $(n-m)$-th replica. Since the embedding space is hyperbolic, the geodesic distance can be simply expressed as
\begin{equation}
\cosh D=-Y_{(n-m)}^{\a}(\tau')\cdot Y_{(1)}^{p}(\tau)=-Y_{(1)}^{p}(\tau')\cdot M_{3}(2(1-n+m)\t_{n})\cdot Y_{(1)}^{p}(\tau)
\end{equation}
Analytic continuation to $n\ra1$, $m\ra-is$, and
$\tau\ra it=i(2\pi u/b_p)$ leads to
\begin{equation}
\avg{\Psi_p;\mO^l_i|O^l_{j,R}(s;u)|\Psi_p}=\d_{ij}(-Y_{(1)}^{\a}(\tau')\cdot M_{3}(-2i\pi s)\cdot Y_{(1)}^{p}(it))^{-\D}\label{eq:46}
\end{equation}
Note that $Y_{(1)}^{p}$ is the right boundary coordinate in the geometry
of Lorentzian continuation of Figure \ref{fig:Geometry-of-boundary.}
because the saddle equations of the replica geometry reduce back to
the ones of a single copy when $n\ra1$. From (\ref{eq:46}) it is
clear that the modular flow acts just as an $SL(2)$ transformation on the embedding coordinate of $\mO_{j,R}^l(u)$
\begin{equation}
Y^{p}(it)\ra \tilde Y^{p}(s;it)\equiv M_{3}(-2i\pi s)\cdot Y^{p}(it) \label{eq:sl2}
\end{equation}

Similarly, for the replica three-point function \eqref{eq:me3}, the configuration is given in Figure \ref{fig:rep3p}. We assume $\mO^l_{k}$ is on the $\g$-th disk of $\bra{\Psi_p}$ with boundary coordinate $\tau''$. By our notation of replicas, there are two cases. If $\g\leq e$, it is on the $(n-m)$-th replica; if $\g>e$, it is on the $(n-m+1)$-th replica. For the first case, the three-point function is given by
\begin{align} 
g_{ijk}(n,m,u_E)=&C_{ijk}(-Y^\a_{(n-m)}(\tau')\cdot Y^p_{(1)}(\tau))^{-\f{\D_i+\D_j-\D_k}{2}}(-Y^\a_{(n-m)}(\tau')\cdot Y^\g_{(n-m)}(\tau''))^{-\f{\D_i+\D_k-\D_j}{2}}\nn\\
&\times(-Y^\g_{(n-m)}(\tau'')\cdot Y^p_{(1)}(\tau))^{-\f{\D_k+\D_j-\D_i}{2}}
\end{align} 
Analytic continuation to $n\ra1$, $m\ra-is$, and
$\tau\ra it$ leads to
\begin{align} 
\avg{\Psi_p;\mO^l_i|O^l_{j,R}(s;u)|\Psi_p;\mO^l_k}=& C_{ijk}(-Y^\a_{(1)}(\tau')\cdot \tilde Y^p_{(1)}(s;it))^{-\f{\D_i+\D_j-\D_k}{2}}(-Y^\a_{(1)}(\tau')\cdot Y^\g_{(1)}(\tau''))^{-\f{\D_i+\D_k-\D_j}{2}}\nn\\
&\times(-Y^\g_{(1)}(\tau'')\cdot \tilde Y^p_{(1)}(s;it))^{-\f{\D_k+\D_j-\D_i}{2}} \label{eq:3.21}
\end{align} 
which again justifies the $SL(2)$ transformation \eqref{eq:sl2}. For the second case, we need to change $Y^\g_{(n-m)}$ to $Y^\g_{(n-m+1)}$. Noting the periodicity $M_3(2\pi+x)=M_3(x)$, we have
\begin{align}
Y^\g_{(n-m+1)}(\tau'')\cdot Y^p_{(1)}(\tau)=Y^\g_{(1)}(\tau'')\cdot M_3(2m\t_n)\cdot Y^p_{(1)}(\tau)\\
Y^\a_{(n-m)}(\tau')\cdot Y^\g_{(n-m+1)}(\tau'')=Y^\a_{(1)}(\tau')\cdot M_3(2\t_n)\cdot Y^\g_{(1)}(\tau'')
\end{align} 
which implies the same equation \eqref{eq:3.21} after analytic continuation.

\begin{figure}
\begin{centering}
\includegraphics[height=5cm]{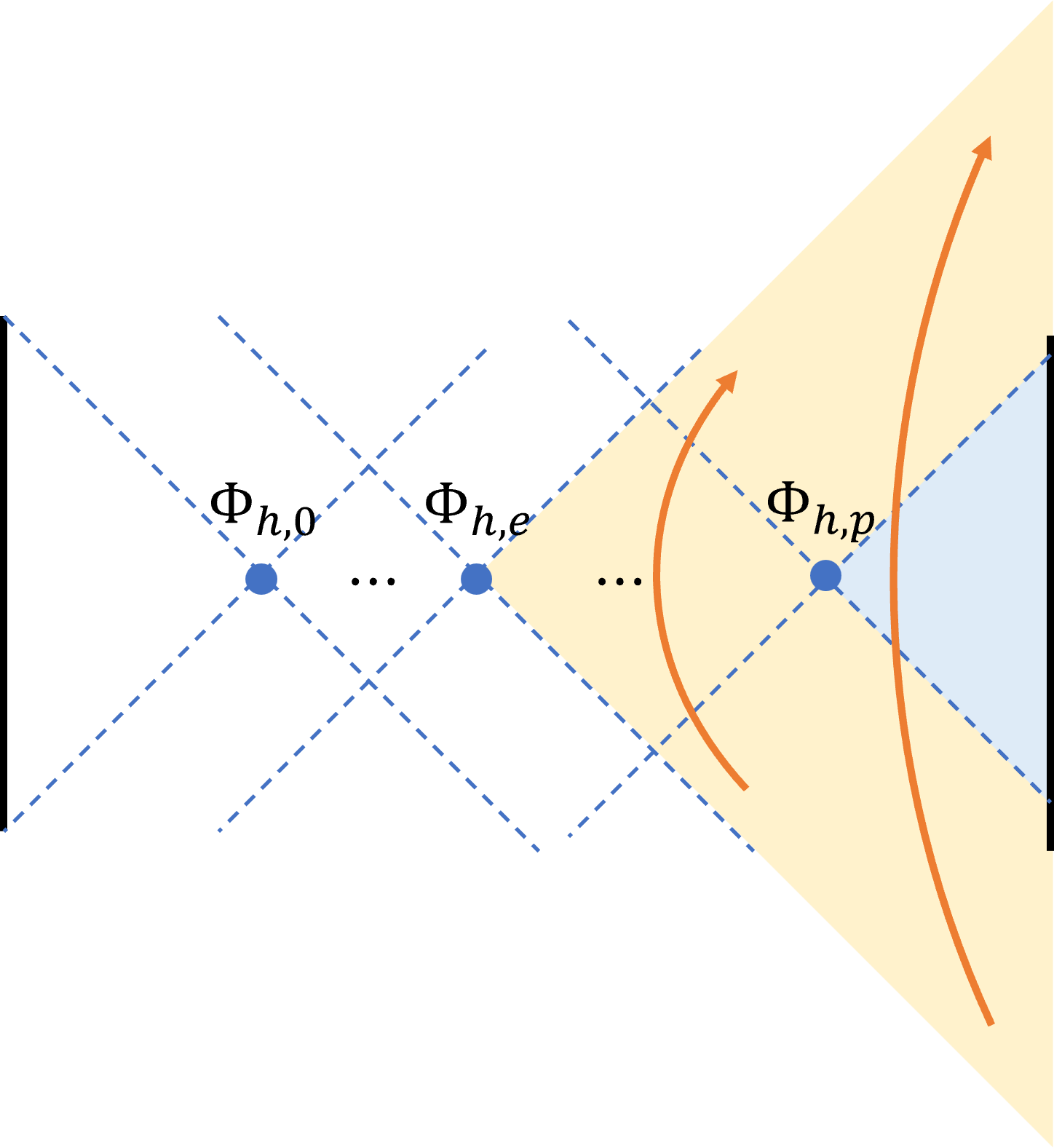}
\par\end{centering}
\caption{For each $e$-th saddle, the modular flow is the boost transformation around the local minimal dilaton point at $\Phi_{h,e}$. The dominant saddle over the $p+1$ saddles is the one corresponding to the global minimal dilaton point $\Phi_{h,e^*}$.}\label{fig:5}
\end{figure}

The physical meaning of this $SL(2)$ is simple. In EAdS$_{2}$, $M_{3}(x)$
rotates the angular coordinate $\tau$ of the disk at the center labeled
by $k_{e}$ in Figure \ref{fig:Replica-geometry.}. The analytic continuation
of $\tau\ra it$ changes EAdS to AdS with metric 
\begin{equation}
ds^{2}=d\r^{2}-\sinh^{2}\r dt^{2}
\end{equation}
for the Rindler wedge corresponding to $e$-th disk. Therefore, $M_{3}(-2i\pi s)$ is the boost
transformation around a local minimal point of dilaton $\Phi_{h,e}$
(see Figure \ref{fig:5}).

This $SL(2)$ transformation of boundary coordinate can be uniquely extended into the bulk using the HKLL reconstruction \eqref{eq:hkll}. For the bulk field $\p_i$ in the right causal wedge, the modular flowed field is given as
\be 
\p_j(s;\r,t)\equiv\r_R^{is}\p_j(\r,t)\r_R^{-is}=\int dt' K_{\D_j}(\r,t;t')O^l_{j,R}(s;u(t'))
\ee 
The matrix elements of $\p_j(s;\r,t)$  are similar to \eqref{eq:3.3} by integrating $O^l_{j,R}(t')(s;t')$ against the HKLL kernel $K_{\D_j}(\r,t;t')$, for example, 
\begin{align} 
\avg{\Psi_p;\mO^l_i|\p_{j}(s;\r,t)|\Psi_p}&=\d_{ij}\int dt' K_{\D_j}(\r,t;t') (-Y_{(1)}^{\a}(\tau')\cdot M_{3}(-2i\pi s)\cdot Y_{(1)}^{p}(it'))^{-\D} \nn\\
&=\d_{ij}\int dt' K_{\D_j}(\r,t;t') (-\tilde Y^{\a}(-s;\tau')\cdot Y^{p}(it'))^{-\D} \nn\\
&=\avg{\Psi_p;\tilde \mO^l_i|\p_{j}(\r,t)|\Psi_p} \label{eq:3.26}
\end{align} 
where $\ket{\Psi_p;\tilde \mO^l_i}$ is a state with $\mO^l_i$ inserted at a deformed location $\tilde Y^\a(-s;\tau')$. On the other hand, with the HKLL kernel, we should be able to compute the matrix element $\avg{\Psi_p;\mO^l_i|\p_{j}(\r,t)|\Psi_p}$, which is in the form of a bulk-to-boundary propagator and is an $SL(2)$ invariant function of the locations of $\mO^l_j$ and $\p_j(\r,t)$ thanks to the isometry of AdS$_2$ and the decoupling from heavy operators and the dilaton field. Therefore, we can do a global $SL(2)$ transformation $M_3(-2i\pi s)$ to both coordinates in  \eqref{eq:3.26} and have
\be 
\avg{\Psi_p;\mO^l_i|\p_{j}(s;\r,t)|\Psi_p}=\avg{\Psi_p;\mO^l_i|\p_{j}(M_{3}(-2i\pi s)\cdot Y^{p}(\r,t))|\Psi_p}\label{eq:3.27}
\ee
Though we assume $\p_j(\r,t)$ as the bulk field in the right causal wedge, the modular transformation \eqref{eq:3.27} has a unique smooth extension beyond the right causal wedge because $M_3$ is well-defined in the whole global Lorentzian spacetime. The same computation applies to $\avg{\Psi_p;\mO^l_i|\p_{j}(s;\r,t)|\Psi_p;\mO^l_k}$ because the three-point function \eqref{eq:3.21} is also an $SL(2)$ invariant function, and we may conclude that the modular flow of a bulk field from the $e$-th saddle is just the boost transformation around the $e$-th minimal dilaton point.

Now let us determine which $e$ dominates. In semiclassical limit,
for different choices of $e$, only the one with minimal on-shell
action survives. Since the on-shell action is replica symmetric, it
suffices to show which $e$ leads to minimal 
\begin{equation}
I_{e}^{(n)}/(2\phi_{b})=\sum_{i=0}^{p}\left(\f{k_{i}^{2}}2\b_{i}-2\pi k_{i}\right)+2\pi(1-1/n)k_{e}+2\sum_{i=1}^{p+1}k_{i-1}(\a_{i-1}+\g_{i})+\mu_{i}\log(\cos\a_{i}+\cos\g_{i})\label{eq:51}
\end{equation}
where $k_{i},\a_{i},\g_{i}$ are functions of $n$ as solutions to
the saddle equations. Taking derivative to $n$ will have two pieces.
One acts explicitly on $(1-1/n)$ and the other acts only implicitly
on $k_{i},\a_{i},\g_{i}$, which are all proportional to the equations
of motion and thus vanish \cite{Goel:2018ubv}. It follows that
\begin{equation}
\del_{n}I_{e}^{(n)}/\phi_{b}=2\pi k_{e}(n)/n^{2}=2\pi\Phi_{h,e}(n)/n^{2}\label{eq:52}
\end{equation}
where $\Phi_{h,e}(n)$ is the dilaton minimal value on $e$-th circle.
As we will take analytic continuation $n\ra1$, we only need to find
the minimal action for $n$ near 1. When $n=1$, the action (\ref{eq:51})
and equations of motion are independent on the choice of $e$. The
derivative (\ref{eq:52}) under $n\ra1$ limit is minimal when $e$
is picked such that the dilaton on $e$-th disk is the global minimum in the $n=1$
geometry of Figure \ref{fig:Geometry-of-boundary.}.\footnote{In this paper we always assume there is a unique global minimal dilaton point. For the case with multiple global minimal dilaton point with the same dilaton value by fine-tuning, the RT surface is ambiguously defined and one needs to compare the subleading in $\phi_b$ part of bulk entanglement entropy to define the entanglement wedge in terms of the QES formula. But this is beyond the scope of this paper of strictly $\phi_b\ra \infty$ limit. } This is exactly
the same type of proof of RT formula \cite{Lewkowycz:2013nqa}. This implies that
the modular flow in the semiclassical limit is the boost transformation
around the global minimum of dilaton $\Phi_{h,e^*}$. As the von Neuman algebra is close in the weak operator topology, we may conclude in bulk operator sense
\begin{equation}
\p_{j}(s;\r,t)=\p_j(M^*_{3}(-2i\pi s)\cdot Y^{p}(\r,t))
\end{equation}
where $M^*_3(-2i\pi s)$ is the boost around $\Phi_{h,e^*}$.

This is consistent with the conjecture of the entanglement wedge reconstruction \cite{Faulkner:2017vdd,Jafferis:2015del} by modular flow because the entanglement wedge of the right boundary is defined to be the causal diamond between $\Phi_{h,e^*}$ and the right boundary. Starting from a bulk location $(\r_0,0)$ with $\r_0$ distance from $\Phi_{h,e^*}$, and extending $s$ to all real numbers, we can reconstruct the bulk field along a trajectory parameterized by $s$  arbitrarily close to the boundary of the right entanglement wedge (see the orange arrow in Figure \ref{fig:5}). This is equivalent to reconstructing all light operators on a new AdS ``boundary" at $(\r_0,t_{e^*}(s))$ corresponding to the $e^*$-th disk. Here $t_{e^*}(s)$ is the Lorentzian time continued from the $e^*$-th disk and parameterized by $s$. Then we can apply HKLL reconstruction again for this new boundary at $(\r_0,t_{e^*}(s))$ to reconstruct the light fields in the whole right Rindler wedge of $\Phi_{h,e^*}$, which is exactly the entanglement wedge of $\bra{\Psi_p}$. Equivalently, we can regard this as a special time tube theorem \cite{borchers1961vollstandigkeit,araki1963generalization,Strohmaier:2000ib,Strohmaier:2023hhy,Strohmaier:2023opz} for the right entanglement wedge.

\subsection{Extension of operator algebra by the modular flow} \label{sec:3.2}

As we explained in Section \eqref{sec:2.5}, the light operators living on the left or right boundary become an emergent III$_1$ von Neumann algebra $\mA^0_{L/R}$ dual to the left or right bulk causal wedge $\mA^{CW}_{L/R}$ in terms of HKLL reconstruction respectively. However, these two algebras do not generate the full sub-Hilbert space $\mH_g$ for a PETS state $\bra{\Psi}$ because it does not excite states with bulk fields living in the long wormhole between two causal wedges. 

Given a generic PETS $\bra{\Psi}$ with all different heavy operators, it can be regarded as the ``vacuum" state in $\mH_g$. The modular flow of this state for $\mA_R$ can be constructed in terms of reduced density matrix $\r_R$. This modular flow becomes the boost transformation around the global minimal dilaton point when acting on a light operator $a\in \mA^0_R$ in the semiclassical limit. Therefore, in the bulk algebra viewpoint, the modular flow extends the causal wedge algebra $\mA_R^{CW}$ to the entanglement wedge algebra $\mA_R^{EW_g}$
\be 
\mA_R^{EW_g}\simeq \{\r_R^{is}\mA_R^{CW}\r_R^{-is}\}_{s\in \R}''\simeq \{\r_R^{is}\mA_R^0\r_R^{-is}\}_{s\in \R}''\supset \mA_R^0 \label{eq:3.31}
\ee
where the double prime means double commutant, which completes it as a von Neumann algebra. From the bulk viewpoint, the entanglement wedge algebra $\mA^{EW_g}_R$ is a type III$_1$ algebra. Similarly, we can use the left modular flow to extend $\mA^{CW}_L\simeq \mA^{0}_L$ to the left entanglement wedge algebra $\mA^{EW_g}_L$. 

By the structure of the sub-Hilbert space $\mH_g$, it is equivalent to a bulk sub-Hilbert space by acting all light bulk fields $\p_i$ on the bulk Cauchy slice corresponding to $\bra{\Psi}$. It follows that $\p_i$ forms the full set of bounded operators $\mB(\mH_g)$. On the other hand, from the bulk viewpoint, the algebra of light fields living on the Cauchy slice is generated by the union of $\mA_R^{EW_g}$ and $\mA_L^{EW_g}$, which leads to
\be  
\mB(\mH_g)=(\mA_R^{EW_g}\cup \mA_L^{EW_g})'' \label{eq:3.32}
\ee 
Note that the HKLL reconstruction preserves bulk locality and the modular flow is a $SL(2)$ transformation, which acts locally as well. Therefore, in the reconstruction \eqref{eq:3.31}, there is an {\it emergent} bulk locality. This bulk locality implies two things: first, $\mA^{EW_g}_R$ commutes with $\mA^{EW_g}_L$ because they are spacelike separated; second, there is no nontrivial bounded bulk field living on the border of left and right entanglement wedge, namely the point $\Phi_{h,e^*}$, because of the universal UV divergence of bulk QFT.\footnote{Equivalently, every bounded operator needs smearing to be well-defined.} It follows that $\mA^{EW_g}_R=(\mA^{EW_g}_L)'$. With \eqref{eq:3.32}, we have \cite{Kolchmeyer:2023gwa}
\be 
(\mA^{EW_g}_R\cap(\mA^{EW_g}_R)')'=(\mA^{EW_g}_R\cup(\mA^{EW_g}_R)')''=(\mA^{EW_g}_R\cup\mA^{EW_g}_L)''=\mB(\mH_g)=(\C \cdot \I)'
\ee
which implies that $\mA^{EW_g}_{L/R}$ are two type III$_1$ factors. With the extension by the modular flow, the full algebra of JT gravity with matter in the semiclassical limit contains the following decomposition
\be 
\mB(\mH)\supset \oplus_g \mB(\mH_g)=\oplus_g (\mA_R^{EW_g}\cup \mA_L^{EW_g})''
\ee

\section{Conclusion and discussions} \label{sec:4}

In this paper, we studied the semiclassical limit of the theory of JT gravity with matter with two boundaries by taking $\phi_b\ra\infty$ limit. We assume the heavy operators with conformal weight $\sim O(\phi_b)$ become generalized free fields and decouple from light operators with conformal weight $\sim O(1)$. For a generic partially entangled thermal state (PETS) $\bra{\Psi}$ with different heavy operators inserted on the Euclidean path, the Euclidean path integral in the semiclassical limit is dominated by a saddle, which has a simple geometric formalism by gluing EAdS$_2$ disks with $Sl(2)$ charge conservation. We show that the geometry corresponds to a fixed Euclidean long wormhole connecting the left and right boundaries. Different configurations of heavy operators typically correspond to different Euclidean geometries, and these PETS are orthogonal to each other in the semiclassical limit. Therefore, the full Hilbert space contains a direct sum decomposition labelled by the Euclidean geometry $g$ corresponding to each generic PETS $\bra{\Psi}$
\be 
\mH\supset \oplus \mH_g \label{eq:4.1}
\ee
where each $\mH_g$ is generated by inserting light operators on the Euclidean path of $\bra{\Psi}$.

On the other hand, all heavy operators do not survive in the semiclassical limit as their conformal weights become infinite. The left and right Hamiltonian does not survive either because their on-shell expectation values are $\sim O(\phi_b)$ in a generic PETS. Nevertheless, the Lorentzian evolution of light operators on the left or right boundary are still bounded operators and they form a universal type III$_1$ von Neumann algebra $\mA_{L/R}^0$. For a generic PETS $\bra\Psi$ with a long wormhole, these two algebras are dual to the causal wedge in the bulk through HKLL reconstruction and do not generate the full sub-Hilbert space $\mH_g$ because there exists a gap between the left and right causal wedge on the global bulk Cauchy slice. In other words, the entanglement wedge is larger than the causal wedge. To reconstruct the entanglement wedge, we show that the modular flow by the reduced density matrix $\r_{L/R}$ for $\bra{\Psi}$ extends $\mA_{L/R}^0$ to the algebra of the full entanglement wedge because the flow acts geometrically as the boost transformation around the global RT surface in the semiclassical limit. This proves the entanglement wedge reconstruction proposed in \cite{Jafferis:2015del,Faulkner:2017vdd} in this scenario. The extended algebra, by the bulk dual in the entanglement wedge, is also type III$_1$. The full algebra in the theory contains a decomposition similar to \eqref{eq:4.1}
\be 
\mB(\mH)\supset \oplus_g \mB(\mH_g)=\oplus_g (\mA_R^{EW_g}\cup \mA_L^{EW_g})''
\ee 
where $\mA_{L/R}^{EW_g}$ is the entanglement wedge bulk algebra and equivalent to the extension of $\mA_{L/R}^0$ by the modular flow.

The following are discussions and future directions.

\subsection*{Other modular flows}

The modular flow in our construction for a generic PETS $\bra{\Psi}$ is a canonical choice at finite $\phi_b$ because we only have two algebras $\mA_L$ and $\mA_R$, which are commutant to each other. We take this modular flow generated by the reduced density matrix $\r_R$ in the $\phi_b\ra\infty$ limit though $\r_R$ does not survive in this limit. On the other hand, if we take $\phi_b\ra\infty$ first, as we argued in Section \ref{sec:2.5}, all light operators form an emergent subalgebra $\mA^0_{L/R}$. As we know that for any type III$_1$ algebra and its commutant, modular operator exists. From bulk QFT viewpoint, we can definitely consider the modular operator $\D_{0,R}$ for $\mA_R^0$ and $(\mA_R^0)'$, whose modular flow is an automorphism for $\mA_R^0$, namely $\D_{0,R}^{is}\mA_R^0\D_{0,R}^{-is}= \mA_R^0$ for $s\in\R$. More generally, for the bulk fields on the Cauchy slice of $g$ corresponding to $\bra{\Psi}$, we can choose an arbitrary causal diamond $\mD$ containing right causal wedge and consider the modular operator $\D_{\mD,R}$ for the bulk field algebra in this causal diamond $\mA_R^\mD$ and its commutant. We may expect the modular flow by $\D_{\mD,R}$ extends $\mA_R^0$ to some larger sub-algebra of $\mA_R^\mD$ though no precise statement can be easily made. In particular, we do not have a finite $\phi_b$ version of $\D_{\mD,R}$ because the operators in $\mA^\mD_R$ does not form a close algebra away from the semiclassical limit. Furthermore, in lack of a Euclidean path integral representation, it is unclear if such modular flow acts geometrically.

Even starting from finite $\phi_b$, we may have more choices of modular flows. In this work, we choose the PETS $\bra{\Psi}$ with only heavy operator insertions to construct the modular operator. This is the modular operator for the ``vacuum" state in $\mH_g$. We can definitely consider an ``excited" state $\bra{\Psi;\mO^l_i\cdots}$ and the corresponding modular operator. To compute the matrix elements of $\mO_{j,R}^l(t)$ under this new modular flow, we can use the same replica trick in Section \ref{sec:3.1}. Suppose there are $k$ light operators in $\mO_{R,i}^l(t)$. The $n$ replica wormhole saddle geometry is unchanged because light operators does not affect the saddle, then we need to evaluate a $(2kn+2)$-point function of light operators for $g_{ij}(n,m,\tau)$ and a $(2kn+3)$-point function for $g_{ijk}(n,m,\tau)$. Since the multi-point functions depend on $n$ and $m$ in a perhaps nontrivial way, it is unclear if the analytic continuation $n\ra 1$, $m\ra -is$ and $\tau\ra it$ still leads to a boost transformation of $Y^p(it)$ in a $(2k+2)$- or $(2k+3)$-point function. Naively, we should expect some nontrivial interaction between  $\mO_{j,R}^l$ and the light operators $\mO^l_i\cdots$ in the state along the modular flow. We leave this to a future investigation.

There is an interesting comment about the geometric feature of the modular flow in JT gravity in the semiclassical limit. It has been shown in \cite{Engelhardt:2021mue,Chen:2022eyi} that if the modular flow of a boundary subregion is geometric, then the entanglement wedge must be coincident with the causal wedge. On the other hand, in this work we show that the causal wedge is completely nested in the entanglement wedge but the modular flow acts geometrically. It turns out that both results are consistent with each other because the modular flow considered in \cite{Engelhardt:2021mue,Chen:2022eyi} is for a boundary subregion in higher dimensions rather than the two-sided case in JT gravity. Indeed, since the boundary in JT gravity is 0+1 dimensional, there does not exist a boundary subregion to apply the analysis of \cite{Engelhardt:2021mue,Chen:2022eyi} and thus no contradiction follows.

\subsection*{Modular flow of a non-generic PETS}

In Section \ref{sec:3}, we assume a generic PETS $\bra{\Psi_p}$ with all different heavy operators. For a non-generic PETS $\bra{\Psi}$ with identical operators, there will be more subtleties similar to the discussion of the structure of Hilbert space in Section \ref{sec:2.4}. For a PETS without crossing contractions, as the Euclidean geometry holds, the modular flow will be the boost transformation around the global minimal dilaton point. However, it is unclear if this minimal point always lies in the central row of Euclidean disks (e.g. Figure \ref{fig:dif1}). If it lies in the central row (see Appendxi \ref{app:3} for an example), the modular flow is the boost around the RT surface, exactly the same as the generic case; If it does not, there must be two identical minimal points with one from the ket state and the other from the bra state. This implies that the modular flow in the Lorentzian theory is no longer geometric. More generally, if there are crossing contractions between different types of heavy operators, the modular flow could be more involved. We leave the investigation of these non-generic cases to future work.

\subsection*{Operational algebraic ER=EPR}
It was recently proposed by Engelhardt and Liu in \cite{Engelhardt:2023xer} that the type of the full boundary algebra can be used to distinguish whether the dual bulk has a connected geometry or disconnected geometry for the entanglement wedge. The statement is as follows. For two boundary regions $R_1$ and $R_2$ and a pure entangled state $\bra{\Psi(R_1,R_2)}$, in $G_N\ra 0$ limit the entanglement wedge for $R_1\cup R_2$ has three cases
\begin{enumerate}
\item If and only if both boundary algebra $\mA_{R_1}$ and $\mA_{R_2}$ are type I, then the dual entanglement wedge of $R_1\cup R_2$ is disconnected;
\item If an only if both boundary algebra $\mA_{R_1}$ and $\mA_{R_2}$ are type III$_1$, then the dual entanglement wedge of $R_1\cup R_2$ is classically connected;
\item If both boundary algebra $\mA_{R_1}$ and $\mA_{R_2}$ are not type I, and the state $\bra{\Psi(R_1,R_2)}$ does not obey ``classical condition", then the dual entanglement wedge of $R_1\cup R_2$ is quantumly connected.
\end{enumerate}
This is an interesting extension of the well-known proposal of ``ER=EPR" \cite{Maldacena:2013xja} by Maldacena and Susskind, in which a large amount of quantum entanglement is conjectured to be dual to connected spacetime. As pointed out in \cite{Engelhardt:2023xer}, the connectedness of the spacetime depends not only on the amount of entanglement but also on the structure of the entanglement, which can be naturally characterized by the type of boundary von Neumann algebra.

However, this algebraic ER=EPR criterion based on the type of boundary algebra is not operationally easy to verify. Given a boundary theory, we may not know all operators in the boundary algebra but just the set of single trace operators, which in semiclassical limit forms an emergent type III$_1$ sub-algebra $\mA^0$ dual to the causal wedge \cite{Leutheusser:2021frk}. This is the universal type III$_1$ sub-algebra of the full algebra. What do we need to know additionally empowers us to determine the type of full algebra? This is equivalent to entanglement wedge reconstruction if we assume that we can determine the type of the algebra after we know all the operators. 

Although it is not rigorously proven, let us assume the entanglement wedge reconstruction by the modular flow \cite{Faulkner:2017vdd}. The operational algebraic ER=EPR problem becomes what types of algebra the modular flow extends to from a sub-algebra, say $\mA^0$. As the modular flow is a one-parameter group for each operator, the flowed algebra $\D^{is}\mA^0\D^{-is}$ for each $s\in\R$ is the same type as $\mA^0$ because $\D^{is}$ is a unitary that preserves trace. However, if we consider the algebra generated by the union of all $s\in\R$, we may be able to have different types of algebra. 

In this paper, we show that $\{\D^{is}\mA^0_R\D^{-is}\}_{s\in\R}''$ is of type III$_1$ for the right entanglement wedge. Here we provide another example that the modular flow extends a type III$_1$ subalgebra to a type I algebra. Consider the two CFTs on two distinct spheres with radius $\l$. Prepare them in a thermofield double state below Hawking-Page temperature, so the dual bulk consists of two disconnected AdS spaces with entangling thermal gas. The full algebra on either side is type I while a spatial hemisphere on the left (or right) CFT has a type III$_1$ subalgebra $\mA_{L,hs}^0$ (or $\mA_{R,hs}^0$). Let us focus on the right system. The modular flow in this case is simple, that is, the Hamiltonian evolution on either side. Consider the extended algebra $\{\D^{is}\mA_{R,hs}^0\D^{-is}\}_{|s|\leq s_0}''$. It is clear that by the time tube theorem \cite{borchers1961vollstandigkeit,araki1963generalization,strohmaier2000local,Strohmaier:2023hhy} when $s_0$ is greater than a value that is of order $O(\l)$, this extended algebra will be the same as the full CFT algebra on the left, which is type I. We can also see from this example that the extended algebra covers the whole entanglement wedge, namely the whole bulk on the left. For general cases, what is the type of modular extended algebra, and how to determine it? We leave this question for future work.

\subsection*{Low-complexity reconstruction}

It was argued in \cite{Engelhardt:2021mue} that one can use Einstein
equations to reconstruct at most to the outmost extremal surface from
causal wedge with low complexity operations on the boundary. This
means that python's lunch \cite{Brown:2019rox} is the complexity obstacle for
the reconstruction to the deep bulk from the boundary. On the other hand, we show in this work that the modular flow in JT gravity acts geometrically and
extends the causal wedge all the way to the entanglement wedge regardless
of python's lunch. Clearly, we should regard the modular operator
$\D_{\Psi_{p}}$ as a very complex operator because it lies outside
of the algebra $\mA_{R}^0$ and its construction requires
Euclidean evolution that is quite complex from boundary viewpoint.

In JT gravity, there is an analogous method of reconstruction beyond
causal wedge towards the outmost extremal surface without modular
flow. Similar to using the Einstein equation in \cite{Engelhardt:2021mue}, which is at $O(G_N)$ perturbative order, this method also requires large but finite $\phi_b$. Large $\phi_b$ means we can still treat the geometry as the saddle solution, but finite $\phi_b$ allows us to measure some quantities the same order as $\phi_b$ in the bulk. 

Consider the state (\ref{eq:1}) as an example. For $\mu>\mu_{*}$
the edge of causal wedge is the outmost extremal surface and there
is a python's lunch (Figure \ref{fig:1d}); for $\mu<\mu_{*}$ the
edge of causal wedge is not the outmost extremal surface and there
is no python's lunch (Figure \ref{fig:1c}). Our method applies to the second case $\mu<\mu_{*}$ as follows.

As $\phi_b$ is finite, let us assume we can measure the dilaton profile in the causal wedge. Note that the geodesic (red line) from $\mO_i$ lies inside the causal
wedge and split it into two parts. The dilaton profile on the two Euclidean disks is given by (\ref{eq:20}) with different charges $Q_L$ and $Q_R$. Witten in $(\r,\tau)$ coordinates of each disk, they are
\be 
\Phi_L=\f{2\pi \phi_b}{b_L}\cosh \r_L,\quad \Phi_R=\f{2\pi \phi_b}{b_R}\cosh \r_R
\ee
Since the radial metric is flat $d\r^2$ on $t=0$ Cauchy slice, we can measure the dilaton profile fall-off along the radial direction within the causal wedge for both disks, and learn $b_L$ and $b_R$. On the other hand, the dilaton profile is continuous across
the geodesic connecting two $\mO_i$'s but its normal derivative is discontinuous and is proportional
to $\mu$ \cite{Goel:2018ubv}. To see this, recall that the dilaton profiles on the two EAdS$_2$ disks are given by (\ref{eq:20}) with different charges $Q_L$ and $Q_R$. On the border of these two disks (namely along the geodesic connecting two $\mO_i$'s), the dilaton value should be continuous
\begin{equation}
\Phi_L(Y)-\Phi_R(Y)=\phi_{b}(Q_L-Q_R)\cdot Y=\phi_{b}Q_{\mO_i}\cdot Y=0 \label{eq:65}
\end{equation}
where $Q_{\mO_i}$ is the $SL(2)$ charge carried by $\mO_i$. However, the difference of the normal derivative to the dilaton profile is not continuous
\begin{equation}
n\cdot\del_{Y}(\Phi_L(Y)-\Phi_R(Y))=\phi_{b}n\cdot(Q_L-Q_R)=\phi_{b}n\cdot Q_{\mO_i}
\end{equation}
By (\ref{eq:65}), for a tangent vector $T$ along the geodesic we
have $Q_{\mO_i}\cdot T=0$. It follows that 
\begin{equation}
|n\cdot\del_{Y}(\Phi_L(Y)-\Phi_R(Y))|^{2}=\phi_{b}^{2}|n\cdot Q_{\mO_i}|^{2}=\phi_{b}^{2}Q_{\mO_i}^2=\phi_{b}^{2}\mu^{2}
\end{equation}
and we learn $\mu$ from this measurement. Knowing $b_L,b_R,\mu$, we can use \eqref{eq:30} to compute $z\equiv z_R-z_L$ with \eqref{eq:27-1}. The physical meaning of $z$ is the distance between the centers of the two disks. In particular, their centers are related by a $SL(2)$ transformation $M_2(z)$.

\begin{figure}
\begin{centering}
\includegraphics[height=5cm]{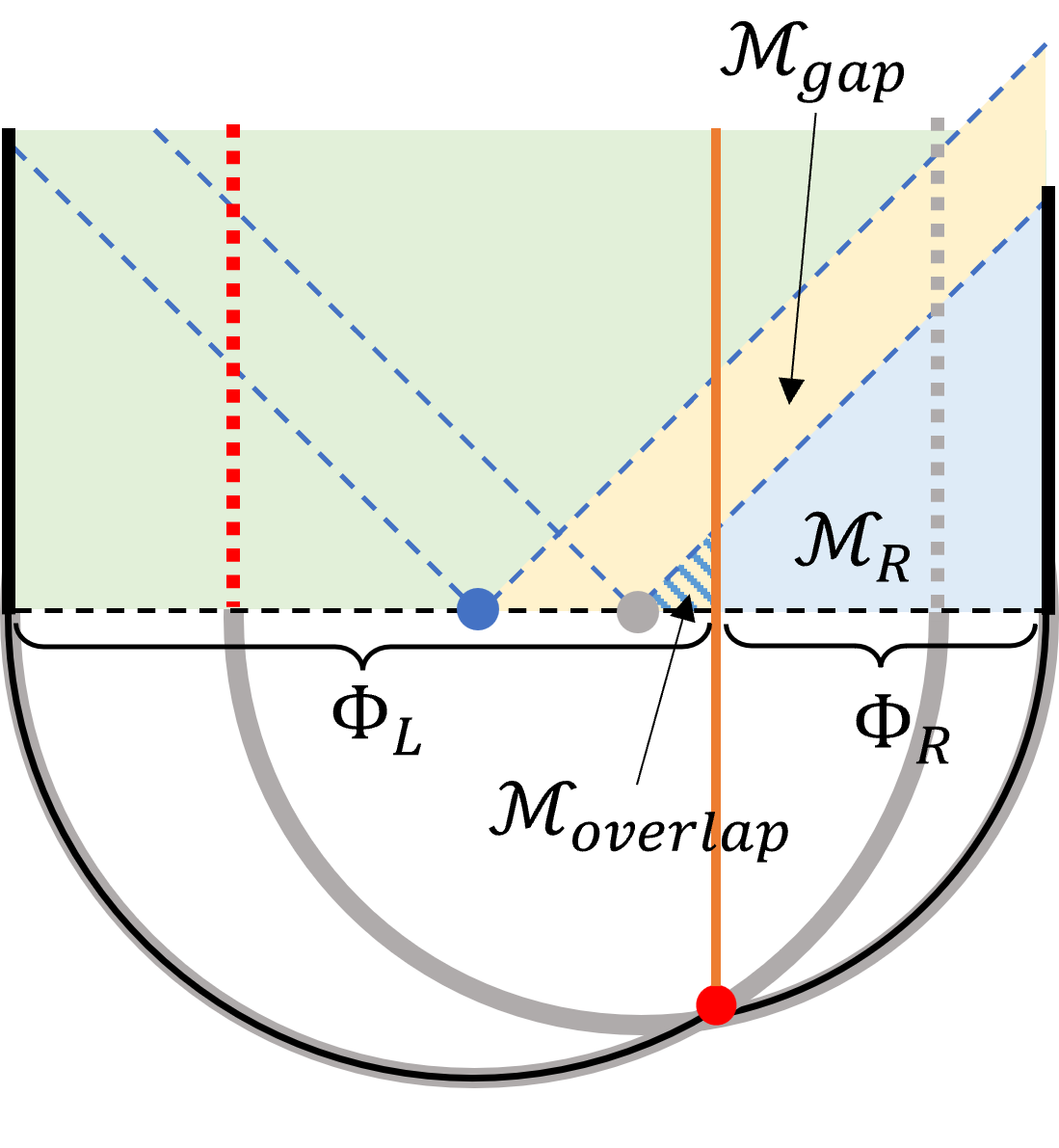}
\par\end{centering}
\caption{The lower half are two Euclidean disks and the upper half is the Lorentzian continuation of a long wormhole. There is only one extremal dilaton point (blue dot), which is the outmost extremal surface from the right boundary. The geodesic (orange line) separates the region between the outmost extremal surface and right boundary into $\mM_{gap}$ (yellow region) and $\mM_R$ (blue region). The right causal wedge has an overlap with $\mM_{gap}$, which we call $\mM_{overlap}$  (shaded region). The right causal wedge $\mM_{CW}=\mM_R\cup\mM_{overlap}$. If there were no right disk, the fictitious right boundary would be along the dashed gray line. }\label{fig:lowC}
\end{figure}

Let us separate the spacetime region between the outmost extremal surface and right boundary along the red geodesic in Figure \ref{fig:1c} into two parts $\mM_{gap}$ and $\mM_R$. The gap region has an overlap with the right causal wedge $\mM_{CW}$, which we call $\mM_{overlap}$. See Figure \ref{fig:lowC} as an illustration. The reconstruction in $\mM_{R}$ is just the ordinary HKLL from the right boundary
\be  
\p_i(Y)=\int dt K_{\D_i}(\r,t;t')\mO^l_{i,R}(u_R(t')),\quad Y=(\r,t)\in\mM_{R} \label{eq:4.5}
\ee
For $\mM_{gap}$, we need some extension of HKLL. If there were no right disk, the state would just be thermofield double with inverse temperature $b_L$. The fictitious right boundary would be the Lorentzian continuation along the right side of the left disk and we could reconstruct the whole fictitious right wedge by the HKLL reconstruction along the fictitious right boundary
\be 
\p_{i,F}(\r_F,t_F)=\int dt K_{\D_i}(\r_F,t_F;t)\mO^l_{i,F}(u_F(t)) \label{eq:4.6}
\ee 
where $u_F=b_L t/(2\pi)$ the subscript $F$ is for the fictitious right spacetime. This fictitious right wedge exactly covers $\mM_{gap}$. Since the two disks are related by $M_2(z)$, we can rewrite \eqref{eq:4.6} as a reconstruction in $\mM_{gap}$
\be  
\p_i(Y)=\int dt K_{\D_i}(\r_z,t_z;t)\mO^l_{i,F}(u_F(t)),~(\r_z,t_z)=M_2(z)\cdot Y\in \mM_{CW},~ Y\in \mM_{gap} \label{eq:4.7}
\ee 
where we stress that $(\r_z,t_z)$ is written in the coordinate of the right causal wedge $\mM_{CW}$. 

To determine the fictitious boundary operator $\mO^l_{i,F}(u)$, we need to impose the smoothness condition for the bulk field in $\mM_{overlap}$ between the $\mM_{R}$ and $\mM_{gap}$. Since $\mM_{overlap}\subset\mM_{CW}$, we can safely extend the reconstruction \eqref{eq:4.5} from $\mM_R$ to $\mM_{CW}$ and the two reconstructions \eqref{eq:4.5} and \eqref{eq:4.7} should match with each other. In particular, extending \eqref{eq:4.7} to $\mM_{CW}$ smoothly and taking $\r_F\ra \infty$ limit, we have
\begin{align}
\mO^l_{i,F}(u_F(t_F))&=\lim_{\r_F\ra\infty} e^{\D_i \r_F}\p_i(\r_F,t_F) \nn\\
&=\lim_{\r_F\ra\infty} e^{\D_i \r_F}\int dt K_{\D_i}(\r_{-z},t_{-z};t)\mO^l_{i,R}(u_R(t)) \label{eq:4.8}
\end{align}
where $(\r_{-z},t_{-z})=M_2(-z)\cdot (\r_F,t_F)$ is the coordinate in $\mM_{CW}$, and in the second line we used \eqref{eq:4.5} to reconstruct the bulk field $\p_i(\r_F,t_F)$. The $SL(2)$ transformation in Lorentzian signature can computed using \eqref{eq:8} with analytic continuation $\tau\ra it$, which leads to
\be 
\cosh\r_{-z}=\cosh z \cosh \r_F-\sinh z \sinh \r_F \cosh t_F,\quad \sinh t_{-z}=\sinh \r_F \sinh t_F/\sinh \r_{-z}
\ee
In large $\r_F$ limit, we have
\be 
\sinh t_{-z}=\f{\sinh t_F}{\cosh z-\sinh z \cosh t_F},\quad e^{\r_{-z}}=e^{\r_F}(\cosh z-\sinh z\cosh t_F)
\ee 
Note that the HKLL kernel obeys
\be 
\lim _{\r\ra\infty}e^{\r \D}K_{\D}(\r,t;t')=\d(t-t')
\ee 
We take them into \eqref{eq:4.8} and derive
\be 
\mO^l_{i,F}(u_F(t_F))=\f{\mO_{i,R}^l\left(\f{b_R}{2\pi}\arcsinh\left(\f{\sinh t_F}{\cosh z-\sinh z \cosh t_F}\right)\right)}{(\cosh z-\sinh z\cosh t_F)^{\D_i}} \label{eq:4.12}
\ee
which by \eqref{eq:4.7} extends the reconstruction all the way to the outmost extremal surface.

For our setup, $z>0$ and \eqref{eq:4.12} becomes singular at finite $t_F=\pm t_*=\pm\arccosh(\coth z)$, which corresponds to $u=\pm \infty$ for $\mO^\l_{i,R}(u)$. From Figure \eqref{fig:lowC}, we see that a light ray from the causal horizon (gray dot) intersects the fictitious boundary (dashed gray line) at a finite $t_F$ but the right boundary at $\infty$. Indeed, in $(\r_F,t_F)$ coordinate, the light ray from causal horizon obeys $t_F=\log (\tanh(\r_F/2)/\tanh(z/2))$. Taking $\r_F\ra\infty$ leads to finite $t_F=\log \coth(z/2)=t_*$. which is consistent with the singularity in \eqref{eq:4.12}. To continue \eqref{eq:4.12} beyond the singularity at $|t_F|=t_*$, we need to add appropriate $i\e$ to $t_F$ depending on the ordering of operators in a correlation function involving $\mO^l_{i,F}$. For example, in a time ordered correlation function, we should modify $t_F\ra t_F-i\e$ with infinitesimal $\e>0$ in \eqref{eq:4.12} and the argument in $\mO_{i,R}^l$ gains an imaginary part $-i b_R/2$ for $|t_F|>t_*$
\be 
\mO^l_{i,F}(u_F(t_F))=\f{\mO_{i,R}^l\left(\f{b_R}{2\pi}\arcsinh\left(\f{\sinh t_F}{\sinh z \cosh t_F-\cosh z}\right)-ib_R/2\right)}{e^{i\pi \D_i \sgn(t_F)}(\sinh z\cosh t_F-\cosh z)^{\D_i}},\quad |t_F|>t_*\label{eq:4.13}
\ee
The imaginary part $b_R/2$ is nothing but the circumference of the semicircle of the right disk in Figure \ref{fig:lowC}. In the thermofield double state prepared by the right disk, this Euclidean analytic continuation would correspond to a left boundary operator, which lives on a fictitious world line to the left of $\mM_{gap}$ (the dashed red line in Figure \ref{fig:lowC}). \eqref{eq:4.13} is the mirror operator in the state-dependent reconstruction \cite{Papadodimas:2012aq,Papadodimas:2013jku}. Indeed, the state dependence is manifest in \eqref{eq:4.12} because the fictitious operator $\mO^l_{i,F}$ is parameterized by $b_R$ and $z$ from the state. As $t_F\ra \pm\infty$, the real part of the argument of $\mO^l_{i,R}$ becomes $\pm b_R/(2\pi)\arcsinh(1/\sinh z)=\pm u_R(t_*)$. This means that for $|t_F|\in(t_*,\infty)$, $\mO^l_{i,t_F}(u_F(t_F))$ corresponds to the operator on the fictitious left boundary within the time band $\pm(t_*,\infty)$. This is quite reasonable as it is exactly the region on the fictitious left boundary with causal dependence to $\mM_{gap}$.

It is easy to see that this method applies to a generic PETS as long as there is no python's lunch right behind the causal horizon, which forbids us to measure the dilaton profile non-smoothness within the causal wedge. Since this method only requires local bulk measurement and HKLL kernel, we may naturally say that this is a low-complexity reconstruction.

\section*{Acknowledgements} I thank Ahmed Almheiri, Chris Akers, Xiao-liang Qi, Huajia Wang, Zhenbin Yang, Yi Wang and Edward Witten for stimulating and helpful discussions. I especially thank David Kolchmeyer for his collaboration 
on the initial stage of this project, and for his talk at ``A Quantum Al-Khawarizm for Spacetime: A Workshop on von Neumann Algebras in Quantum Field Theory \& Gravity" at NYU Abu Dhabi Institute in New York. PG is supported by the US Department of Energy under grant DE-SC0010008.

\appendix 
\section{Generalized free fields and $sl(2)$ 6j symbol} \label{app:1}
In this section, we will use the Euclidean path integral method to
derive the $sl(2)$ 6j symbol for the following matrix element in
the generalized free fields theory
\begin{equation}
\avg{b;k_{4},k_{3}|\mO_{a,L}\mO_{a,R}|b;k_{1},k_{2}}=\avg{b;k_{4},k_{3}|(4e^{-\l})^{\D_{a}}|b;k_{1},k_{2}}=\begin{Bmatrix}\D_{a} & k_{1} & k_{4}\\
\D_{b} & k_{3} & k_{2}
\end{Bmatrix}\sqrt{\G_{k_{1}k_{4}}^{\D_{a}}\G_{k_{2}k_{3}}^{\D_{a}}}\label{eq:app1}
\end{equation}
where the pair of $\mO_{a}$ and $\mO_{b}$ operators are contracted
respectively. This corresponds to the third diagram in Figure \ref{fig:phi2}
when $a=b$. There are a few different ways to derive this formula
\cite{Mertens:2017mtv, Suh:2020lco}. However, the following derivation directly from the Euclidean
path integral is much simpler than \cite{Mertens:2017mtv, Suh:2020lco} and will be most convenient for discussing the semiclassical
limit. 

\begin{figure}
\begin{centering}
\includegraphics[height=4.5cm]{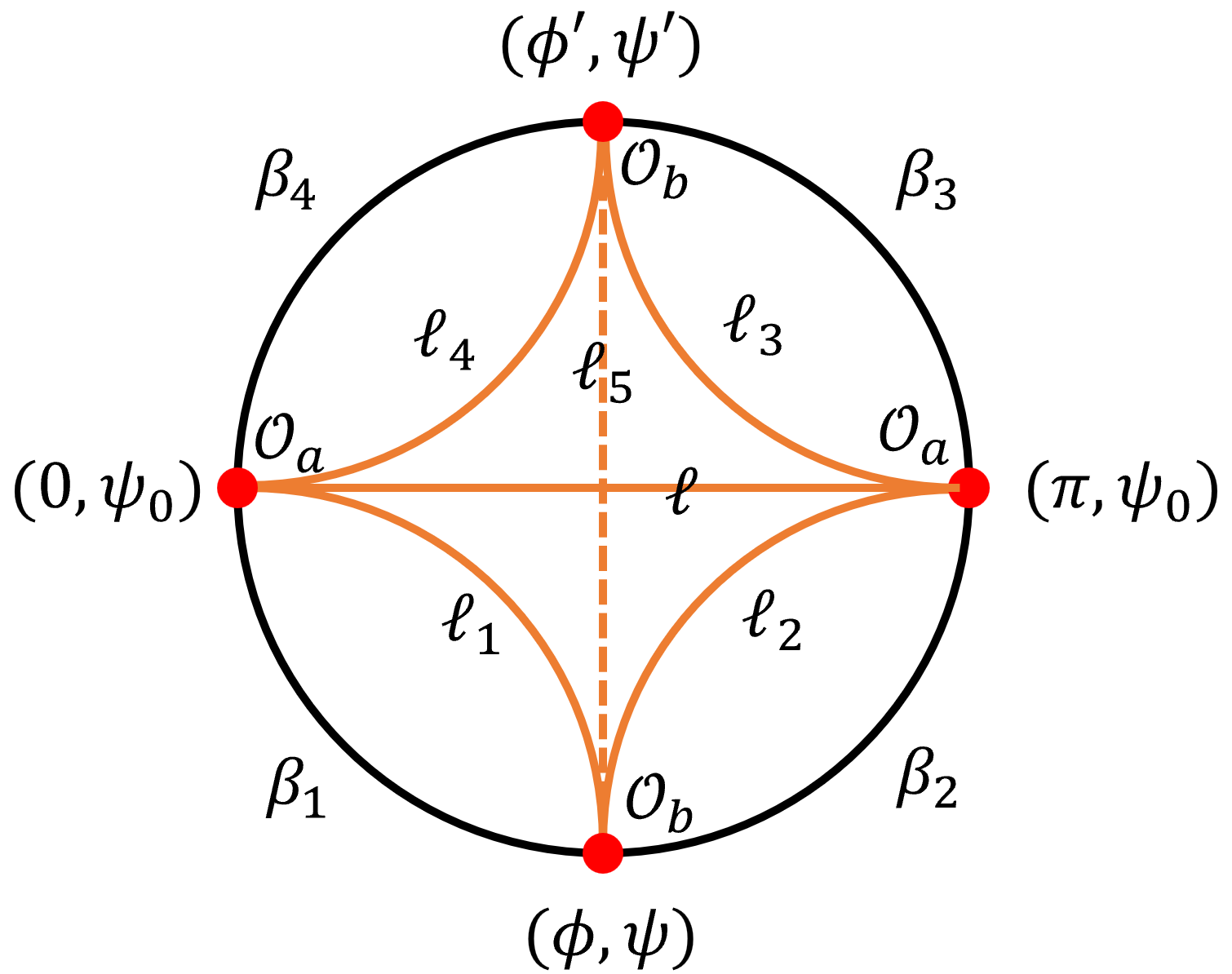}
\par\end{centering}
\caption{The Euclidean path integral for $\avg{b;\b_{4},\b_{3}|(4e^{-\l})^{\D_{a}}|b;\b_{1},\b_{2}}$. }\label{fig:appotoc}
\end{figure}

In the notation of this paper, the explicit formula for the Hartle-Hawking
state $\bra{\Psi_{\b_{1},\b_{2}}^{\mO}(\l)}$ is\footnote{This notation is a simpler version of Eq.(4.2), (2.28) and (2.29)
in \cite{Kolchmeyer:2023gwa} but with a few differences: the spectral density $\r_{here}(k)=2\r_{there}(k)$,
there is an additional factor of 2 in $\mK_{\b}$ but no factor of
2 in $\bra{\Psi_{\b_{1},\b_{2}}^{\mO}(\l)}$.} 
\begin{align}
& \bra{\Psi_{\b_{1},\b_{2}}^{\mO}(\l)} =\int_{-\infty}^{\infty}d\psi_{1}\int_{0}^{\pi}d\phi_{1}\mK_{\b_{2}}(\pi,\psi_{0};\phi_{1},\psi_{1})\mK_{\b_{1}}(\phi_{1},\psi_{1};0,\psi_{0})e^{\D\psi_{1}}\bra{\mO(\phi_{1})}|_{\psi_{0}=\log2-\l/2}\\
&\mK_{\b}(\phi_{1},\psi_{1};\phi_{0},\psi_{0}) =\f{e^{(e^{\psi_{1}}+e^{\psi_{2}})\cot\f{\phi_{1}-\phi_{2}}2}}{\sin\f{\phi_{2}-\phi_{1}}2}\int_{0}^{\infty}dk\r(k)e^{-\b k^{2}/(4\phi_{b})}2K_{2ik}\left(\f{2e^{(\psi_{1}+\psi_{2})/2}}{\sin\f{\phi_{2}-\phi_{1}}2}\right) \label{eq:a3}
\end{align}
By \eqref{hh1}, we have
\begin{equation}
\avg{b;\b_{4},\b_{3}|(4e^{-\l})^{\D_{a}}|b;\b_{1},\b_{2}}=\int\f{(4e^{-\l})^{\D_{a}}(4e^{-\l_{5}})^{\D_{b}}}{-\sin\phi\sin\phi'/4}\prod_{j=1,2,3,4}e^{m_{j}}\r(k_{j})e^{-k_{j}^{2}\b_{j}/(4\phi_{b})}2K_{2ik_{j}}(4e^{-\l_{j}/2})\label{eq:app2}
\end{equation}
where the integral measure is $d\phi d\phi'd\psi d\psi'd\l$ with
$\phi\in[0,\pi],\phi'\in[\pi,2\pi],\psi,\psi',\l\in\R$, and $\l_{j}$
and $m_{j}$ are defined as
\begin{align}
e^{-\l_{1}/2} & =\f{e^{(\psi_{0}+\psi)/2}}{2\sin\f{\phi}2},\quad e^{-\l_{2}/2}=\f{e^{(\psi_{0}+\psi)/2}}{2\cos\f{\phi}2},\quad e^{-\l_{3}/2}=\f{e^{(\psi_{0}+\psi')/2}}{-2\cos\f{\phi'}2},\quad e^{-\l_{4}/2}=\f{e^{(\psi_{0}+\psi')/2}}{2\sin\f{\phi'}2}\label{eq:app3}\\
\psi_{0} & =\log2-\l/2,\quad e^{-\l_{5}/2}=\f{e^{(\psi+\psi')/2}}{2\sin\f{\phi'-\phi}2},\quad m_{1}=-(e^{\psi_{0}}+e^{\psi})\cot\f{\phi}2\label{eq:app4}\\
m_{2} & =-(e^{\psi_{0}}+e^{\psi})\tan\f{\phi}2,\quad m_{3}=(e^{\psi_{0}}+e^{\psi'})\tan\f{\phi'}2,\quad m_{4}=(e^{\psi_{0}}+e^{\psi'})\cot\f{\phi'}2
\end{align}
where $\phi$ and $\psi$ are the canonical degrees of freedom
in quantization of Schwarzian action \cite{Kolchmeyer:2023gwa,Jafferis:2019wkd}, and $\l_{k}$ are geometric
variables for the renormalized geodesic lengths in Figure \ref{fig:appotoc}. Moreover,
$m_{k}$ can be further decomposed as geometric integrals using \eqref{eq:2.22}
\begin{equation}
e^{m_{1}+m_{2}}=I(\l_{1},\l_{2},\l),\quad e^{m_{3}+m_{4}}=I(\l_{3},\l_{4},\l)
\end{equation}
Changing $d\phi d\phi'd\psi d\psi'$ to $d\l_{1}d\l_{2}d\l_{3}d\l_{4}$
gives an additional Jacobian $-(\sin\phi\sin\phi')/4$ which cancels
the same factor in the denominator in (\ref{eq:app2}). Note from
(\ref{eq:app3}) and (\ref{eq:app4}) that $\l_{5}$ is not independent
from $\l_{1,2,3,4}$ and $\l$ but can be written as
\begin{equation}
e^{\l_{5}/2}=(e^{(\l_{1}+\l_{3})/2}+e^{(\l_{2}+\l_{4})/2})e^{-\l/2}\label{eq:app7}
\end{equation}
It follows that (\ref{eq:app2}) can be written as an integral over
geometric quantities
\begin{equation}
\avg{b;\b_{4},\b_{3}|(4e^{-\l})^{\D_{a}}|b;\b_{1},\b_{2}}=\int\f{(4e^{-\l})^{\D_{a}}(4e^{\l})^{\D_{b}}I(\l_{1},\l_{2},\l)I(\l_{3},\l_{4},\l)}{(e^{(\l_{1}+\l_{3})/2}+e^{(\l_{2}+\l_{4})/2})^{2\D_{b}}}\prod_{j=1,2,3,4}\Psi_{\b_{j}}(\l_{j})\label{eq:app8}
\end{equation}
which after inverse Laplace transformation gives
\begin{equation}
\avg{b;k_{4},k_{3}|(4e^{-\l})^{\D_{a}}|b;k_{1},k_{2}}=\int\f{(4e^{-\l})^{\D_{a}}(4e^{\l})^{\D_{b}}I(\l_{1},\l_{2},\l)I(\l_{3},\l_{4},\l)}{(e^{(\l_{1}+\l_{3})/2}+e^{(\l_{2}+\l_{4})/2})^{2\D_{b}}\sqrt{\G_{k_{1}k_{2}}^{\D_{b}}\G_{k_{3}k_{4}}^{\D_{b}}}}\prod_{j=1,2,3,4}2K_{2ik_{j}}(4e^{-\l_{j}/2})\label{eq:app9}
\end{equation}

To evaluate this integral, let us do a Barnes integral representation
for the $(4e^{-\l_{5}})^{\D_{b}}$ factor as
\begin{equation}
\f{(4e^{\l})^{\D_{b}}}{(e^{(\l_{1}+\l_{3})/2}+e^{(\l_{2}+\l_{4})/2})^{2\D_{b}}}=(4e^{\l})^{\D_{b}}e^{-(\l_{2}+\l_{4})\D_{b}}\int_{L_{t}}\f{dt}{2\pi i}\f{\G(-t)\G(2\D_{b}+t)}{\G(2\D_{b})}(e^{(\l_{1}+\l_{3}-\l_{2}-\l_{4})/2})^{t}
\end{equation}
where $L_{t}$ is the line from $-i\infty$ to $i\infty$ with real
part $t_{r}\in(-2\D_{b},0)$. Taking this back to (\ref{eq:app9}),
we can easily integrate over $\l_{i}$ and $\l$ using \eqref{Gd19} and have
\begin{align}
\mM & \equiv\avg{b;k_{4},k_{3}|(4e^{-\l})^{\D_{a}}|b;k_{1},k_{2}}\sqrt{\G_{k_{1}k_{2}}^{\D_{b}}\G_{k_{3}k_{4}}^{\D_{b}}}\nonumber \\
 & =\int_{L_{t}}\f{dt}{2\pi i}\f{\G(-t)\G(2\D_{b}+t)}{\G(2\D_{b})}\int dkdk'\r(k)\r(k')\G_{kk_{1}}^{-t/2}\G_{kk_{2}}^{\D_{b}+t/2}\G_{kk'}^{\D_{a}-\D_{b}}\G_{k'k_{3}}^{-t/2}\G_{k'k_{4}}^{\D_{b}+t/2}
\end{align}
Note that the spectrum can be written as $\r(k)=(\pi|\G(2ik)|^{2})^{-1}$.
Let us assume $\D_{a}>\D_{b}$. We can integrate $k$ as
\begin{align}
 & \f 1{\pi}\int_{0}^{\infty}dk\f{\G_{kk_{1}}^{-t/2}\G_{kk_{2}}^{\D_{b}+t/2}\G_{kk'}^{\D_{a}-\D_{b}}}{|\G(2ik)|^{2}}=\f{\G(\D_{b}\pm ik_{1}+ik_{2},\D_{a}+t/2-ik_{2}\pm ik')}{2\pi i\G(2\D_{b}+t)}\nonumber \\
 & \times\int_{L_{s}}ds\f{\G(\D_{a}+t/2+s+ik_{2}\pm ik',s+t+2\D_{b},\pm ik_{1}-ik_{2}-s-t-\D_{b},-s)}{\G(-s-t,s+t+2\D_{a})}\label{eq:app12}
\end{align}
where we used (2.3) in \cite{yurii2002beta}\footnote{There are a few different ways to apply (2.3) in \cite{yurii2002beta}. For our
purpose of (\ref{eq:app12}) to separate $\D_{a}$ and $\D_{b}$,
we choose parameters there as $a,f=\D_{a}+t/2\pm ik_{2}$, $b,e=\D_{a}-\D_{b}\pm ik'$,
and $c,d=-t/2\pm ik_{1}$.} and $L_{s}$ is a line from $-i\infty$ to $i\infty$ with real part
$s_{r}$ such that the real part of all gamma functions in the numerator
of (\ref{eq:app12}) is positive, namely $\max(-2\D_{b}-t_{r},-\D_{a}-t/2)<s_{r}<\min(0,-\D_{b}-t)$.
Applying (2.3) in \cite{yurii2002beta}\footnote{This time we choose parameters there as $a,f=-t/2\pm ik_{3}$, $b,e=\D_{b}+t/2\pm ik_{4}$,
$c=\D_{a}+t/2+s+ik_{2}$, and $d=\D_{a}+t/2-ik_{2}$.} again to the $k'$ integral, we have
\begin{align}
 & \f 1{\pi}\int_{0}^{\infty}dk'\f{\G_{k'k_{3}}^{-t/2}\G_{k'k_{4}}^{\D_{b}+t/2}\G(\D_{a}+t/2+s+ik_{2}\pm ik',\D_{a}+t/2-ik_{2}\pm ik')}{|\G(2ik')|^{2}}\nonumber \\
= & \f{\G(\D_{a}+s+ik_{2}+ik_{3},s+t+2\D_{a},\D_{b}-ik_{3}\pm ik_{4},\D_{a}+ik_{3}-ik_{2})}{2\pi i\G(-t)}\nonumber \\
 & \times\int_{L_{s'}}ds'\f{\G(\D_{b}+ik_{3}\pm ik_{4}+s',s'-t,\D_{a}+t-ik_{2}-ik_{3}-s',\D_{a}+t+ik_{2}-ik_{3}+s-s',-s')}{\G(2\D_{b}+s',2\D_{a}+s+t-s')}\label{eq:app12-1}
\end{align}
where $L_{s'}$ is defined similarly and the real part of $s'$ obeys
$\max(t,-\D_{b})<s_{r}'<\min(\D_{a}+t+s,0)$. Putting (\ref{eq:app12})
and (\ref{eq:app12-1}) together, we have
\begin{align}
\mM & =\int\f{dtdsds'}{(2\pi i)^{3}}\f{\G(\D_{b}\pm ik_{1}+ik_{2},\D_{b}-ik_{3}\pm ik_{4},\D_{a}+ik_{3}-ik_{2})}{\G(2\D_{b})}\nonumber \\
 & \times\f{\G(s+t+2\D_{b},\pm ik_{1}-ik_{2}-s-t-\D_{b},-s,\D_{a}+s+ik_{2}+ik_{3})}{\G(-s-t)}\nonumber \\
 & \times\f{\G(\D_{b}+ik_{3}\pm ik_{4}+s',s'-t,\D_{a}+t-ik_{2}-ik_{3}-s',\D_{a}+t+ik_{2}-ik_{3}+s-s',-s')}{\G(2\D_{b}+s',2\D_{a}+s+t-s')}\label{eq:app14}
\end{align}
Define $\tilde{t}=t-s$ and we find that there are four gamma functions
only involving $s$ for fixed $\tilde{t}$
\begin{align}
 & \f 1{2\pi i}\int_{L_{s}}ds\G(\D_{a}+ik_{2}+ik_{3}+s,s'-\tilde{t}+s,\D_{a}-s'+\tilde{t}-ik_{2}-ik_{3}-s,-s)\nonumber \\
= & \f{\G(2\D_{a}-s'+\tilde{t},\D_{a}+ik_{2}+ik_{3},\D_{a}-ik_{2}-ik_{3},s'-\tilde{t})}{\G(2\D_{a})}
\end{align}
which is the first Barnes lemma. Taking this back
to (\ref{eq:app14}) we have
\begin{align}
\mM & =\int\f{d\tilde{t}ds'}{(2\pi i)^{2}}\f{\G(\D_{b}\pm ik_{1}+ik_{2},\D_{b}-ik_{3}\pm ik_{4},\D_{a}\pm ik_{3}\pm ik_{2})}{\G(2\D_{a},2\D_{b})}\nonumber \\
 & \times\f{\G(\tilde{t}+2\D_{b},\pm ik_{1}-ik_{2}-\tilde{t}-\D_{b},\D_{b}+ik_{3}\pm ik_{4}+s',-s')}{\G(-\tilde{t},2\D_{b}+s')}\nonumber \\
 & \times\f{\G(\D_{a}+\tilde{t}+ik_{2}-ik_{3}-s',s'-\tilde{t})}{\G(\D_{a}+ik_{2}-ik_{3})}\label{eq:app14-1}
\end{align}
Note that the last line is a beta function, which has an integral
representation
\begin{equation}
\f{\G(\D_{a}+\tilde{t}+ik_{2}-ik_{3}-s',s'-\tilde{t})}{\G(\D_{a}+ik_{2}-ik_{3})}=\int_{0}^{1}dxx^{\D_{a}+\tilde{t}+ik_{2}-ik_{3}-s'-1}(1-x)^{s'-\tilde{t}-1}
\end{equation}
Taking this back to (\ref{eq:app14-1}), we find that the $\tilde{t}$
and $s'$ integrals can be done separately. Using the integral representation
of hypergeometric function
\begin{equation}
_{2}F_{1}(a,b;c;z)=\f{\G(c)}{2\pi i\G(a,b)}\int_{L}dt\f{\G(a+t,b+t,-t)}{\G(c+t)}(-z)^{t}
\end{equation}
where $L$ is a line from $-i\infty$ to $i\infty$ with real part
of $a+t,b+t,-t$ all positive, the $\tilde{t}$ and $s'$ integrals
lead to a product of two hypergeometric functions
\begin{align}
\mM & =\f{\G(\D_{b}\pm ik_{1}\pm ik_{2},\D_{b}\pm ik_{3}\pm ik_{4},\D_{a}\pm ik_{2}\pm ik_{3})}{\G(2\D_{a},2\D_{b})\G(2\D_{b})^{2}}\int_{0}^{1}dxx^{\D_{a}+ik_{1}+ik_{4}-1}\nonumber \\
 & \times(1-x)^{2\D_{b}-1}{}_{2}F_{1}(\D_{b}+ik_{1}\pm ik_{2};2\D_{b};1-x){}_{2}F_{1}(\D_{b}+ik_{4}\pm ik_{3};2\D_{b};1-x)\nonumber \\
 & =\G_{k_{1}k_{4}}^{\D_{a}}\G_{k_{2}k_{3}}^{\D_{a}}\G_{k_{3}k_{4}}^{\D_{b}}\G_{k_{1}k_{2}}^{\D_{b}}\left[\f{\G(-2ik_{4})}{\G(\D_{a}+\D_{b}+ik_{4}\pm ik_{2},\D_{a}\pm ik_{1}-ik_{4},\D_{b}\pm ik_{3}-ik_{4})}\right.\nonumber \\
 & \left.\times{}_{4}F_{3}(\D_{a}\pm ik_{1}+ik_{4},\D_{b}\pm ik_{3}+ik_{4};\D_{a}+\D_{b}+ik_{4}\pm ik_{2},1+2ik_{4};1)+(k_{4}\leftrightarrow-k_{4})\right]
\end{align}
where we used the equation 12 on page 331 of \cite{prudnikov1990more}. Comparing with
the definition of the $sl(2)$ 6j symbol in \cite{Jafferis:2022wez} ((A.17) and (A.18)),
we find that
\begin{equation}
\mM=\sqrt{\G_{k_{1}k_{4}}^{\D_{a}}\G_{k_{2}k_{3}}^{\D_{a}}\G_{k_{3}k_{4}}^{\D_{b}}\G_{k_{1}k_{2}}^{\D_{b}}}\begin{Bmatrix}\D_{a} & k_{3} & k_{2}\\
\D_{b} & k_{1} & k_{4}
\end{Bmatrix}
\end{equation}
which matches with (\ref{eq:app1}) due to the symmetry of the $sl(2)$
6j symbol
\begin{equation}
\begin{Bmatrix}\D_{a} & k_{3} & k_{2}\\
\D_{b} & k_{1} & k_{4}
\end{Bmatrix}=\begin{Bmatrix}\D_{a} & k_{1} & k_{4}\\
\D_{b} & k_{3} & k_{2}
\end{Bmatrix}\label{eq:app23}
\end{equation}

One of the advantage of the geometric integral formula (\ref{eq:app9})
for the $sl(2)$ 6j symbol is that one can easily see the permutation
symmetries besides (\ref{eq:app23}). First, (\ref{eq:app9}) is symmetric
under exchange of $(k_{1},k_{2})\leftrightarrow(k_{3},k_{4})$ by
simultaneous exchange of $(\l_{1},\l_{2})\leftrightarrow(\l_{3},\l_{4})$;
second, (\ref{eq:app9}) is symmetric under exchange of $(k_{1},k_{2})\leftrightarrow(k_{4},k_{3})$
by simultaneous exchange of $(\l_{1},\l_{2})\leftrightarrow(\l_{4},\l_{3})$.
Divided by $\sqrt{\G_{k_{1}k_{4}}^{\D_{a}}\G_{k_{2}k_{3}}^{\D_{a}}}$,
(\ref{eq:app9}) is symmetric under exchange of $(\D_{a},k_{2})\leftrightarrow(\D_{b},k_{4})$.
This can be seen as follows. We can first exchange $\l_{2}\leftrightarrow\l_{4}$,
then define $\l_{5}$ as (\ref{eq:app7}) to replace $\l$, which
makes the power of $\D_{a}$ and $\D_{b}$ part switch back. Using
the explicit formula \eqref{eq:2.22} for $I(\l_{1},\l_{2},\l_{3})$, we can
easily show that
\begin{equation}
I(\l_{1},\l_{2},\l)I(\l_{3},\l_{4},\l)=I(\l_{1},\l_{4},\l_{5})I(\l_{2},\l_{3},\l_{5})\label{eq:app22}
\end{equation}
which justifies this symmetry. Indeed, the geometric meaning of (\ref{eq:app22})
is obvious: we can split a rectangle into two triangles in two equivalent
ways in JT gravity. 

\section{Geometric integral in the semiclassical limit} \label{app:2}

Another advantage of the geometric integral formula is that we can
easily see the OTOC (the third diagram in Figure \ref{fig:phi2}) is exponentially
suppressed relative to TOC (the first two diagrams in Figure \ref{fig:phi2}) in
the semiclassical limit. Let us take $\D_{a}=\D_{b}=\D$ in (\ref{eq:app8}),
which corresponds to the third diagram in Figure \ref{fig:phi2} with $\b_{1}=\b_{L},\b_{2}=\b_{4}=\b_{0}$
and $\b_{3}=\b_{R}$. For the first diagram in Figure \ref{fig:phi2}, in terms
of the geometric integral we have
\begin{equation}
D_{1}=\int(4e^{-\l_{1}})^{\D}(4e^{-\l_{3}})^{\D}I(\l_{1},\l_{2},\l)I(\l_{3},\l_{4},\l)\prod_{j=1,2,3,4}\Psi_{\b_{j}}(\l_{j})
\end{equation}
Similarly, for the second diagram in Figure \ref{fig:phi2}, we have
\begin{equation}
D_{2}=\int(4e^{-\l_{2}})^{\D}(4e^{-\l_{4}})^{\D}I(\l_{1},\l_{2},\l)I(\l_{3},\l_{4},\l)\prod_{j=1,2,3,4}\Psi_{\b_{j}}(\l_{j})
\end{equation}
Note that (\ref{eq:app8}) and $D_{1,2}$ share the same $I$ and
$\Psi_{\b}$ factors. By the geometric relation (\ref{eq:app7}),
we have
\begin{equation}
e^{-\l-\l_{5}}=(e^{(\l_{1}+\l_{3})/2}+e^{(\l_{2}+\l_{4})/2})^{-2}<\min(e^{-\l_{1}-\l_{3}},e^{-\l_{2}-\l_{4}})
\end{equation}
In the semiclassical limit, we take $\D\sim O(\phi_{b})\ra\infty$,
which implies 
\begin{equation}
(e^{-\l-\l_{5}})^{\D}/\min(e^{-\l_{1}-\l_{3}},e^{-\l_{2}-\l_{4}})^{\D}\sim e^{-O(\phi_{b})}\ra0
\end{equation}
In order to show (\ref{eq:app8}) is exponentially suppressed relative
to $D_{1,2}$, it suffices to show that both $I$ and $\Psi_{\b}$
are positive in the semiclassical limit. By \eqref{eq:2.22}, $I$ is manifestly
positive. The wavefunction $\Psi_{\b}(\l)$ is positive in the semiclassical
limit \cite{Harlow:2018tqv} because we can use integral representation \eqref{eq:7} and
saddle approximation
\begin{align}
\Psi_{\b}(\l) & =\int dk\r(k)e^{-\b k^{2}/(4\phi_{b})}K_{2ik}(4e^{-\l/2})\nonumber \\
 & \sim\int dkd\a e^{2\phi_{b}(2\pi k-\b k^{2}/2-4e^{-\l/2}\cos\a-2k\a)}\nonumber \\
 & \sim\int d\a e^{2\phi_{b}(2(\pi-\a)^{2}/\b-4e^{-\l/2}\cos\a)}
\end{align}
where we rescale $k\ra2\phi_{b}k$ and $e^{-\l/2}\ra2\phi_{b}e^{-\l/2}$.
Taking derivative to $\a$, we have the saddle value $\a_{s}$ obeying
\begin{equation}
\sin\a_{s}=e^{\l/2}(\pi-\a_{s})/\b
\end{equation}
Define $x=\pi-\a_s$, and we have \cite{Harlow:2018tqv}
\be 
\Psi_\b(\l)\sim e^{4\phi_b/\b(x(\l)^2+2x \cot x(\l))},\quad \f{\sin x(\l)}{x(\l)}=\b^{-1}e^{\l/2}
\ee 
When $\b^{-1}e^{\l/2}\leq 1$, there is a real solution $x(\l)\in[0,\pi]$; when  $\b^{-1}e^{\l/2}> 1$, there are two conjugate pure imaginary solutions $x(\l)\in i\R$. Both solutions give positive wavefunction $\Psi_\b(\l)$ though it is not smooth (but continuous) at $\l=2\log \b $.

This method can be easily generalized to a Euclidean path integral
with multiple heavy operator insertions. For a generic Euclidean path
integral, whose boundary consists of $n$ segment of length $\b_{i}$,
it can be written as
\begin{equation}
\mZ_{n}=\int\prod_{j=1}^{n}d\l_{j}\Psi_{\b_{j}}(\l_{j})\mI(\l_{1},\cdots,\l_{n})
\end{equation}
where $\mI$ is the region bounded by $n$ geodesics and can be further
triangularized in terms of a product of functions $I(\l_{a},\l_{b},\l_{c})$
where $\l_{a,b,c}$ are the three sides of each triangle. See Figure
\ref{fig:triang} as an illustration. Let us label all sides of a triangularization
as $\{\l_{a}\}$ with $a=1,\cdots,2n-3$ (clearly $\{\l_{1},\cdots,\l_{n}\}$
is a subset). Then we can write
\begin{equation}
\mI(\l_{1},\cdots,\l_{n})=\int\prod_{a\geq n+1}d\l_{a}F_{\D_{i}}(\{\l_{a}\})\prod_{\{a,b,c\}}I(\l_{a},\l_{b},\l_{c})
\end{equation}
where $F_{\D_{i}}$ is a product of $(4e^{-\tilde{\l}})^{\D_{i}}$
representing the Wick contractions and each $\tilde{\l}$ is a function
of $\l_{a}$ determined by geometric relations similar to (\ref{eq:app7}).
Note that different triangularizations are equivalent by repetitively
using (\ref{eq:app22}).

\begin{figure}
\begin{centering}
\includegraphics[height=4.5cm]{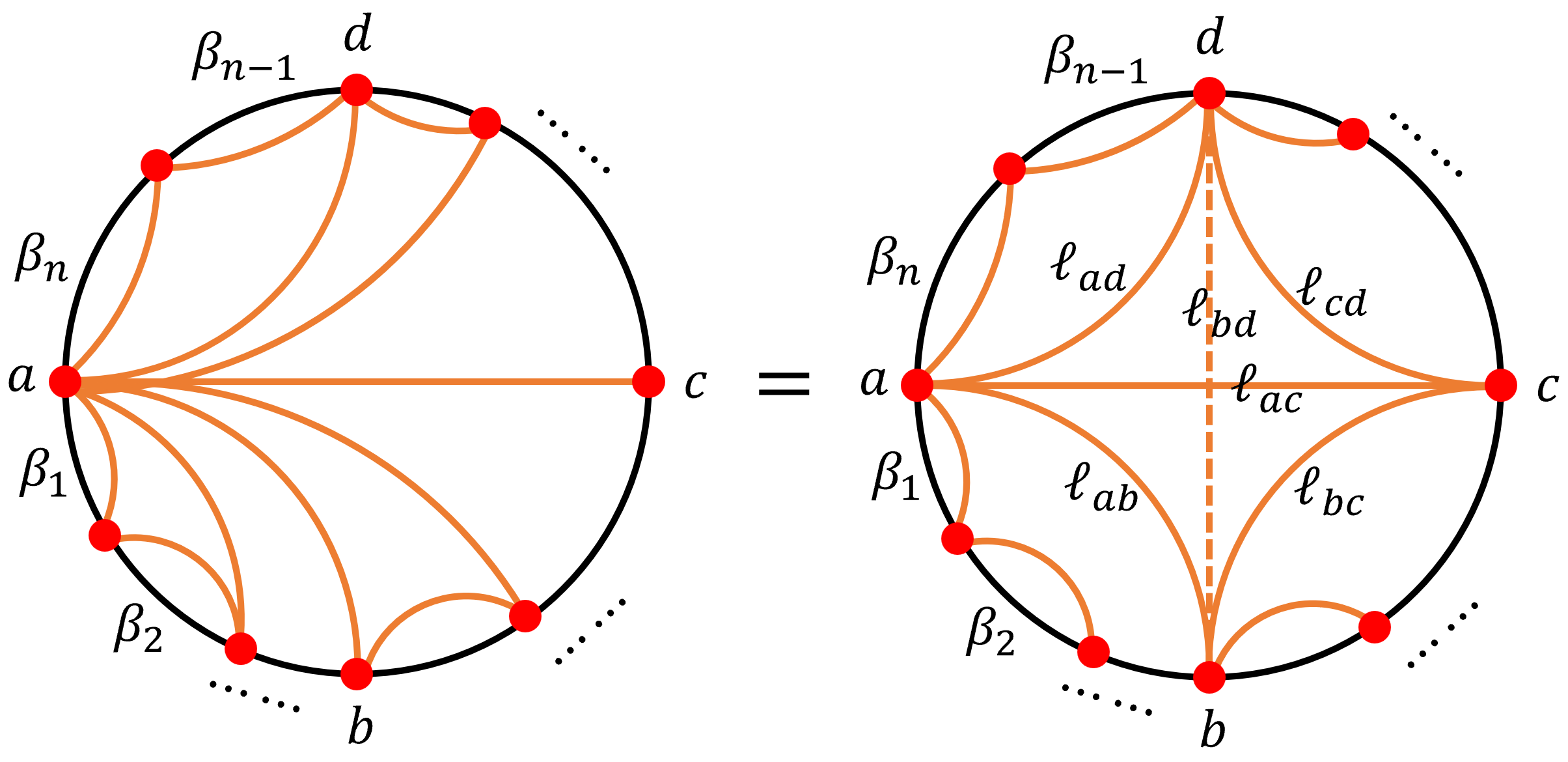}
\par\end{centering}
\caption{If we directly use the propagator $\mK_\b$ in \eqref{eq:a3}, we will have the triangularization in the left diagram, which is equivalent to any triangularization, for example, the one in the right diagram.}\label{fig:triang}
\end{figure}

By this formalism, we will be able to show the following proposition.
\begin{prop}
If there are more than one pair of $\mO_{i}$ in $\mZ_{n}$, the Wick
contraction with crossing geodesics is exponentially suppressed relative
to a non-crossing Wick contraction in the semiclassical limit.
\end{prop}
\begin{proof}
Similar to the proof for Figure \ref{fig:phi2}. Suppose there are
four $\mO_{i}$ with two crossing Wick contractions (see the right diagram in Figure \ref{fig:triang}).
Label them clockwise as $a,b,c,d$ respectively. The the contribution
from these two crossing Wick contractions are $(4e^{-\l_{ac}})^{\D_{i}}(4e^{-\l_{bd}})^{\D_{i}}$.
On the other hand, the geometric relation (\ref{eq:app7}) holds
\begin{equation}
e^{(\l_{ac}+\l_{bd})/2}=e^{(\l_{ab}+\l_{cd})/2}+e^{(\l_{ad}+\l_{bc})/2}
\end{equation}
which means that this diagram is exponentially suppressed relative
to the one with $ab$ and $cd$ contracted (and also the one with
$ad$ and $bc$ contracted). The later does not have crossing for
these four $\mO_{i}$'s. Applying this argument to all pairs of crossing
implies that the leading order contribution to $\mZ_{n}$ has no crossing
Wick contractions between any two pairs of identical $\mO_{i}$'s.
\end{proof}

\section{Semiclassical limit of $\bra{\Phi_2}$} \label{app:3}
Let us define the following function
\begin{align}     
F(\b_0,\b_2;\b_1)=&\int dk_1dk_2dk_3d\l_1d\l_2e^{-(\b_0k_0^2+\b_1 k_1^2+\b_2 k_2^2)/(4\phi_b)}\r(k_0)\r(k_1)\r(k_2)(4e^{-\l_1})^\D(4e^{-\l_2})^\D\nn\\
& \times2K_{2ik_0}(4e^{-\l_1/2})2K_{2ik_1}(4e^{-\l_1/2})2K_{2ik_1}(4e^{-\l_2/2})2K_{2ik_2}(4e^{-\l_2/2}) \label{c1}
\end{align}
where the first two arguments are symmetric. It is easy to see that the first diagram in Figure \ref{fig:phi2} is $F_1=F(\b_L,\b_R;2\b_0)$ and the second diagram is $F_2=F(\b_0,\b_0;\b_L+\b_R)$. As we show in Appendix \ref{app:2}, we only need to compare these two to determine the leading contribution in the saddle approximation. 

As $F_2$ depends on $\b_{L,R}$ only through $\b_L+\b_R$, we can define $\b_{\pm}=(\b_L\pm\b_R)/2$ and first study the $\b_-$ dependence of $F_1$ while keeping $F_2$ invariant. Taking $\b_-$ derivative to $F_1$, we have $\del_{\b_-}F_1=\avg{k_2^2-k_0^2}$, where the vev means evaluating in the integral \eqref{c1}. As $F_1$ takes value at the saddle, the derivative of $\del_{\b_-}$ acts on two parts: the explicit $\b_-$ dependence in the exponent, which gives $k_2^2-k_0^2$ at the saddle value, and the implicit $\b_-$ dependence through the saddle value of other variables. The latter vanishes at the saddle because it is proportional to the saddle equation. It boils down to determining which saddle value of $k_0$ and $k_2$ is larger. By symmetry, $F_1$ is an even function of $\b_-$ and we only need to compare $k_{0,2}$ for $\b_->0$. 

Solving the saddle equations \eqref{eq:11-1} and \eqref{eq:12-1} for $F_1$, we have 
\begin{align} 
\cos k_1 \b_0&=\f{k_1^2\sin \f{k_0\b_L}2 \sin \f{k_2\b_R}2-\left( \mu\sin \f{k_0\b_L}2+k_0\cos \f{k_0\b_L}2\right)\left( \mu \sin \f{k_2\b_R}2+k_2\cos \f{k_2\b_R}2\right)}{k_0 k_2}  \label{c2}\\
k_1^2+\mu^2&=k_2^2-2\mu k_2 \cot \f{k_2 \b_R}2=k_0^2-2\mu k_0 \cot \f{k_0 \b_L}2 \label{c3}
\end{align}
Note that \eqref{c3} is symmetry under $(k_0,\b_L)\leftrightarrow(k_2,\b_R)$. Define $f(k,\b)=k^2-2\mu k \cot \f{k \b}2$, and we have
\be 
\del_\b f(k,\b)=\f{\mu^2}{\sin^2 (k\b/2)}>0,\quad \del_k f(k,\b)=2k+\f{\mu(\b k-\sin \b k)}{\sin^2 (k\b/2)}>0
\ee 
For any value of $k_1$, \eqref{c3} implies that $k_0<k_2$ and $\del_{\b_-}F_1>0$ for $\b_->0$. Therefore, $F_1$ has a unique minimum at $\b_-=0$ and is monotonically increasing for $\b_->0$. On the other hand, if we take the maximal value $\b_-\ra\b_+$, namely $\b_R\ra 0$, we must have $k_2\sim (2\mu^{1/2}+\a \b_R+\cdots)/\b_R^{1/2}\ra\infty$, for which \eqref{c2} and \eqref{c3} in leading order become
\begin{align} 
&k_0 \cos\f{k_0\b_L}2+\mu \sin \f{k_0 \b_L}{2}+ k_0 \cos k_1 \b_0=0\\
k_1^2 &+\mu^2=4\mu^2/3+4\a \mu^{1/2}=k_0^2-2\mu k_0 \cot \f{k_0 \b_L}2
\end{align}
which gives finite $k_{0,1}$ and $\a$. Therefore $F_1$ is unbounded from above and scales as $F_1\sim \log (\b_+-\b_-)\ra\infty$ when $\b_-\ra \b_+$.

\begin{figure}
\begin{centering}
\subfloat[]{\begin{centering}
\includegraphics[height=4cm]{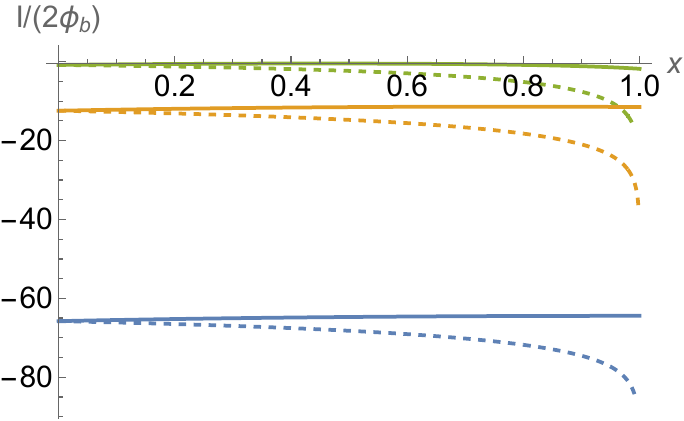}\label{fig:onshellI}
\par\end{centering}

}\subfloat[]{\begin{centering}
\includegraphics[height=4cm]{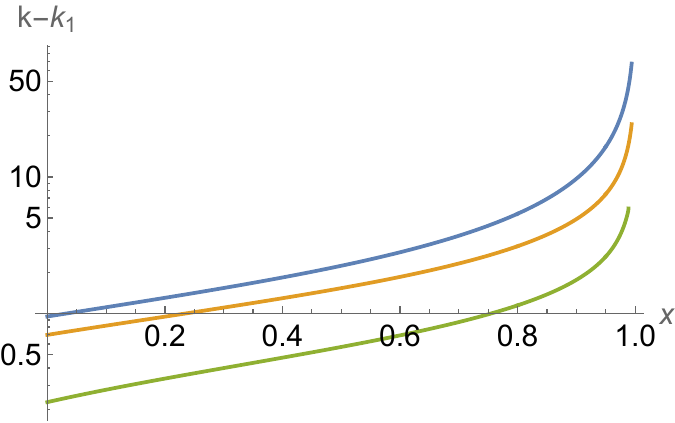}\label{fig:kk1}
\par\end{centering}

}
\par\end{centering} 
\caption{(a) Compare the on-shell action for $F_1$ (solid curves) and $F_2$ (dashed curves). (b) The on-shell $k-k_1$ of $F_2$. In both plots, the horizontal axis is $x=(\b_+-\b_0)/(\b_++\b_0)$ and blue, yellow and green curves are for $(\b_++\b_0)/2=0.1,1,10$ respectively.} 
\end{figure}

Then we will compare the minimum of $F_1=F(\b_+,\b_+;2\b_0)$ at $\b_L=\b_R$ with $F_2$. Since they are now in the same form, we can just study $F_1$. By symmetry $\b_L=\b_R=\b_+$, we have $k_0=k_2=k$ at the saddle and the saddle equations \eqref{c2} and \eqref{c3} become
\begin{align}
k_1^2 \sin^2 \f{k\b_+}2&=k^2 \cos k_1 \b_0+\left(\mu\sin \f{k\b_+}2+k\cos \f{k\b_+}2\right)^2 \label{c6}\\
k_1&=\sqrt{k^2-\mu^2 -2\mu k \cot \f{k \b_+}2} \label{c7}
\end{align}
The on-shell action \eqref{eq:2.40} is 
\be 
I/(2\phi_b)=-k^2\b_+-k_1^2\b_0+2\mu \log\left(\f{\mu^2 \sin\f{k \b_+}2 \cos \f{k_1 \b_0}{2}}{k k_1}\right)
\ee
Note that there are infinitely many solutions to \eqref{c7} due to the oscillation nature of the trigonometric functions. We only consider the solution in the range of $k\b_+,~2k_1\b_0\in(0,2\pi)$ because it smoothly interpolates to $k=k_1\ra\pi/(\b_0+\b_+)$ in the $\mu\ra 0$ limit. We can rescale $k\ra \mu k,k_1 \ra \mu k_1$, $\b_0\ra\b_0/\mu,\b_+\ra \b_+/\mu$, which is equivalent to taking $\mu=1$ in \eqref{c6} and \eqref{c7}. Since we do not have any analytic solution and the on-shell action turns out not always monotonic, we just numerically solve the saddle equation and compare the on-shell actions of $F_1$ and $F_2$ for $(\b_++\b_0)/2=0.1,1,10$ with $\b_+>\b_0$. The result in Figure \ref{fig:onshellI} shows that in all cases, the on-shell action of $F_2$ is smaller than that of $F_1$ and the second diagram in Figure \ref{fig:phi2} is preferred. This can be inferred from the small $(\b_++\b_0)/2$ case, which corresponds to small $\mu$ before rescaling. In this case, we can ignore the backreaction of heavy operators, and the Wick contraction in the second diagram of Figure \ref{fig:phi2} has a shorter geodesic distance on the EAdS$_2$ disk, which is preferred due to the $(4e^{-\l})^\mu$ suppression. 

Let us summaries the phases of $\bra{\Phi_2}$ in the semiclassical limit. If $\b_L+\b_R<2\b_0$, the first diagram of Figure \ref{fig:phi2} dominates; if $\b_L+\b_R>2\b_0$, there is a critical value $\b_*$ for $|\b_L-\b_R|$ such that the second diagram dominates when $|\b_L-\b_R|<\b_*$ and the first diagram dominates when $|\b_L-\b_R|>\b_*$. In the scenario where $F_2$ dominates, we find that the on-shell $k-k_1>0$ in Figure \ref{fig:kk1}, which means that the center disk has the global minimal dilaton point. It is important that this global minimum is located in the central row of disks (in the case where we have identical heavy operators and self-contraction dominates) because otherwise, the modular flow would not be the boost transformation in the Lorentzian spacetime. However, to prove this in general, we need additional effort. 

\bibliographystyle{JHEP.bst}
\bibliography{main}

\end{document}